\documentclass[runningheads,a4paper]{llncs}

\def\BibTeX{{\rm B\kern-.05em{\sc i\kern-.025em b}\kern-.08emT\kern-.1667em\lower.7ex\hbox{E}\kern-.125emX}}

\usepackage{nicefrac}
\usepackage{siunitx}
\usepackage{array,framed}
\usepackage{booktabs}
\usepackage{
  color,
  float,
  epsfig,
  wrapfig,
  graphics,
  graphicx,
  subcaption
}
\usepackage{textcomp}
\usepackage{setspace}
\usepackage{amsfonts}
\usepackage{latexsym,fancyhdr,url}
\usepackage{enumerate}
\usepackage{multicol}
\usepackage[utf8]{inputenc}
\usepackage{algorithmicx, algorithm, algpseudocode}
\usepackage{caption}
\usepackage{graphics}
\usepackage{xparse}
\usepackage{xspace}
\usepackage{multirow}
\usepackage{csvsimple}
\usepackage{balance}
\usepackage{graphicx}
\usepackage{amsmath}
\usepackage{cuted}
\usepackage{thmtools,thm-restate}
\usepackage[bookmarksdepth=3,hypertexnames=false]{hyperref}
\usepackage{orcidlink}
\renewcommand{\orcidID}[1]{\,\orcidlink{#1}}
\usepackage{cleveref}
\usepackage{bookmark}

\usepackage{datetime2}
\usepackage{lipsum}
\usepackage{placeins}
\usepackage{float}
\usepackage{paralist}
\usepackage{stmaryrd}
\usepackage{tabularx}
\usepackage{xifthen}
\usepackage[normalem]{ulem}
\newfloat{module}{htbp}{loa}
\floatname{module}{Module}
\crefname{module}{Module}{Modules}
\Crefname{module}{Module}{Modules}

\usepackage{etoolbox}

\newcommand{\cadence}{\textsc{Cadence}\xspace}    %
\newcommand{\chorus}{\textsc{Chorus}\xspace}      %
\newcommand{\conductor}{\textsc{Conductor}\xspace} %

\algdef{SE}[UPON]{Upon}{EndUpon}[1]{\textbf{upon} #1}{}

\algtext*{EndFor}
\algtext*{EndIf}
\algtext*{EndWhile}
\algtext*{EndProcedure}
\algtext*{EndFunction}
\algtext*{EndUpon}

\newcommand{\metablock}{meta-block\xspace}
\newcommand{\metablocks}{meta-blocks\xspace}
\newcommand{\Metablock}{Meta-block\xspace}

\newcommand{\pid}{j}                     %
\newcommand{\modcall}[1]{\mbox{\hyperref[#1]{Alg~\ref*{#1}}}}

\newcommand{\mroot}{\rho}                %
\newcommand{\cdata}{d}                   %

\newcommand{\false}{\mathsf{false}}

\newcommand{\FastQC}{\mathsf{FastQC}}

\newcommand{\chunk}{c}

\newcommand{\proc}{p}              %

\algnewcommand{\BlueComment}[1]{\textcolor{blue}{\hfill\(\triangleright\) #1}}

\newif\ifformalpart
\formalparttrue   %

\usepackage{
  tikz,
  pgfplots,
  pgfplotstable
}
\usetikzlibrary{
  shapes.geometric,
  arrows,
  external,
  pgfplots.groupplots,
  matrix
}
\pgfplotsset{compat=1.9}

\usepackage[a4paper,margin=2.5cm]{geometry}
\usepackage{mathtools}

\DeclareMathAlphabet{\mathcal}{OMS}{cmsy}{m}{n}
\DeclareGraphicsExtensions{%
    .png,.PNG,%
    .pdf,.PDF,%
    .jpg,.mps,.jpeg,.jbig2,.jb2,.JPG,.JPEG,.JBIG2,.JB2}

\usepackage[many]{tcolorbox}

\makeatletter
\renewcommand\tableofcontents{\section*{\contentsname}%
 \def\authcount##1{\setcounter{auco}{##1}\setcounter{@auth}{1}}
 \def\lastand{\ifnum\value{auco}=2\relax
                 \unskip{} \andname\
              \else
                 \unskip \lastandname\
              \fi}%
 \def\and{\stepcounter{@auth}\relax
          \ifnum\value{@auth}=\value{auco}%
             \lastand
          \else
             \unskip,
          \fi}%
 \@starttoc{toc}\if@restonecol\twocolumn\fi}
\makeatother

\spnewtheorem{claim}{Claim}{\bfseries}{\itshape}
\spnewtheorem{subclaim}{Subclaim}{\bfseries}{\itshape}
\crefname{claim}{Claim}{Claims}
\Crefname{claim}{Claim}{Claims}
\crefname{subclaim}{Subclaim}{Subclaims}
\Crefname{subclaim}{Subclaim}{Subclaims}
\newif\ifsubqedplaced
\newcommand{\subqedhere}{\hfill$\diamond$\global\subqedplacedtrue}
\newenvironment{subproof}[1][Proof]{\global\subqedplacedfalse\noindent\textit{#1.}\enspace}{\ifsubqedplaced\else\hfill$\diamond$\fi\par\medskip}
\newcommand{\vdec}{V_{\mathsf{dec}}}
\newcommand{\valdec}{v_{\mathsf{dec}}}
\newcommand{\vmax}{V_{\mathsf{max}}}
\newcommand{\proofstep}[1]{\par\smallskip\noindent\emph{#1}}
\algrenewcommand{\algorithmiccomment}[1]{\hfill{\color{blue}$\lhd$\ #1}}
\algrenewcommand\alglinenumber[1]{\scriptsize #1:}

\newif\iffull %
\newcommand{\fullmedskip}{\medskip}
\newcommand{\LineLabel}[1]{\label{#1}}
\newcommand{\coda}{\textsc{Coda}\xspace}

\newcommand{\quorus}{\textsc{Quantum Cadence}\xspace}
\newcommand{\qchorus}{\textsc{Quantum Chorus}\xspace} %
\newcommand{\para}[1]{\paragraph{\normalfont\bfseries #1}} %

\newcommand{\VALUE}{\mathsf{Value}}
\newcommand{\VIEW}{\mathsf{View}}
\newcommand{\MESSAGE}{\mathsf{Message}}

\newcommand{\METABLOCK}{\mathsf{MetaBlock}}
\newcommand{\PROPOSER}{\mathsf{Proposer}}
\newcommand{\ALLSAFE}{\textsc{empty}}
\newcommand{\NONE}{\ensuremath{\mathsf{none}}}
\newcommand{\PREPARE}{\textsc{prepare}}
\newcommand{\VOTEONE}{\textsc{vote-1}}
\newcommand{\VOTETWO}{\textsc{vote-2}}
\newcommand{\ECHO}{\textsc{echo}}

\newcommand{\propi}{\mathit{input}_i}
\newcommand{\viewi}{\mathit{view}_i}
\newcommand{\guardi}{\mathit{guard}_i}

\newcommand{\locki}{\mathit{lock}_i}
\newcommand{\voteone}{\mathit{vote1}_i}
\newcommand{\votetwo}{\mathit{vote2}_i}

\newcommand{\echosone}{\mathit{sent}_i}

\newcommand{\leader}{\mathsf{leader}}

\newcommand{\justified}{\mathsf{justified}}

\begin{document}

\def\thetitle{\quorus: Post-Quantum Consensus with Multiple Concurrent Proposers}
\title{\thetitle}
\author{Mohammad Mussadiq Jalalzai\orcidID{0000-0003-0183-4536} \and Tobias Klenze\orcidID{0000-0002-6505-3942} \and Jovan Komatovic\orcidID{0009-0006-9714-4079} \thanks{\textbf{Author order:} alphabetical by last name.}}
\institute{Category Labs}
\maketitle

\begin{abstract}
The consensus protocols of most high-performance blockchains use quorum certificates, which bundle the signatures of a quorum of validators.
Signature schemes such as BLS can aggregate the signatures of a whole quorum into a single short signature, which keeps these certificates small.
A large-scale quantum computer would break these schemes, and no post-quantum scheme with comparable aggregate size and aggregation time is known.
Without aggregation, a certificate holds one large post-quantum signature per member of the quorum, so it becomes infeasible for each validator to broadcast certificates, as modern consensus protocols often require.

We present \quorus, a post-quantum consensus protocol with multiple concurrent proposers that avoids certificates altogether: every validator counts quorums among the messages it has received itself.
Its safety and liveness do not even depend on signatures, although we add signatures to obtain accountability.
\quorus builds on \cadence and replaces its slot consensus, \chorus, with \qchorus, which keeps the optimal fast-path latency of \chorus: dissemination followed by two rounds of votes.
When its fast path fails, \qchorus runs \coda, our signature-free Byzantine agreement for multiple proposers.
We also outline how the \conductor, which schedules the slots of \cadence, can be adapted to \quorus.
In simulations with 200 validators, \qchorus has a latency comparable to that of \chorus and scales better with the number of proposers.
\end{abstract}

\section{Introduction}\label{sec:intro}

Large-scale quantum computers would break the signature schemes that blockchains use today.
Several parts of a blockchain rely on these signatures, from the transactions that users sign to the votes that validators exchange in the consensus protocol.
This paper is about consensus: we present \quorus, a post-quantum version of \cadence~\cite{cadence}, a consensus protocol with multiple concurrent proposers.

\cadence divides time into slots and finalizes one block per slot, which merges the proposals that several proposers disseminate for that slot.
Each slot is finalized by its own instance of a single-shot consensus protocol, which \cadence calls the \emph{slot consensus} and instantiates with \chorus.
Because blocks do not build on their predecessors, these instances run independently, and many slots can be in progress at once.
\cadence calls this \emph{extreme pipelining}: no slot waits for the block of an earlier slot, so blocks can follow one another at intervals shorter than the network delay.
A second protocol, the \conductor, schedules the slots, at a regular interval in normal operation and more slowly under asynchrony, so that the number of open slots stays bounded.
In the common case, \chorus finalizes a slot on its fast path in three communication rounds, and a fallback path takes over when the fast path cannot finish.
Besides safety and liveness, \cadence provides short-term censorship resistance: shortly after the network becomes synchronous, the block of every slot contains the proposal of every correct proposer.

Many messages of \chorus carry quorum certificates, most of them made of signatures from $2f+1$ of the $n = 3f+1$ validators, of which at most $f$ are faulty.
In \chorus, these signatures are combined into one compact aggregate signature that anyone can verify, with a pairing-based scheme such as BLS~\cite{bls01}, which a quantum computer would break.
Moving to post-quantum signatures changes two things.
Each post-quantum signature is at least several times larger than a traditional one, and the post-quantum aggregation schemes known today produce aggregates that are orders of magnitude larger than a BLS aggregate signature and take far longer to compute.
Without aggregation, a certificate is a list of $2f+1$ of these larger signatures, whose size grows with the number of validators. If every validator broadcast even one such certificate per slot, the transmission delays alone would be unacceptable at the bandwidth available to validators, and \chorus sends certificates for every proposer.

\quorus therefore keeps the structure of \cadence, one instance of the slot consensus per slot under the \conductor, but does without certificates: no validator ever has to prove a quorum to another, as each validator observes quorums in the messages it receives itself.
Simply removing the certificates from \chorus does not work, because several of its parts rely on them.
Our slot consensus, \qchorus, changes these parts.
Its fast path keeps the structure of \chorus, dissemination followed by two voting rounds, but no vote carries a certificate.
Its fallback path is new: \chorus ends its fallback path with an agreement that accepts only certified inputs, and \qchorus replaces it by \coda, a signature-free Byzantine agreement based on Forget-IT~\cite{forgetit}.

\qchorus goes further than removing certificates: it needs no signatures for safety and liveness, which rest only on authenticated channels and collision-resistant hashing.
We add post-quantum signatures to obtain accountability: validators sign their messages, so the correct validators can prove to anyone that a validator misbehaved.
Proofs of finality for parties outside consensus, such as light clients and bridges, are outside the scope of this paper; certificates of ordinary post-quantum signatures can serve them, since such a certificate need not be broadcast by every validator.

The \conductor needs changes too.
\cadence proves the guarantees of the \conductor with the help of certificates: a correct validator that finishes a slot forwards its certificate, so the others finish within one message delay.
The \conductor also relies on an agreement among the validators, which has to be made post-quantum as well.
We outline how the \conductor can work without certificates and with a signature-free agreement, and leave its full specification and analysis to a later version.

The fast path of \qchorus keeps the optimal latency of \chorus~\cite{cadence}: correct validators finalize in three communication rounds, dissemination included, and reach speculative finality one round earlier, even when some proposers are offline.
Even without certificates, post-quantum signatures are large enough that their number matters, and \qchorus is designed to keep the number of signatures small: in each slot, a validator signs at most two votes, each covering all proposers, and \coda uses no signatures at all. In \chorus, by contrast, a validator signs its vote on each proposer separately, so its number of signatures grows with the number of proposers.
We evaluate \qchorus against \chorus in simulation, for 200 validators with 1\,Gbps links: its latency is comparable to that of \chorus, and it scales better with the number of proposers. With 20 proposers, a validator finalizes a slot on average about 200\,ms after a proposer's broadcast, and reaches speculative finality about 50\,ms earlier.

After recalling \cadence and \chorus in \Cref{sec:background}, we turn to the core of the paper: \qchorus, which \Cref{sec:protocol} describes step by step together with its guarantees, and its fallback consensus, \coda, presented in \Cref{sec:fallback-implementation}.
We then outline how the slot consensus runs within the \conductor in \Cref{sec:conductor}, evaluate the cost of \qchorus in \Cref{sec:evaluation}, and discuss related work in \Cref{sec:related}; the proofs are in \Cref{app:slot,app:coda}.

\section{Background: \cadence and \chorus}\label{sec:background}

\cadence~\cite{cadence} runs among $n = 3f+1$ validators, at most $f$ of which are Byzantine, over a partially synchronous network: after an unknown global stabilization time (GST), every message between correct validators arrives within a known bound $\Delta$.
The validators have synchronized clocks, which \cadence uses to schedule the slots: the \conductor gives each slot a \emph{deadline}, the time at which validators vote on the proposals they have received, and each slot then finalizes one block in its own single-shot instance of \chorus.
In the rest of this section, we describe \chorus for a single slot, whose deadline we denote by $D$.
If the slot's starting time $D - \Delta$ is at least $\mathrm{GST} + \mathcal{G}$, for a \emph{grace period} $\mathcal{G}$, the \conductor ensures that every correct validator starts participating by $D - \Delta$.
For an earlier slot, a validator may learn the deadline only after it has passed, and it then starts participating and votes at once.

Each slot has a fixed set of proposers.
A proposer erasure-codes its proposal into $n$ chunks, any $f+1$ of which suffice to decode it, commits to the chunks with a Merkle root, and sends each validator its chunk, together with a Merkle proof under the root and the proposer's signature on the root.
(In \chorus, proposals are also encrypted before encoding and can be decrypted no earlier than the deadline. This provides \emph{hiding}, a property that fairness between proposers requires: no proposer can craft its proposal in reaction to the others'. We leave this step out, as this paper does not address hiding.)
A correct proposer does so when it starts participating, so in a slot whose starting time is at least $\mathrm{GST} + \mathcal{G}$, the chunks of correct proposers arrive by the deadline.
At $D$, each validator votes once per proposer: for the Merkle root of the proposer's chunk, with the chunk attached, if it holds a valid one, and for the empty value $\bot$ otherwise.
A validator that holds a valid chunk from every proposer before $D$ may cast its votes right away instead (\emph{early voting}).

\chorus agrees on one \emph{entry} per proposer: either the Merkle root of the proposer's proposal, which selects the proposal for the slot's block, or $\bot$, which excludes the proposer.
The block consists of the selected proposals, except any whose chunks turn out not to be a valid encoding. Every correct validator can recover the proposals of every block that \chorus finalizes (\emph{data availability}).
Validators exchange entries in \emph{\metablocks}, vectors with one entry per proposer, in which every entry carries a certificate that justifies it.
In the common case, the certificate is a $\FastQC$: a quorum certificate formed from $2f+1$ matching votes, all for the same root or all for $\bot$, which \chorus implements as a compact aggregate signature.
A validator that holds a $\FastQC$ for every proposer assembles the \emph{fast \metablock} and broadcasts a fast vote on it, together with the \metablock, and it finalizes the entries once it receives $2f+1$ fast votes on the same entries.
This is the fast path, three communication rounds in all: dissemination and two rounds of votes.
A validator may also speculatively finalize one round earlier, that is, treat the fast \metablock as final as soon as it assembles it; this speculation is reverted only if some validator equivocated.

If the fast path does not finish, for instance because a proposal reached only some correct validators by the deadline, the fallback path takes over.
From time $D + \Delta$ on, a validator that holds $2f+1$ votes but cannot assemble a fast \metablock announces that it casts no fast vote and casts a second round of per-proposer votes, after which the validators run a multi-valued validated Byzantine agreement (MVBA).
Each validator proposes a fast \metablock if it has assembled one or received one with a fast vote, and otherwise a \emph{fallback \metablock}.
A fallback \metablock carries $2f+1$ announcements that no fast vote was cast, and it certifies each entry in one of three ways: with a $\FastQC$, with $f+1$ matching votes of the second round, or, for $\bot$, with two different roots signed by the proposer, which prove that it equivocated.
The external validity of the MVBA admits only such certified \metablocks, so each validator can check a proposed \metablock on its own.
Once the MVBA decides, validators cast commit votes on the decided \metablock, and a validator finalizes the slot once it holds $2f+1$ of these votes.

\chorus guarantees \emph{proposal inclusion}: in a slot whose starting time is at least $\mathrm{GST} + \mathcal{G}$, the block contains the proposal of every correct proposer.
This gives \cadence its short-term censorship resistance: a transaction that a correct proposer includes in its proposal enters the block of that slot, so it is neither censored nor deferred to a later block.
\qchorus keeps this property, as well as agreement, termination, and speculative finality.

\section{Model}\label{sec:model}

We consider $n = 3f+1$ validators, at most $f$ of which are \emph{Byzantine} and may deviate arbitrarily from the protocol; the others are \emph{correct}.
Validators communicate over point-to-point authenticated channels: the recipient of a message can tell who sent it.
Reliability, however, is not assumed.
A broadcast unicasts its message to every validator and \emph{retransmits} it, again and again, until the broadcast is explicitly \emph{abandoned}; a point-to-point send retransmits to its recipient in the same way.
A message is \emph{alive} at time $\tau$ if its broadcast or point-to-point send has been initiated, but not abandoned, by time $\tau$, and \emph{forever-alive} if, once initiated, it is never abandoned.
A forever-alive message of a correct validator is eventually received by every correct recipient, including after the recipient recovers from a crash or starts participating.
We use this model throughout the paper. It is similar to the model of Forget-IT~\cite{forgetit}, which, however, communicates by broadcasts only.
Bounding memory is meaningful only with unreliable channels: with reliable channels, undelivered messages must be buffered until delivery, as \cadence observes while retaining reliable channels~\cite{cadence}. With signatures and certificates, bounded memory is well understood; without them, it is not obvious, which motivates our model and the persistent-memory bound for \coda (\Cref{sec:fallback-implementation}).

\para{Crashes and persistent memory.}
Beyond corrupting validators, the adversary may temporarily \emph{crash} a correct validator, which eventually \emph{recovers}.
A crash strikes only between handler executions. It leaves the validator's persistent memory intact, timers included, and wipes the rest of its state, notably the received messages that it has not stored in persistent memory. Upon recovering, the validator resumes execution from its persistent memory.
We assume that a crash does not abandon the validator's ongoing broadcasts: their retransmission resumes upon recovery.
This assumption costs no generality, as every message a validator broadcasts is recorded in its persistent memory, so that a recovery procedure could simply re-initiate their broadcasts.

\para{Partial synchrony and clocks.}
We adopt the standard partially synchronous model: an unknown Global Stabilization Time (GST) splits the execution into an asynchronous period and a synchronous one.
After GST, message delays are bounded by the known $\Delta$: for every time $\tau$, a message of a correct validator that remains alive throughout $[\tau, \max(\tau, \mathrm{GST}) + \Delta]$ is received by every correct recipient during $[\max(\tau, \mathrm{GST}), \max(\tau, \mathrm{GST}) + \Delta]$.
This applies to broadcasts and point-to-point sends, including messages received and lost before $\tau$.
Since a message sent before GST may be lost, $\Delta$ also covers the wait for its next retransmission.
The adversary crashes correct validators only before GST, and every crashed validator recovers before GST, so that all correct validators are continuously online from GST onward.
Local computation is instantaneous, and each validator executes only finitely many handlers in any finite time interval.
As in \cadence, validators have synchronized clocks.

\para{Cryptography.}
We use a collision-resistant hash function.
Every validator holds a key pair of a post-quantum signature scheme.
The protocol needs signatures only for accountability, that is, to prove misbehavior, as the accountability lemmas and the speculation lemma do (\Cref{lem:qv:acc-half,lem:qv:acc-f,lem:qv:speculation}). Agreement, proposal inclusion, termination, and halting rest on the authenticated channels and on the hash function.

\para{Slots.}
Each slot of \quorus has a deadline $D$, and we call $D - \Delta$ the slot's \emph{starting time}.
A validator \emph{participates} in a slot from the time it starts running the slot consensus for that slot, a time that the \conductor determines (\Cref{sec:conductor}).
We call a slot \emph{timely} if its starting time $D - \Delta$ is at least $\mathrm{GST}$ and every correct validator starts participating in it by $D - \Delta$.
In \cadence, the \conductor ensures that every slot whose starting time is at least $\mathrm{GST} + \mathcal{G}$ is timely, where the grace period $\mathcal{G}$ is the time the \conductor needs to settle after GST (\Cref{sec:background}).
We expect the \conductor of \quorus to do the same (\Cref{sec:conductor}).

\section{\qchorus: The Slot Consensus of \quorus}
\label{sec:protocol}

\qchorus takes the place of \chorus: in every slot, it finalizes one block.
Like \chorus, it agrees on one entry per proposer, a Merkle root or a value that excludes the proposer, but its \metablocks carry no certificates: a \metablock is just a vector with one entry per proposer.
Validators agree on a \metablock, and finalizing a \metablock finalizes the slot's block, which consists of the proposals behind its roots. For every finalized block, these proposals can be recovered.
Here, we first describe the protocol and present its pseudocode.
Then, we state its guarantees.

\subsection{Protocol Description}\label{sec:qv:description}

\para{Setting.}
\qchorus runs in the model of \Cref{sec:model}.
Some of the validators are \emph{proposers}, each of which disseminates one proposal in the slot we describe.
Proposals are encoded with an erasure code, any $f+1$ of whose $n$ symbols suffice to decode, and committed to with Merkle trees over the hash function.

\para{Conventions.}
Every rule fires as soon as its condition holds, and a handler runs to completion before the next one starts.
Time guards remain enabled after their deadline.
\Cref{alg:qv1} never abandons a broadcast or a point-to-point send, also not when the validator halts, so all its messages are forever-alive.\footnote{In \quorus, we expect the \conductor to keep the number of open slots bounded, as in \cadence, and to end the broadcasts of slots that no correct validator needs any more (\Cref{sec:conductor}).}
A validator keeps no message of the slot that arrives before it starts participating: retransmission delivers the message again once it participates (\Cref{sec:model}).
Every message, and every tuple that a validator signs, also contains the slot number, which we omit from the notation.
Throughout the paper, ``$f+1$ messages'' means messages from $f+1$ different validators, and likewise for other counts.

\para{The steps.}
We now describe the protocol step by step.
\Cref{fig:protocol-overview} gives a bird's-eye view, simplified for a single proposer, and \Cref{alg:qv1} gives the pseudocode for a validator $\proc_i$.
Steps~1 to~3 (dissemination, a vote at the deadline, and a path vote) serve both paths. The fast path ends after them: a validator finalizes once, for every proposer, $2f+1$ path votes (defined in Step~3) carry matching fast entries. The fallback path continues with Steps~4 to~6 and ends in \coda.

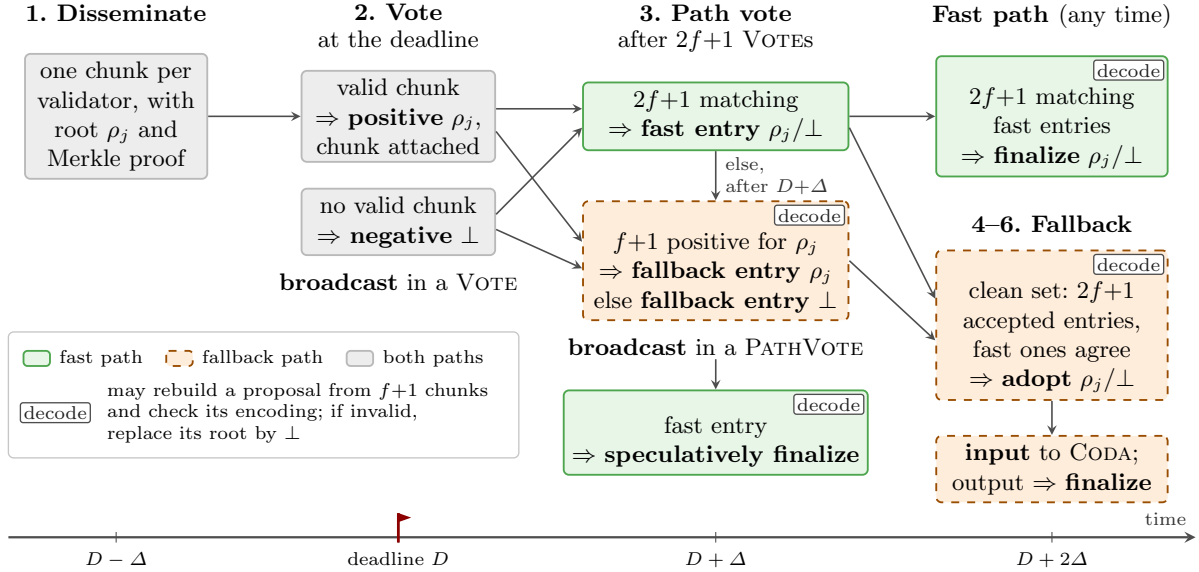
\begin{figure}[tp]
\centering
\begin{tikzpicture}[
  font=\footnotesize,
  >=stealth,
  box/.style={rounded corners=2pt, line width=.7pt, align=center, inner sep=3pt},
  common/.style={box, draw=gray!55, fill=gray!14},
  fast/.style={box, draw=green!50!black!70, fill=green!60!black!9},
  fallback/.style={box, draw=orange!60!black, fill=orange!14, dashed},
  title/.style={font=\footnotesize\bfseries, inner sep=0pt},
  whole/.style={align=center, inner sep=0pt},
  decode/.style={draw=gray!60!black, fill=white, text=gray!25!black,
    rounded corners=1pt, inner sep=1.3pt, font=\scriptsize, solid, line width=.5pt},
  arr/.style={->, semithick, gray!60!black},
]
\def\gap{0.22}     %
\def\ytop{6}       %
\def\xa{1.475}\def\bwa{2.41}
\def\xb{5.207}\def\bwb{2.56}
\def\xc{9.404}\def\bwc{3.51}
\def\xd{13.85}\def\bwd{3.06}
\newcommand{\dectop}{\rule{0pt}{15pt}}

\node[title, anchor=north] (t1) at (\xa,\ytop) {1. Disseminate};
\node[common, minimum width=\bwa cm, anchor=north] (s1) at ([yshift=-0.42cm]t1.south)
  {one chunk per\\validator, with\\root $\mroot_\pid$ and\\Merkle proof};

\node[title, anchor=north] (t2) at (\xb,\ytop) {2. Vote};
\node[whole, anchor=north] (t2s) at ([yshift=-0.12cm]t2.south) {at the deadline};
\node[common, minimum width=\bwb cm] (s2) at (\xb,0 |- s1)
  {valid chunk\\$\Rightarrow$ \textbf{positive} $\mroot_\pid$,\\chunk attached};
\node[common, minimum width=\bwb cm, anchor=north] (s2n) at ([yshift=-0.33cm]s2.south)
  {no valid chunk\\$\Rightarrow$ \textbf{negative} $\bot$};
\node[whole, anchor=north] (w2) at ([yshift=-\gap cm]s2n.south) {\textbf{broadcast} in a \textsc{Vote}};
\coordinate (p2bot) at (w2.south);

\node[title, anchor=north] (t3) at (\xc,\ytop) {3. Path vote};
\node[whole, anchor=north] (t3s) at ([yshift=-0.12cm]t3.south) {after $2f{+}1$ \textsc{Vote}s};
\node[fast, minimum width=\bwc cm] (s3f) at (\xc,0 |- s1) {$2f{+}1$ matching\\$\Rightarrow$ \textbf{fast entry} $\mroot_\pid/\bot$};
\node[fallback, minimum width=\bwc cm, anchor=north] (s3b) at ([yshift=-0.65cm]s3f.south)
  {\dectop$f{+}1$ positive for $\mroot_\pid$\\$\Rightarrow$ \textbf{fallback entry} $\mroot_\pid$\\else \textbf{fallback entry} $\bot$};
\node[whole, anchor=north] (w3) at ([yshift=-\gap cm]s3b.south) {\textbf{broadcast} in a \textsc{PathVote}};
\node[fast, minimum width=\bwc cm, anchor=north] (spec) at ([yshift=-0.45cm]w3.south)
  {\rule{0pt}{13pt}fast entry\\$\Rightarrow$ \textbf{speculatively finalize}};%
\coordinate (spbot) at (spec.south);

\node[title, anchor=north] (t4) at (\xd,\ytop) {Fast path \textmd{(any time)}};
\node[fast, minimum width=\bwd cm] (ff) at (\xd,0 |- s1)
  {\dectop$2f{+}1$ matching\\fast entries\\$\Rightarrow$ \textbf{finalize} $\mroot_\pid/\bot$};

\node[title, anchor=north] (t5) at ([yshift=-0.5cm]ff.south) {4--6. Fallback};
\node[fallback, minimum width=\bwd cm, anchor=north] (ad) at ([yshift=-\gap cm]t5.south)
  {\dectop clean set: $2f{+}1$\\accepted entries,\\fast ones agree\\$\Rightarrow$ \textbf{adopt} $\mroot_\pid/\bot$};
\node[fallback, minimum width=\bwd cm, anchor=north] (cd) at ([yshift=-0.45cm]ad.south)
  {\textbf{input} to \coda;\\output $\Rightarrow$ \textbf{finalize}};
\coordinate (p5bot) at (cd.south);

\foreach \b in {s3b, spec, ff, ad}
  \node[decode, anchor=north east] at ([xshift=-1.6pt, yshift=-1.6pt]\b.north east) {decode};

\draw[arr] (s1) -- (s2);
\draw[arr] ([yshift=0.1cm]s2.east) -- ([yshift=0.1cm]s3f.west);
\draw[arr] ([yshift=0.1cm]s2n.east) -- ([yshift=-0.15cm]s3f.west);
\draw[arr] ([yshift=-0.1cm]s2n.east) -- ([yshift=-0.1cm]s3b.west);
\draw[arr] ([yshift=-0.2cm]s2.east) -- ([yshift=0.25cm]s3b.west);
\draw[arr] (s3f) -- (s3b) node[midway, right, align=left, font=\scriptsize, text=gray!40!black] {else,\\after $D{+}\Delta$};
\draw[arr] (s3f) -- (ff);
\draw[arr] (s3b.east) -- ([yshift=-0.2cm]ad.west);
\draw[arr] ([yshift=-0.15cm]s3f.east) -- ([yshift=0.35cm]ad.west);
\draw[arr] ([yshift=-0.05cm]w3.south) -- (spec);
\draw[arr] (ad) -- (cd);

\begin{scope}[legend/.style={align=left, inner sep=0pt, font=\scriptsize}]
  \coordinate (lgtop) at ([yshift=-0.55cm]p2bot -| 0.05,0);
  \coordinate (lg2) at ([yshift=-0.35cm]lgtop -| 0.25,0);
  \draw[fast, rounded corners=1.5pt] ([yshift=0.1cm]lg2) rectangle ([xshift=0.35cm, yshift=-0.1cm]lg2);
  \node[legend, anchor=west] (lc1) at ([xshift=0.47cm]lg2) {fast path};
  \draw[fallback, rounded corners=1.5pt] ([xshift=0.3cm, yshift=0.1cm]lc1.east) rectangle ([xshift=0.65cm, yshift=-0.1cm]lc1.east);
  \node[legend, anchor=west] (lc2) at ([xshift=0.77cm]lc1.east) {fallback path};
  \draw[common, rounded corners=1.5pt] ([xshift=0.3cm, yshift=0.1cm]lc2.east) rectangle ([xshift=0.65cm, yshift=-0.1cm]lc2.east);
  \node[legend, anchor=west] at ([xshift=0.77cm]lc2.east) {both paths};
  \node[legend, anchor=north west] (lg3) at ([yshift=-0.32cm]lg2 -| 1.35,0)
    {may rebuild a proposal from $f{+}1$ chunks\\and check its encoding; if invalid,\\replace its root by $\bot$};
  \node[decode, anchor=west] at (lg3.west -| 0.2,0) {decode};
  \coordinate (lgbot) at ([yshift=-0.2cm]lg3.south -| 6.8,0);
  \draw[gray!45, line width=.5pt, rounded corners=2pt] (lgtop) rectangle (lgbot);
\end{scope}

\path (spbot); \pgfgetlastxy{\tmpx}{\yspec}
\path (p5bot); \pgfgetlastxy{\tmpx}{\yfb}
\path (lgbot); \pgfgetlastxy{\tmpx}{\ylg}
\pgfmathsetmacro{\yax}{min(\yspec,\yfb,\ylg)-0.45cm}
\begin{scope}[yshift=\yax pt]
\draw[->, line width=.85pt, gray!60!black] (0.05,0) -- (15.75,0);
\foreach \x/\t in {\xa/{$D-\Delta$}, \xb/{deadline $D$}, \xc/{$D+\Delta$}, \xd/{$D+2\Delta$}}
  \draw[gray!60!black] (\x,.07) -- (\x,-.07) node[below, font=\scriptsize, text=black] {\t};
\node[font=\scriptsize, text=gray!40!black, anchor=south east] at (15.75,.04) {time};
\draw[red!55!black, line width=1.05pt] (\xb,-.07) -- (\xb,.32);
\fill[red!55!black] (\xb,.32) -- (\xb+.17,.26) -- (\xb,.20) -- cycle;
\end{scope}
\end{tikzpicture}
\caption{\qchorus{}, simplified for a single proposer $\proc_\pid$.
Time flows left to right. In a timely slot (\Cref{sec:model}), a step starts at the latest at the time marked below its box, and in the common case every correct validator finalizes on the fast path by $D + 2\Delta$.
With several proposers, every vote and path vote carries one entry per proposer, and a validator speculates, finalizes on the fast path, or inputs to \coda only once the condition shown holds for every proposer; a path vote may thus hold a fast entry for one proposer and a fallback entry for another.
Crucially, no message ever carries a certificate: every quorum is observed locally, by each validator in the messages it receives.}
\label{fig:protocol-overview}
\end{figure}

\begin{algorithm}[!tp]
\caption{\qchorus{}: Pseudocode (for validator $\proc_i$)}
\label{alg:qv1}
\begin{algorithmic}[1]
\scriptsize
\renewcommand{\fullmedskip}{\vspace{1.25pt plus 0.5pt minus 0.5pt}}%

\State \textbf{Notation:}
\State \hskip2em - Entries for a proposer $\proc_\pid$ are \emph{positive}, $\langle \pid, \mroot \rangle$, or \emph{negative}, $\langle \pid, \bot \rangle$.

\fullmedskip
\State \textbf{Local variables (persistent memory):}
\State \hskip2em $\mathit{chunk}_i[\pid] \gets \mathsf{none}$ for every proposer $\proc_\pid$ \BlueComment{$\proc_i$'s validated chunk of $\proc_\pid$'s proposal, together with its root $\mroot_\pid$}
\State \hskip2em $\mathit{entry}_i[\pid] \gets \mathsf{none}$ for every proposer $\proc_\pid$ \BlueComment{$\proc_i$'s path-vote entry: $\langle \mathsf{fast}, v \rangle$ or $\langle \mathsf{fallback}, v \rangle$}
\State \hskip2em $\mathit{adopted}_i[\pid] \gets \mathsf{none}$ for every proposer $\proc_\pid$ \BlueComment{value adopted for the fallback consensus}
\State \hskip2em $\mathit{finalized}_i \gets \false$ \BlueComment{whether $\proc_i$ has finalized (set by \textbf{finalize})}
\State \hskip2em Messages sent, halt status, and completion of once-only steps
\State \hskip2em Received \textsc{Finalized} messages (at most one per validator)
\State \hskip2em $\mathit{votes}_i \gets \emptyset$ \BlueComment{received valid \textsc{Vote} messages}
\State \hskip2em $\mathit{pathVotes}_i \gets \emptyset$ \BlueComment{received \textsc{PathVote} messages}

\fullmedskip
\State \textbf{Rules:}
\State \hskip2em - Persist all received messages of this algorithm; halting does not abandon broadcasts.
\State \hskip2em \parbox[t]{\dimexpr\linewidth-4em\relax}{\settowidth{\dimen0}{- }\hangindent=\dimen0\hangafter=1 \leavevmode- A received $\langle \textsc{Chunk}, \mroot, \cdata, \pi, \sigma \rangle$ from a proposer $\proc_\pid$ is valid only if $\sigma$ is a valid signature of $\proc_\pid$ on $\langle \pid, \mroot \rangle$ and $\pi$ proves that $\cdata$ is the $i$-th symbol under $\mroot$.}
\State \hskip2em \parbox[t]{\dimexpr\linewidth-4em\relax}{\settowidth{\dimen0}{- }\hangindent=\dimen0\hangafter=1 - A received $\langle \textsc{Vote}, E, C, \sigma \rangle$ from $\proc_r$ is added to $\mathit{votes}_i$ only if $\sigma$ is a valid signature of $\proc_r$ on $E$, $E$ holds an entry $\langle \pid, \cdot \rangle$ for every proposer $\proc_\pid$, and $C$ holds, for every positive entry $\langle \pid, \mroot \rangle$, $\proc_r$'s symbol with a valid Merkle proof under $\mroot$.} \LineLabel{ln:qv:allornothing}
\State \hskip2em \parbox[t]{\dimexpr\linewidth-4em\relax}{\settowidth{\dimen0}{- }\hangindent=\dimen0\hangafter=1 - A received $\langle \textsc{PathVote}, F, \sigma \rangle$  from $\proc_r$ is added to $\mathit{pathVotes}_i$ only if $\sigma$ is a valid signature of $\proc_r$ on $F$.}
\State \hskip2em \parbox[t]{\dimexpr\linewidth-4em\relax}{\settowidth{\dimen0}{- }\hangindent=\dimen0\hangafter=1 - At most one \textsc{Vote} from each validator is added to $\mathit{votes}_i$, and at most one \textsc{PathVote} to $\mathit{pathVotes}_i$.}
\State \hskip2em \parbox[t]{\dimexpr\linewidth-4em\relax}{\settowidth{\dimen0}{- }\hangindent=\dimen0\hangafter=1 - A path-vote entry $\langle \mathsf{fallback}, \cdot \rangle$ for $\proc_\pid$ is \emph{accepted} upon receipt. An entry $\langle \mathsf{fast}, \mroot \rangle$ for $\proc_\pid$ is accepted  as soon as $\mathit{votes}_i$ holds $f + 1$ positive entries $\langle \pid, \mroot \rangle$; an entry $\langle \mathsf{fast}, \bot \rangle$ once $\mathit{votes}_i$ holds $f + 1$ negative entries for $\proc_\pid$.} \LineLabel{ln:qv:accept}

\fullmedskip
\State \textbf{Functions:}
\State \hskip2em - $\mathsf{clean}(\pid)$: \BlueComment{Step 4: a clean set for proposer $\proc_\pid$, or $\mathsf{none}$}
\State \hskip4em \parbox[t]{\dimexpr\linewidth-8em\relax}{\settowidth{\dimen0}{\textbf{return} }\hangindent=\dimen0\hangafter=1 \textbf{return} the first set $S$ obtained by $\proc_i$ of $2f + 1$ accepted entries for $\proc_\pid$, from different validators, such that all fast entries in $S$ carry the same value; $\mathsf{none}$ if no such set exists} \LineLabel{ln:qv:clean}

\fullmedskip
\State \hskip2em - $\mathsf{adopt}(S)$: \BlueComment{Step 5: the adoption rule, reading a clean set $S$}
\State \hskip4em \textbf{if} $S$ holds an entry $\langle \mathsf{fast}, \bot \rangle$: \textbf{return} $\bot$ \BlueComment{unambiguous: all fast entries in $S$ agree}
\State \hskip4em \textbf{else if} $S$ holds an entry $\langle \mathsf{fast}, \mroot \rangle$:
\State \hskip6em Decode the proposal from the $f + 1$ chunks that corroborated the entry, and re-encode it
\State \hskip6em \textbf{if} the recomputed root equals $\mroot$: \textbf{return} $\mroot$ \LineLabel{ln:qv:adoptcheck}
\State \hskip6em \textbf{else return} $\bot$ \BlueComment{invalid encoding: the proposal is excluded}
\State \hskip4em \textbf{else if} $S$ holds $f + 1$ entries $\langle \mathsf{fallback}, \mroot \rangle$ for the same $\mroot$: \textbf{return} $\mroot$
\State \hskip4em \textbf{else return} $\bot$ \LineLabel{ln:qv:adoptbot}

\fullmedskip
\State \textbf{upon starting:} \BlueComment{Step 1: dissemination at time $D - \Delta$}
\State \hskip2em \textbf{if} $p_i$ is a proposer:
\State \hskip4em Encode the proposal into $n$ erasure-coded symbols $\cdata_1, \ldots, \cdata_n$
\State \hskip4em Build a Merkle tree over $\cdata_1, \ldots, \cdata_n$ with root $\mroot_i$
\State \hskip4em $\sigma_i \gets \mathsf{Sign}_i(\langle i, \mroot_i \rangle)$ \BlueComment{signed header: binds the root to the proposer}
\State \hskip4em \textbf{for each} validator $\proc_r$: 
\State \hskip6em \textbf{send} $\langle \textsc{Chunk}, \mroot_i, \cdata_r, \pi_r, \sigma_i \rangle$ to $\proc_r$ \BlueComment{$\pi_r$: Merkle proof of $\cdata_r$} \LineLabel{ln:qv:disseminate}

\fullmedskip
\State \textbf{upon receipt of} a valid $\langle \textsc{Chunk}, \mroot, \cdata, \pi, \sigma \rangle$ \textbf{from} a proposer $\proc_\pid$ \textbf{and} $\mathit{chunk}_i[\pid] = \mathsf{none}$ \textbf{and} time has not passed $D$: \LineLabel{ln:qv:chunk}
\State \hskip2em $\mathit{chunk}_i[\pid] \gets \langle \mroot, \cdata, \pi, \sigma \rangle$

\fullmedskip
\State \textbf{upon} time reaches $D$ \textbf{or} $\mathit{chunk}_i[\pid] \neq \mathsf{none}$ for every proposer $\proc_\pid$: \BlueComment{Step 2: deadline vote at time $D$ or earlier (runs once)}
\State \hskip2em \textbf{for each} proposer $\proc_\pid$:
\State \hskip4em \textbf{if} $\mathit{chunk}_i[\pid] \neq \mathsf{none}$: $E_\pid \gets \langle \pid, \mroot_\pid \rangle$ \BlueComment{positive entry}
\State \hskip4em \textbf{else} $E_\pid \gets \langle \pid, \bot \rangle$ \BlueComment{negative entry: no valid chunk by the deadline}
\State \hskip2em $E \gets (E_\pid)_{\text{proposer } \proc_\pid}$
\State \hskip2em $C \gets \{ \langle \mroot, \cdata, \pi \rangle \text{ from } \mathit{chunk}_i[\pid] : E_\pid \text{ is positive} \}$ \BlueComment{no signed headers}
\State \hskip2em $\sigma_i \gets \mathsf{Sign}_i(E)$ \BlueComment{one signature over the whole entry vector}
\State \hskip2em \textbf{broadcast} $\langle \textsc{Vote}, E, C, \sigma_i \rangle$ \LineLabel{ln:qv:vote}

\fullmedskip
\State \textbf{upon} time reaches $D + \Delta$ \textbf{and} $|\mathit{votes}_i| \geq 2f + 1$, \textbf{or} for every proposer $\proc_\pid$, $\mathit{votes}_i$ holds $2f + 1$ entries for $\proc_\pid$ that name the same value: \BlueComment{Step 3: path vote (runs once)}
\State \hskip2em \textbf{for each} proposer $\proc_\pid$:
\State \hskip4em \textbf{if} $\mathit{votes}_i$ holds $2f + 1$ positive entries $\langle \pid, \mroot \rangle$ for the same Merkle root $\mroot$: $\mathit{entry}_i[\pid] \gets \langle \mathsf{fast}, \mroot \rangle$ \LineLabel{ln:qv:fastroot}
\State \hskip4em \textbf{else if} $\mathit{votes}_i$ holds $2f + 1$ negative entries for $\proc_\pid$: $\mathit{entry}_i[\pid] \gets \langle \mathsf{fast}, \bot \rangle$ \LineLabel{ln:qv:fastbot}
\State \hskip4em \textbf{else if} $\mathit{votes}_i$ holds $f + 1$ positive entries $\langle \pid, \mroot \rangle$ for some Merkle root $\mroot$: \LineLabel{ln:qv:decode}
\State \hskip6em Decode the proposal from the chunks attached to these entries, and re-encode it
\State \hskip6em \textbf{if} the recomputed root equals $\mroot$: $\mathit{entry}_i[\pid] \gets \langle \mathsf{fallback}, \mroot \rangle$
\State \hskip6em \textbf{else} $\mathit{entry}_i[\pid] \gets \langle \mathsf{fallback}, \bot \rangle$
\State \hskip4em \textbf{else} $\mathit{entry}_i[\pid] \gets \langle \mathsf{fallback}, \bot \rangle$
\State \hskip2em $F \gets (\mathit{entry}_i[\pid])_{\text{proposer }\proc_\pid}$
\State \hskip2em $\sigma_i \gets \mathsf{Sign}_i(F)$ \BlueComment{one signature over the whole path vote}
\State \hskip2em \textbf{broadcast} $\langle \textsc{PathVote}, F, \sigma_i \rangle$ \LineLabel{ln:qv:pathvote}
\State \hskip2em \textbf{if} every entry of $F$ is fast, with $\mathit{entry}_i[\pid] = \langle \mathsf{fast}, v_\pid \rangle$ for every proposer $\proc_\pid$, \textbf{and} $\neg \mathit{finalized}_i$:
\State \hskip4em \textbf{for each} $\proc_\pid$ whose $v_\pid$ is a root: re-encode the proposal decoded from $f + 1$ chunks under $v_\pid$\LineLabel{ln:qv:speccheck}
\State \hskip6em \textbf{if} the root differs: $v_\pid \gets \bot$
\State \hskip4em \textbf{speculatively finalize} the \metablock $(v_\pid)_{\text{proposer } \proc_\pid}$ \LineLabel{ln:qv:speculate}

\fullmedskip
\State \textbf{upon} for every proposer $\proc_\pid$, $\mathit{pathVotes}_i$ holds $2f + 1$ messages whose entry for $\proc_\pid$ is $\langle \mathsf{fast}, v_\pid \rangle$ for the same $v_\pid$, and, for every $\proc_\pid$ whose $v_\pid$ is a root, $\mathit{votes}_i$ holds $f + 1$ positive entries $\langle \pid, v_\pid \rangle$: \BlueComment{fast path (at any time)}
\State \hskip2em \textbf{for each} $\proc_\pid$ whose $v_\pid$ is a root: decode the proposal from these chunks and re-encode it; \textbf{if} the root differs: $v_\pid \gets \bot$ \LineLabel{ln:qv:fastcheck}
\State \hskip2em \textbf{finalize} the \metablock $x = {}$$\bigl((v_\pid)_{\text{proposer } \proc_\pid}\bigr)$ and \textbf{broadcast} $\langle \textsc{Finalized}, x \rangle$, unless already finalized \LineLabel{ln:qv:fastdecide}

\fullmedskip
\State \textbf{upon} time reaches $D + 2\Delta$ \textbf{and} $\mathsf{clean}(\pid) \neq \mathsf{none}$ for every proposer $\proc_\pid$: \BlueComment{Step 5: adoption (runs once)}
\State \hskip2em \textbf{for each} proposer $\proc_\pid$:
\State \hskip4em $\mathit{adopted}_i[\pid] \gets \mathsf{adopt}(\mathsf{clean}(\pid))$ \LineLabel{ln:qv:adopt}
\State \hskip2em $a \gets (\mathit{adopted}_i[\pid])_{\text{proposer } \proc_\pid}$
\State \hskip2em \textbf{input} $a$ to the fallback consensus \LineLabel{ln:qv:input}

\fullmedskip
\State \textbf{upon} the fallback consensus \textbf{outputs} $x$: \BlueComment{Step 6}
\State \hskip2em $x[\pid] \gets \bot$ for every proposer $\proc_\pid$ with $x[\pid] = \ALLSAFE$
\State \hskip2em \textbf{finalize} the \metablock $x$ and \textbf{broadcast} $\langle \textsc{Finalized}, x \rangle$, unless already finalized \BlueComment{$x[\pid] {}= \bot$ excludes $\proc_\pid$} \LineLabel{ln:qv:fallbackdecide}

\fullmedskip
\State \textbf{upon} receipt of $\langle \textsc{Finalized}, x \rangle$ from $f + 1$ validators, for the same $x$: \LineLabel{ln:qv:echo}
\State \hskip2em \textbf{finalize} the \metablock $x$ and \textbf{broadcast} $\langle \textsc{Finalized}, x \rangle$, unless already finalized
\State \textbf{upon} receipt of $\langle \textsc{Finalized}, x \rangle$ from $2f + 1$ validators, for the same $x$: \LineLabel{ln:qv:halt}
\State \hskip2em \textbf{finalize} the \metablock $x$ and \textbf{broadcast} $\langle \textsc{Finalized}, x \rangle$, unless already finalized; \textbf{halt}

\end{algorithmic}
\end{algorithm}

\paragraph{Step 1: Dissemination ($\Delta$ time before the deadline).}
Dissemination proceeds exactly as in \chorus{}.
Each proposer $\proc_\pid$ encodes its proposal into $n$ erasure-coded symbols and builds a Merkle tree over them, with root $\mroot_\pid$.
It then sends to each validator $\proc_i$ the $i$-th symbol, $\proc_i$'s \emph{assigned chunk}, together with the Merkle proof of the chunk's inclusion under $\mroot_\pid$ and the header $\langle \pid, \mroot_\pid \rangle$ signed by $\proc_\pid$ (line~\ref{ln:qv:disseminate}).

\paragraph{Step 2: Voting (at the deadline).}
The deadline step, too, mirrors  \chorus.
When time reaches the deadline $D$, a correct validator $\proc_i$ forms one entry per proposer: a \emph{positive} entry $\langle \pid, \mroot_\pid \rangle$ if a valid chunk of $\proc_\pid$'s proposal has arrived (the chunk assigned to $\proc_i$, whose Merkle proof verifies against a header signed by $\proc_\pid$), and a \emph{negative} entry $\langle \pid, \bot \rangle$ otherwise, attesting that nothing arrived on time.
It broadcasts the entries in a single \textsc{Vote} message, attaching for every positive entry its assigned chunk, that is, the root, the symbol, and its Merkle proof (line~\ref{ln:qv:vote}).
Crucially, the message carries \emph{one} signature, computed over the entire vector of entries: unlike in \chorus{}, where each entry is signed on its own, no entry is ever signed individually, and a validator vouches for its whole vote.
We also underline that the signed header is not included in the vote.
Therefore, a vote contains a single signature.
A correct validator that holds a valid chunk from every proposer before the deadline sends this vote right away instead of at the deadline (\emph{early voting}, as in \cadence).

\paragraph{Step 3: Path votes.}
At the end of the first post-deadline round, namely, once $2f+1$ valid votes have been received and time has reached $D + \Delta$, a correct validator $\proc_i$ examines the received votes and forms, for each proposer $\proc_\pid$ separately, a \emph{path-vote entry}, distinguishing three cases.
\begin{compactitem}
   \item First, if $\proc_i$ observes $2f+1$ positive entries for $\proc_\pid$ with the same Merkle root $\mroot_\pid$, or $2f+1$ negative entries for $\proc_\pid$, its entry is a \emph{fast entry}, carrying $\mroot_\pid$ in the former case and $\bot$ in the latter (lines~\ref{ln:qv:fastroot} and~\ref{ln:qv:fastbot}).
   The important difference from \chorus lies in what a fast entry carries.
   In \chorus, a fast vote comes with its own justification: the $\FastQC$ that aggregates the $2f+1$ deadline votes behind it.
   Here, there is no certificate to attach.
   A fast entry is only a claim: $\proc_i$ states that it has seen a quorum for $\proc_\pid$, and offers no proof.

   \item Second, if no such quorum exists but $\proc_i$ observes $f+1$ positive entries for $\proc_\pid$ with the same root $\mroot_\pid$ (ties broken arbitrarily), it proceeds as in  \chorus: it collects the chunks attached to those entries, decodes the proposal, and re-encodes it. If the recomputed root matches, its entry is a \emph{fallback entry} for $\mroot_\pid$ (line~\ref{ln:qv:decode}).

   \item Third, if the reconstruction fails, or if no root gathers even $f+1$ positive entries, its entry is a fallback entry for $\bot$.
\end{compactitem}
Having formed all entries, $\proc_i$ broadcasts them in a single \textsc{PathVote} message (line~\ref{ln:qv:pathvote}).
A validator does not wait for $D + \Delta$ if, for every proposer, it already holds $2f+1$ votes whose entries for that proposer name the same value: all its entries are then fast, and waiting cannot change them, since no two values can each be matched by $2f+1$ of at most $n$ votes. From the deadline onward, the fast path thus proceeds at the speed of the network, as in \chorus.
As with the deadline vote, the message carries a single signature, computed over the entire vector of entries.
Finally, $p_i$ speculatively finalizes the meta-block formed by its fast entries if and only if \emph{every} entry of its path vote is fast and it has not finalized yet (line~\ref{ln:qv:speculate}), after checking the encodings. A fast entry is formed without decoding and may carry the root of an invalid encoding, so, for every fast entry for a root, $p_i$ rebuilds the proposal from the chunks attached to the votes behind the entry and re-encodes it; if the recomputed root differs, it replaces the root by $\bot$.
A speculatively finalized block can still be reverted, but any reversion leaves behind a proof of misbehavior (\Cref{lem:qv:speculation}).

\paragraph{Finalizing on the fast path (at any time).}
If, at any point in time, a correct validator holds, for \emph{every} proposer $\proc_\pid$, $2f+1$ path votes whose entries for $\proc_\pid$ are fast entries for the same value $v_\pid$, a root or $\bot$, it finalizes the \metablock $(v_\pid)_{\text{proposer } \proc_\pid}$ (line~\ref{ln:qv:fastdecide}).
Before finalizing, it waits until it holds, for every root $v_\pid$, $f+1$ \textsc{Vote}s positive for $v_\pid$ (the correct validators whose \textsc{Vote}s are behind the fast entries have broadcast them), rebuilds the proposal from their chunks, and re-encodes it; if the recomputed root differs from $v_\pid$, the encoding is invalid, and the validator replaces $v_\pid$ by $\bot$. The result of this check does not depend on which chunks are used: a valid encoding decodes to the same proposal from any $f+1$ of its symbols, and an invalid one never re-encodes to the root, whichever symbols are used. A correct validator finalizes at most once, and announces its finalization with a \textsc{Finalized} message, which the halting rule below uses. Finalizing does not end its part in the slot: it keeps running the protocol, the fallback consensus included, until it halts.

\paragraph{Step 4: Accepting second-round entries.}
We now turn to the second post-deadline round, the round during which path votes travel, and, from here on, refer to the entries of path votes collectively as \emph{second-round entries}.
The key notion of this step is that of \emph{accepting} a second-round entry. Crucially, accepting an entry is not the same as receiving the path vote that carries it, and acceptance is decided entry by entry (line~\ref{ln:qv:accept}).
\begin{compactitem}
   \item A fallback entry is accepted the moment it is received.

   \item A fast entry, in contrast, must earn its acceptance.
   Namely, recalling that a fast entry is a bare claim of a quorum, $\proc_i$ accepts one only once it sees the claim corroborated by $f+1$ matching deadline entries: a fast entry for a root $\mroot_\pid$ is accepted once $\proc_i$ has received $f+1$ votes with a positive entry $\langle \pid, \mroot_\pid \rangle$, each carrying its sender's chunk of the proposal, whereas a fast entry for $\bot$ is accepted once $\proc_i$ has received $f+1$ votes with a negative entry for $\proc_\pid$.
The chunks are essential for data availability (as we explain later).\end{compactitem}
A single path vote may therefore contain, at the same time, an accepted entry for one proposer and a not yet accepted entry for another.

Now, a \emph{clean set for proposer $\proc_\pid$} is a set of $2f+1$ accepted second-round entries for $\proc_\pid$, from different validators, whose fast entries all agree: any two fast entries in the set carry the same root, or both carry $\bot$.
Fallback entries are left unconstrained.
For example, a set of $2f+1$ accepted fallback entries for $\proc_\pid$ is clean: it holds no fast entries to disagree.
Similarly, a set in which five fast entries all carry the same root, and all remaining entries are fallback entries, is clean as well.
By contrast, a set holding a fast entry for one root and a fast entry for a different one is \emph{not} clean, and neither is a set holding a fast entry for some root alongside a fast entry for $\bot$.

With this notion in hand, the second post-deadline round ends at $\proc_i$ once (1) time has reached $D + 2\Delta$, and (2) $\proc_i$ holds a clean set for every proposer.
We underline that the clean sets of different proposers may come from different validators.
Also, several clean sets may exist for one proposer, and they need not agree, so $\proc_i$ reads the first one it obtains (line~\ref{ln:qv:clean}).

\paragraph{Step 5: Adopting a \metablock.}
With the second post-deadline round over, $\proc_i$ \emph{adopts} a value for every proposer, a Merkle root or $\bot$, thereby forming the \metablock it inputs to the fallback path.
The adoption rule for proposer $\proc_\pid$ reads $\proc_i$'s clean set for $\proc_\pid$ (line~\ref{ln:qv:adopt}).
\begin{compactenum}
   \item If the clean set holds a fast entry for $\bot$, $\proc_i$ adopts $\bot$.
   If it holds a fast entry for a root $\mroot_\pid$, $\proc_i$ does not take the root on faith: it decodes the proposal from the $f+1$ chunks that corroborated the entry and re-encodes it, and it adopts $\mroot_\pid$ if the recomputed root is $\mroot_\pid$, and $\bot$ otherwise (line~\ref{ln:qv:adoptcheck}).
   (The rule is unambiguous: the set being clean, all fast entries in it agree.)
   Adopting a root from a fast entry thus  rests on a proposal that $\proc_i$ has rebuilt itself.

   \item If, instead, the clean set consists of fallback entries alone, $\proc_i$ examines them: if $f+1$ of them carry the same root, $\proc_i$ adopts that root, and otherwise it adopts $\bot$.
\end{compactenum}
It then inputs the adopted \metablock to the fallback consensus (line~\ref{ln:qv:input}).

\paragraph{Step 6: The fallback consensus.}
The fallback path now takes over.
In  \chorus, the fallback consensus black box is an MVBA, whose external validity admits only certified \metablocks. Here, we replace it with \emph{Byzantine agreement with per-proposer strong unanimity}, which \coda implements (\Cref{sec:fallback-implementation}).
The primitive takes as input, from each validator, a \metablock: one entry per proposer, each a Merkle root or $\bot$.
It outputs a \metablock in which an entry may also take the special value $\ALLSAFE$, the primitive's own way of outputting nothing for a proposer (note that $\bot \neq \ALLSAFE$).
The primitive satisfies the following four properties:
\begin{compactitem}
   \item \emph{Agreement:} No two correct validators output different \metablocks.

   \item \emph{Termination:} If every correct validator inputs, every correct validator eventually outputs.

   \item \emph{Strong unanimity (per proposer):} If all correct validators that input a \metablock{} input the same value for proposer $\proc_\pid$, a root or $\bot$, then no correct validator outputs $\ALLSAFE$ for $\proc_\pid$.

   \item \emph{Obligation (per proposer):} If a correct validator outputs a value $v \neq \ALLSAFE$ for proposer $\proc_\pid$, then some correct validator input $v$ for $\proc_\pid$.
\end{compactitem}
Agreement, strong unanimity, and obligation hold even if some correct validators stop taking part in the primitive, possibly before inputting, as validators do when they halt (see the paragraph on halting below); termination is required only when no correct validator stops taking part, since once a correct validator halts, the halting gadget makes every correct validator that participates finalize without the fallback consensus.
With the primitive in place, the fallback path is simple.
Every correct validator inputs its adopted \metablock.
Once the primitive outputs a \metablock $x$, a correct validator replaces every entry $\ALLSAFE$ of $x$ by $\bot$ and finalizes the result: a root  includes its proposer, whereas $\bot$ excludes it (line~\ref{ln:qv:fallbackdecide}).
Inputting the adopted \metablock to the primitive (line~\ref{ln:qv:input}) and starting the primitive form one handler, and so do an output of the primitive and the finalization it causes (line~\ref{ln:qv:fallbackdecide}), so that a crash cannot fall between them.

\paragraph{On halting.}
In  \chorus, a correct validator halts the moment it finalizes --- and it can afford to.
Namely, it holds a quorum certificate (an aggregate signature) attesting to its finalization: a self-contained proof it can disseminate reliably, so that it reaches everyone even after the validator halts, to convince everyone else to stop as well.
Here, that luxury is gone: there is no aggregate signature to hand out.
And simply halting upon finalizing would be dangerous --- the protocol could silently lose a quorum of correct validators and, with it, termination for the validators still running.
In \qchorus, therefore, a correct validator must \emph{not} halt upon finalizing.

Fortunately, the problem admits an easy fix: a \emph{halting gadget}, in the style of Bracha's amplification~\cite{bracha87} (lines~\ref{ln:qv:echo} and~\ref{ln:qv:halt}).
\begin{compactitem}
   \item Upon finalizing a value $v$, a correct validator broadcasts a special message:  $\langle \textsc{Finalized}, v \rangle$.

   \item Upon receiving $f+1$ matching  \textsc{Finalized} messages, a correct validator  finalizes $v$ and broadcasts $\langle \textsc{Finalized}, v \rangle$ itself, unless it has already finalized.

   \item Upon receiving $2f+1$ matching  \textsc{Finalized} messages, a correct validator finalizes $v$ (if it has not already) and halts.
\end{compactitem}
The $f+1$ threshold keeps the echo honest: among $f+1$ matching messages, at least one comes from a correct validator, so the echoed value is indeed the finalized one.
Finalizing at this threshold is therefore safe, and it lets every correct validator that starts participating at time $S$ finalize by time $\max(S, H, \mathrm{GST}) + \Delta$, where $H$ is the time at which the first correct validator halts.
The $2f+1$ threshold makes halting safe: among $2f+1$ matching messages, at least $f+1$ come from correct validators, whose echoes eventually push every correct validator over the $2f+1$ bar: no correct validator halts alone.
Halting also has a constant bound: if every correct validator starts participating by time $T$, every correct validator halts by time $\max(T, H, \mathrm{GST}) + 2\Delta$, since it finalizes by time $\max(T, H, \mathrm{GST}) + \Delta$ and then halts within $\Delta$ (\Cref{lem:qv:eventualhalting}).
The \textsc{Finalized} messages of correct validators match, since all correct validators finalize the same \metablock. Until it halts, a correct validator keeps taking part in the fallback consensus, even after finalizing on the fast path. A validator that has halted handles no more messages of the slot, including those of the fallback consensus, but its ongoing broadcasts remain alive and keep retransmitting. In this version, these broadcasts never end. We expect the protocol that schedules the slots (in \quorus, the \conductor) to end them once it has moved far enough ahead that no correct validator needs the slot's messages any more (\Cref{sec:conductor}).
A correct validator that finalizes through the halting gadget finalizes a \metablock that some correct validator finalized on the fast path or on the fallback path: the first correct validator to broadcast $\langle \textsc{Finalized}, x \rangle$ did so upon finalizing $x$ on one of the two paths, since finalizing through the gadget requires $f+1$ \textsc{Finalized} messages for $x$, one of them from a correct validator that broadcast earlier.

This replaces the commitment certificates of \chorus: instead of receiving a transferable certificate from another validator, each validator observes a quorum of matching \textsc{Finalized} messages itself.

\para{Data availability.}\label{sec:qv:availability}
As described, validators agree, proposer by proposer, on a Merkle root or on exclusion. A Merkle root, however, is only a digest, whereas what \qchorus must ultimately provide is the proposal the proposer originally disseminated.
We must therefore argue that whenever a root is finalized for $\proc_\pid$, some correct validator holds the corresponding proposal and can send it to the others.
The argument is per proposer, so we fix a proposer $\proc_\pid$.

Consider first the fast path.
 A correct validator that finalizes a root $\mroot$ for $\proc_\pid$ on the fast path holds the proposal itself: before finalizing, it rebuilt the proposal from $f+1$ chunks under $\mroot$, and it kept $\mroot$ only because the proposal re-encodes to it (line~\ref{ln:qv:fastcheck}).

The fallback path asks for slightly more care.
Suppose a Merkle root $\mroot$ is output for $\proc_\pid$ on the fallback path.
As $\mroot \neq \ALLSAFE$, obligation guarantees that some correct validator input $\mroot$ for proposer $\proc_\pid$. We argue that some correct validator then holds the proposal behind $\mroot$ in full.
Recall from the adoption rule the two ways in which a correct validator comes to input $\mroot$.
\begin{compactenum}
   \item It may have adopted $\mroot$ off a clean set holding a fast entry for $\mroot$.
   Before adopting, it decoded the proposal from the $f+1$ chunks that corroborated the entry and checked that it re-encodes to $\mroot$ (line~\ref{ln:qv:adoptcheck}). It therefore holds the proposal.

   \item Otherwise, it adopted $\mroot$ off a clean set consisting of fallback entries alone, $f+1$ of which carry $\mroot$.
   Recall how a fallback entry for $\mroot$ comes to be: its sender collected chunks, decoded the proposal, and re-encoded it back to $\mroot$.
   Among these $f+1$ senders, at least one is correct, and that validator holds the proposal in its entirety.
\end{compactenum}
In either case, some correct validator holds the proposal and can send it to the others.
Finally, a correct validator that finalizes through the halting gadget finalizes a \metablock that some correct validator finalized on the fast path or on the fallback path (see the paragraph on halting), so some correct validator holds the proposal in this case too.
How a correct validator obtains chunks it lacks, for instance after finalizing through the halting gadget, lies outside the protocol.
We assume only a retrieval mechanism, with retransmission over the same unreliable channels, by which a correct validator eventually obtains data that another correct validator holds. It runs outside the slot's instance, so halting stops neither side: correct validators keep the data in persistent memory also after halting, and a halted validator still obtains data through it.

\subsection{Guarantees}

We now state the guarantees of the slot consensus; \Cref{app:slot} proves them.
Agreement, termination, and proposal inclusion are three of the properties that \cadence requires of a slot consensus, and we state them similarly to \cadence~\cite{cadence}; of the other three, slot safety (a finalized block carries the slot's number) holds trivially, and we do not address hiding and quiescence (a correct validator sends messages for a slot only while it participates in it).
Agreement, speculation, and the complexity bound hold for every slot, termination and halting for every slot in which every correct validator starts participating, and proposal inclusion and the accountability lemmas for every timely slot (\Cref{sec:model}).

\begin{restatable}[Agreement]{lemma}{qvagreement}\label{lem:qv:agreement}
\qchorus satisfies agreement: correct validators never finalize conflicting blocks for the slot.
\end{restatable}

\begin{restatable}[Speculation]{lemma}{qvspeculation}\label{lem:qv:speculation}
In \qchorus, if a correct validator speculatively finalizes a block that differs from the block that a correct validator finalizes, then the misbehavior of at least one faulty validator can be proven.
\end{restatable}

\begin{restatable}[Proposal inclusion]{lemma}{qvcensorship}\label{lem:qv:censorship}
\qchorus satisfies proposal inclusion: if the slot is timely, the finalized block contains the proposals of all correct proposers.
\end{restatable}

\begin{restatable}[Termination]{lemma}{qvtermination}\label{lem:qv:termination}
\qchorus satisfies termination: if every correct validator starts participating in the slot, every correct validator eventually finalizes a block for the slot.
\end{restatable}
Each correct validator then obtains the block's proposals through the retrieval mechanism of \Cref{sec:qv:availability}.

A correct validator keeps taking part in the slot after finalizing, until it halts (see the paragraph on halting in \Cref{sec:qv:description}), and the protocol that schedules the slots, the \conductor in \quorus, would complete the slot at that point (\Cref{sec:conductor}).

\begin{restatable}[Halting]{lemma}{qveventualhalting}\label{lem:qv:eventualhalting}
If every correct validator starts participating in the slot, every correct validator eventually halts.
Moreover, if every correct validator has finalized by time $T$, every correct validator halts by time $\max(T, \mathrm{GST}) + \Delta$.
\end{restatable}

\begin{restatable}[Message complexity \& signature count]{lemma}{qvcomplexity}\label{lem:qv:complexity}
In \qchorus, every  message that a correct validator sends in dissemination (\textsc{Chunk}) or in voting (\textsc{Vote} and \textsc{PathVote}) carries exactly one signature.
Moreover, a correct validator  initiates at most one broadcast of a \textsc{Vote}, at most one of a \textsc{PathVote}, and at most one of a \textsc{Finalized} message, and, if it is a proposer, initiates at most one \textsc{Chunk} send to each validator, retransmissions aside.
\end{restatable}

\para{Accountability.}

We now turn the signatures of the voting path into evidence.
When the fast path fails in a timely slot, someone is to blame, and we show that, as long as the Byzantine validators are few enough, the correct validators collectively hold a transferable proof against a specific proposer.
The threshold is half the fault tolerance: with up to $f$ Byzantine validators, $f$ of them can each sign a negative entry for a correct proposer, so no count of negative entries below $f+1$ proves anything, and, as it turns out, the protocol cannot guarantee $f+1$ of them per proposer.
With at most $\lfloor f/2 \rfloor$ Byzantine validators, however, $\lfloor f/2 \rfloor + 1$ signed negative entries are enough to convict, and the protocol does guarantee them.

\begin{restatable}[Accountability under $\lfloor f/2 \rfloor$ faults]{lemma}{qvacchalf}\label{lem:qv:acc-half}
Suppose that the slot is timely and that at most $\lfloor f/2 \rfloor$ validators are Byzantine.
If some correct validator has not finalized by time $D + 2\Delta$, then there exists a proposer $\proc_\pid$ such that the correct validators collectively hold either
(1) two \textsc{Chunk} messages signed by $\proc_\pid$ that carry different roots, or
(2) a \textsc{Chunk} message signed by $\proc_\pid$ together with at least $\lfloor f/2 \rfloor + 1$ \textsc{Vote} messages, from different validators, whose entry for $\proc_\pid$ is negative.
\end{restatable}

Lastly, we exploit a feature of the fallback consensus presented in \Cref{sec:fallback-implementation}.
Its first view is a leaderless view, in which every validator votes for its own input. When the inputs of all correct validators coincide, this view decides within two message delays, whatever the faulty validators do.
A failure to decide quickly therefore reveals that correct validators adopted different values, and the evidence behind that disagreement can be traced back either to a proposer or to a validator that equivocated.

\begin{restatable}[Accountability under $f$ faults]{lemma}{qvaccf}\label{lem:qv:acc-f}
Suppose that \qchorus runs with \coda as its fallback consensus and that the slot is timely.
If some correct validator has not finalized by time $D + 4\Delta$, then there exists a proposer $\proc_\pid$ such that either some correct validator holds $2f + 1$ \textsc{PathVote} messages, from different validators, whose entry for $\proc_\pid$ is a fallback entry, or two correct validators together hold two \textsc{Vote} messages signed by the same validator whose entries for $\proc_\pid$ differ.
\end{restatable}

\section{\coda{}: Our Fallback Consensus}
\label{sec:fallback-implementation}

This section presents \coda, an implementation of the Byzantine agreement with per-proposer strong unanimity that \qchorus{} requires.
Unlike a generic Byzantine agreement, \coda is built for multiple concurrent proposers: validators input and decide \metablocks, with one entry per proposer, and strong unanimity and obligation hold for each proposer separately.
\coda builds on Forget-IT~\cite{forgetit}, a leader-based Byzantine agreement protocol that needs no cryptography beyond authenticated channels and keeps its persistent memory constant.
\Cref{sec:coda:differences} lists how it differs from Forget-IT.
\coda satisfies the four properties that Step~6 of \Cref{sec:protocol} requires, which \Cref{sec:coda:model} restates in the terms of \coda.
Under favorable conditions, \coda{} decides within two message delays, in its first view, which has no leader.
Moreover, \coda{} bounds its persistent memory: for a fixed number of proposers, it is constant in the number of validators when no proposer equivocated before \coda started, and at most linear in it otherwise.
It does not grow with the number of views, so a long asynchronous period costs no persistent memory.

\para{The leaderless view.}
The first view of \coda has no leader.
It consists of two voting rounds, with an implicit proposal step: each validator votes right away for its own input.
When all correct validators start \coda at the same time, after GST, with the same \metablock, this view decides within two message delays.
For the guarantees of \coda, the leaderless view is an optimization.
Without it, \coda would begin with a view change that starts at once and turns the validators' inputs into the values that the first leader proposes, and would remain correct, but its best-case latency would grow from two message delays to five: two for the view change and three for the view of the first leader, if that leader is correct.

The leaderless view has a practical motivation.
In a deployment, proposers do not send their chunks at the slot's starting time, but as late as they can while still reaching all correct validators by the deadline, so that their proposals include the transactions that arrived shortly before it.
Occasionally, a proposer fails to do so, and its chunks reach only between $f+1$ and $2f$ validators in time.
The votes for that proposer then hold neither $2f+1$ positive nor $2f+1$ negative entries, and the fast path fails.
Suppose that the slot is timely and that all validators follow the protocol, except that some proposers send their chunks after the starting time, so that these reach only between $f+1$ and $2f$ validators in time, which makes these proposers faulty in the model of \Cref{sec:model}.
Suppose also that their other messages still arrive within $\Delta$, which the model guarantees only for correct validators.
Then every correct validator receives the same votes and path votes, adopts the same \metablock at time $D + 2\Delta$, with the late proposals included through fallback entries, and finalizes in the leaderless view by $D + 4\Delta$, instead of by $D + 7\Delta$ without it, and then only if the first leader is correct.
Omission failures, asynchrony, and Byzantine voters can give correct validators different inputs or delay the votes of view~$1$ past its timeout, in which case the leaderless view may end without a decision and a later view decides.

The leaderless view also gives \Cref{lem:qv:acc-f} its time bound.
Faulty validators can make the fast path fail without leaving evidence of their faults while all correct validators still adopt the same \metablock.
Without the leaderless view, a faulty first leader could then delay the decision further, again without evidence.
With it, if all correct validators adopt the same \metablock, they finalize by $D + 4\Delta$ whatever the faulty validators do.
If some correct validator has not finalized by then, the correct validators thus adopted different values for some proposer, and such a difference always leaves evidence, as the proof of the lemma shows.

\subsection{Protocol Description}\label{sec:coda:description}

\para{Model.}
\coda runs in the model of \Cref{sec:model}, but it does not need the synchronized clocks of that model: each validator uses only its local clock, and only to measure the time elapsed since it entered a view.
Its proofs allow arbitrary drift of the local clocks before GST and require only drift-free clocks after GST.
\coda uses no cryptography beyond the authentication of channels, so its guarantees hold even against a computationally unbounded adversary, unlike those of \qchorus, which relies on a collision-resistant hash function.
The conventions of \Cref{sec:qv:description} on firing rules, time guards, and counting messages apply to \Cref{alg:fallback} as well. From the time it starts \coda, a validator keeps the \coda messages it receives for its previous view and later views. Upon entering a view, it discards the kept messages of views below its new previous view, and a crash erases them all. Kept messages count toward every rule whose condition they complete later, also one that a view entry completes.
Before it starts \coda, a validator executes none of its rules.
A correct validator may also stop taking part in \coda, that is, execute none of its handlers from some time on while its broadcasts remain alive, as validators of \qchorus do when they halt.

\para{Protocol overview.}
Every validator starts \coda with a \metablock, its \emph{input} ($\propi$ for validator $p_i$): in \qchorus, the \metablock it adopted in Step~5 of \Cref{sec:protocol}, which holds a root or $\bot$ for each proposer.
\coda decides a \metablock whose entry for each proposer $p_\pid$ is either a value that a correct validator input for $p_\pid$ or the special value $\ALLSAFE$, which, like $\bot$, excludes $p_\pid$.
\coda decides $\ALLSAFE$ for $p_\pid$ only if the inputs of correct validators differ for $p_\pid$.
Inside \coda, an entry of a \metablock may also be $\ALLSAFE$.
The protocol unfolds in views $1, 2, 3, \ldots$.
View~$1$ has no leader, and every later view $V$ has a leader, $\leader(V)$, which rotates round-robin among the validators and proposes a \metablock at the beginning of its view.
In each view, validators vote in two rounds.
They first send $\VOTEONE$ messages for a \metablock, in view~$1$ for their own input and later for the leader's proposal.
A validator that receives $2f+1$ $\VOTEONE$ messages for the same \metablock $x$ sends a $\VOTETWO$ message for $x$, unless it has sent one in the view already.
A validator that receives $2f+1$ $\VOTETWO$ messages for $x$ decides $x$.
Every view, also one that decided, ends with a \emph{view change}, in which validators exchange $\ECHO$ messages.
These announce, for each proposer, values that are \emph{safe} to carry into the next view, that is, consistent with every decision made in this view or an earlier one.
A validator enters the next view with a \metablock of such values, its \emph{guard} for that view.
\coda has no rule for stopping: validators keep running it after they decide, since the others may still need their messages to decide.
In \qchorus, a validator stops running \coda when it halts through the halting gadget of \Cref{sec:protocol}, which it does only after finalizing.

\para{State.}
A validator keeps two \metablocks for each of its current and previous views.
Its \emph{guard} for view $V$ is the \metablock with which it entered $V$: its input for view~$1$, and safe values from the view change for later views.
A correct leader of $V$ proposes its guard.
The guard reflects the earlier views: if a \metablock was decided in a view smaller than $V$, it is the guard for $V$ of every correct validator that enters $V$.
A validator's \emph{lock} for view $V$ is the \metablock for which it received $2f+1$ $\VOTEONE$ messages in $V$ before sending its $\VOTETWO$ message, or $\NONE$ if there is none.
The lock reflects view $V$ itself: a view can decide only a \metablock that at least $f+1$ correct validators locked in it.
The $\VOTETWO$ message of a validator in view $V$ carries its lock and its guard: the view change needs both to determine which values are safe.

In persistent memory, a validator keeps its current view, the time it entered that view, and whether it has decided.
For its current and previous views, it also keeps its guard, its lock, and the votes and echoes it has broadcast.
The \coda messages that a validator receives, echoes included, are not kept in persistent memory, so a crash erases them.
After a crash, retransmission delivers again the messages whose broadcasts are still alive.

\para{Corroboration.}
Without certificates, a validator cannot show others that it holds a quorum.
\coda relies instead, in two places, on corroboration by $f+1$ messages, as \qchorus does for fast entries.
A validator votes for a leader's proposal only once the proposal is \emph{justified}, that is, for each of its entries, $f+1$ echoes of the previous view announce that entry as safe.
The view change counts a $\VOTETWO$ message that claims a lock only once it is \emph{accepted}, that is, $f+1$ $\VOTEONE$ messages for the lock have arrived.
Among $f+1$ messages from different validators, at least one comes from a correct validator.
Every entry of a justified proposal thus has an echo from a correct validator, possibly a different one for each entry, and every accepted lock has a $\VOTEONE$ message for it from a correct validator.

\para{The steps of a view.}
\Cref{alg:fallback} gives the pseudocode, which we describe in seven steps, shown in \Cref{fig:coda-view}: Steps~1 to~4 follow the message rounds of a view, and Steps~5 to~7 form the view change.

\begin{figure}[t]
\centering
\begin{tikzpicture}[
  font=\footnotesize,
  box/.style={draw=black!60, rounded corners=2pt, align=center, inner sep=3.5pt, execute at begin node={\hyphenpenalty=10000}},
  decbox/.style={box, fill=green!12, anchor=north west, minimum height=2.6cm},
  vcbox/.style={box, fill=orange!15, anchor=north west, minimum height=2.95cm},
  rowtitle/.style={font=\footnotesize\itshape, anchor=base west, inner sep=0pt},
  arr/.style={-stealth, semithick},
]
\def\ytop{0}      %
\def\yhub{-3.21}  %
\def\ybot{-4.0}   %
\def\yloop{-7.25} %
\node[rowtitle] at (0,\ytop+0.15) {Deciding in view $V$};
\node[decbox, text width=3.1cm] (s1) at (0,\ytop)
  {\textbf{Step 1: proposal}\\ $\leader(V)$ broadcasts its guard $x$: $\langle \PREPARE, V, x \rangle$\\[3pt]
   {\scriptsize\color{black!65} view 1: skipped (no leader)}};
\node[decbox, text width=3.75cm] (s2) at (3.75,\ytop)
  {\textbf{Step 2: first vote}\\ once $x$ is justified ($f{+}1$ $\ECHO$s of $V{-}1$ for each entry): broadcast $\langle \VOTEONE, V, x \rangle$\\[3pt]
   {\scriptsize\color{black!65} view 1: own input, at once}};
\node[decbox, text width=3.45cm] (s3) at (8.15,\ytop)
  {\textbf{Step 3: second vote}\\ on $2f{+}1$ $\langle \VOTEONE, V, x \rangle$, if no $\VOTETWO$ sent in~$V$: lock $x$, broadcast $\langle \VOTETWO, V, x, g \rangle$, where $g$ is its own guard};
\node[decbox, text width=3.4cm] (s4) at (12.25,\ytop)
  {\textbf{Step 4: decide}\\ on $2f{+}1$ $\VOTETWO$ of $V$ or $V{-}1$ with the same lock \mbox{$x \neq \NONE$}: decide $x$\\[3pt]
   {\scriptsize\color{black!65} the view continues}};
\node[box, minimum width=11.6cm, minimum height=0.55cm] (v2) at (8.9,\yhub)
  {$\VOTETWO$ messages of view $V$, each with its sender's lock (or $\NONE$) and guard};
\node[rowtitle] at (0.7,\ybot+0.15) {The view change};
\node[vcbox, text width=2.95cm] (s5) at (0.7,\ybot)
  {\textbf{Step 5: timeout}\\ $5\Delta$ after entry, if no $\VOTETWO$ sent in~$V$: broadcast $\langle \VOTETWO, V, \NONE, g \rangle$ (no lock), where $g$ is its own guard};
\node[vcbox, text width=6.65cm] (s6) at (4.3,\ybot)
  {\textbf{Step 6: echo safe values}\\ from $5\Delta$ after entry, for every set $S$ of $2f{+}1$ accepted $\VOTETWO$ of $V$ with no two different locks, and every proposer $p_j$: broadcast $\langle \ECHO, V, j, v \rangle$, where $v$ is the lock's entry for $p_j$ if $S$ holds a lock, else the entry for $p_j$ that $f{+}1$ guards in $S$ share, else $\ALLSAFE$};
\node[vcbox, text width=4.05cm] (s7) at (11.6,\ybot)
  {\textbf{Step 7: next view}\\ on $f{+}1$ $\ECHO$s with the same value for the same proposer: echo it too\\[2pt]
   on $2f{+}1$ $\langle \ECHO, V, j, x[j] \rangle$ for every proposer $p_j$: enter view $V{+}1$ with guard $x$};
\draw[arr] (s1) -- (s2);
\draw[arr] (s2) -- (s3);
\draw[arr] (s3.south) -- (s3.south |- v2.north);
\draw[arr] (s4.south |- v2.north) -- (s4.south);
\draw[arr] (3.4,0 |- s5.north) -- (3.4,0 |- v2.south);
\draw[arr] (6.4,0 |- v2.south) -- (6.4,0 |- s6.north);
\draw[arr] (s6) -- (s7);
\draw[arr, rounded corners=4pt] (s7.south) -- (s7.south |- 0,\yloop) -| (0.35,0 |- s1.south)
  node[pos=0.25, below, font=\scriptsize] {view $V{+}1$, with guard $x$};
\end{tikzpicture}
\caption{One view $V$ of \coda.
A correct validator sends at most one $\VOTETWO$ message in $V$, with its lock in Step~3 or without a lock in Step~5, and both the decision of Step~4 and the echoes of Step~6 use these messages.
The view change ends every view, also one that decided, and moves the validator to view $V+1$ with a guard, which it proposes in Step~1 if it leads $V+1$.
View~1 has no leader and starts at Step~2.}
\label{fig:coda-view}
\end{figure}
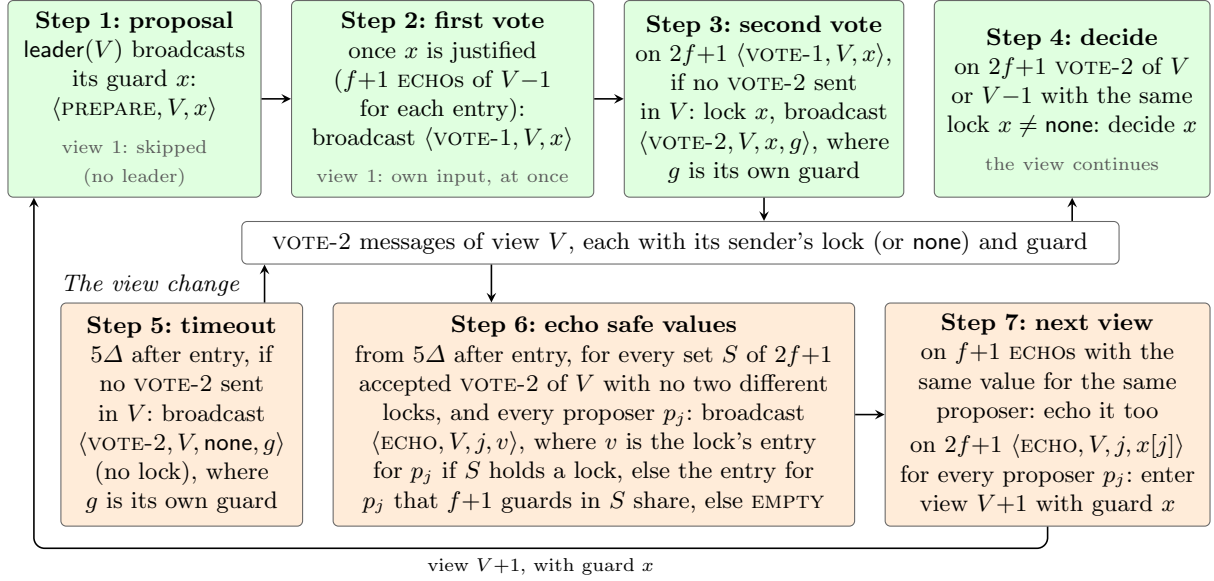

\paragraph{Step 1: Proposal (views $V \geq 2$).}
The leader of a view $V \geq 2$ opens it by broadcasting its guard $x$ in a message $\langle \PREPARE, V, x \rangle$ (line~\ref{ln:core:on_enter_if_leader}).
View~$1$ has no leader and skips this step.

\paragraph{Step 2: First vote.}
In view~$1$, every validator broadcasts $\langle \VOTEONE, 1, x \rangle$ upon starting, for its own input $x$ (line~\ref{ln:core:on_enter_dummy_vote1}).
A validator in a view $V \geq 2$ that has received the leader's $\langle \PREPARE, V, x \rangle$ broadcasts $\langle \VOTEONE, V, x \rangle$ as soon as $x$ is \emph{justified}: for every proposer $p_\pid$, it has received $f+1$ $\ECHO$ messages of view $V - 1$ that announce $x[\pid]$ as safe for $p_\pid$, so that at least one correct validator announced it (line~\ref{ln:core:on_prepare_header}).
A validator sends at most one $\VOTEONE$ message per view.

\paragraph{Step 3: Second vote.}
A validator that receives $2f+1$ $\VOTEONE$ messages of its current view $V$ for the same \metablock $x$ sets its lock for $V$ to $x$, whether or not it voted for $x$ itself, and broadcasts $\langle \VOTETWO, V, x, g \rangle$, where $g$ is its guard for $V$ (lines~\ref{ln:core:on_vote1_quorum_header} to~\ref{ln:core:on_vote1_quorum_set_vote2}).
It skips this step if it has already broadcast a $\VOTETWO$ message in $V$.
No two correct validators lock different \metablocks in the same view, since two sets of $2f+1$ $\VOTEONE$ messages from different validators share a correct sender, which sends only one $\VOTEONE$ message per view.

\paragraph{Step 4: Decision.}
Upon receiving $2f+1$ $\VOTETWO$ messages of the same view, its current or its previous one, with the same lock $x \neq \NONE$, a validator decides $x$, unless it has decided before (line~\ref{ln:core:on_vote2_quorum_header}).
The previous view counts as well, since a validator may enter the next view before these messages reach it.

\para{The view change.}
A view may end without a decision, for instance view~$1$ when correct validators have different inputs, and a later view when its leader is faulty.
The view change, Steps~5 to~7, ends every view and moves the validators to the next one.
It carries every decision forward: if a \metablock $x$ was decided in view $V$ or earlier, every correct validator that enters a later view has guard $x$ there, and correct validators vote only for $x$ in later views.

\paragraph{Step 5: Timeout.}
A validator that has been in a view $V$ for $5\Delta$, measured on its local clock, and has not broadcast a $\VOTETWO$ message in $V$ broadcasts one without a lock, $\langle \VOTETWO, V, \NONE, g \rangle$, where $g$ is its guard for $V$ (line~\ref{ln:core:on_timer_set_vote2}).
The view change, which needs $2f+1$ $\VOTETWO$ messages, can thus proceed even if no \metablock gathers $2f+1$ $\VOTEONE$ messages.

\paragraph{Step 6: Echoing safe values.}
The view change reads the locks and guards that the $\VOTETWO$ messages of the view carry, and it counts only \emph{accepted} messages (see the paragraph on corroboration): a $\VOTETWO$ message with a lock $x$ is accepted once $f+1$ $\VOTEONE$ messages for $x$ of the same view have arrived (line~\ref{ln:vc:rule_accept_vote2_lock}), one without a lock on receipt, and at most one per sender and view.
A lock needs this corroboration here because a single lock determines the values echoed for a whole set of messages.
The decision of Step~4 needs none, as it counts $2f+1$ messages with the same lock.

Once $5\Delta$ has elapsed since the validator entered its current view $V$, for every set $S$ of $2f+1$ accepted $\VOTETWO$ messages of $V$ whose locks are all the same \metablock, and for every proposer $p_\pid$, it \emph{echoes} a value $v$, that is, it broadcasts $\langle \ECHO, V, \pid, v \rangle$, which announces $v$ as safe for $p_\pid$ in view $V + 1$ (line~\ref{ln:vc:on_vote2_quorum_compute_approved}). 
It chooses $v$ as follows.
\begin{compactitem}
   \item If $S$ holds a lock $\ell$, then $v = \ell[\pid]$.
   \item Otherwise, $v$ is the entry for $p_\pid$ that the guards of $f+1$ messages in $S$ share, or $\ALLSAFE$ if there is none.
\end{compactitem}
The set $S$ may mix messages without a lock and messages with the same lock.
Two different locks, however, are never allowed in the same set $S$.
Since no two correct validators lock different \metablocks in a view (Step~3), one of the senders of conflicting locks is faulty, and the messages of the correct validators carry no conflicting locks.
As long as it cannot form a set without conflicting locks, the validator therefore waits for further messages.
As the rule fires for every such set $S$, a validator may echo several values for one proposer in a view.
A validator need not enumerate the sets $S$, however: it suffices to count the accepted $\VOTETWO$ messages by lock and by guard entry, and \Cref{app:coda:count} gives the conditions on these counts.
The delay of $5\Delta$ is long enough that, if every correct validator has started \coda by GST and none stops taking part, then in a view whose leader is correct and which no correct validator enters before GST, every correct validator locks the leader's \metablock before any correct validator leaves the view.

The values that correct validators echo are safe.
If view $V$ decided a \metablock $x$, at least $f+1$ correct validators locked $x$, so every such set $S$ holds the lock $x$.
If an earlier view decided $x$, every correct validator entered $V$ with guard $x$ and votes in $V$ only for $x$, so an accepted lock in $S$ is $x$, and if $S$ holds no lock, the guards of the correct validators in $S$ share the entry $x[\pid]$.

\paragraph{Step 7: Entering the next view.}
Echoes are amplified as in Bracha's broadcast~\cite{bracha87}: upon $f+1$ echoes with the same value for the same proposer, of its current or previous view, a validator echoes that value too, for the same view (line~\ref{ln:vc:upon_echo_amplification}).
A validator enters the next view once, for every proposer, $2f+1$ echoes of the same view agree.
Precisely, let $V$ be its current view, or a later one if the validator has fallen behind.
Once it has received, for every proposer $p_\pid$, $2f+1$ messages $\langle \ECHO, V, \pid, v \rangle$ with the same value $v$, it enters view $V + 1$ with the guard whose entry for each proposer $p_\pid$ is the value $v$ of these echoes (line~\ref{ln:vc:upon_enter_next_view}).
The $2f+1$ echoes for one proposer may come from other validators than those for another proposer.
A validator keeps the state of its current and previous views only, and upon entering a view, it abandons its broadcasts of the views before the previous one (lines~\ref{line:persistant_memory} and~\ref{ln:rule_abandon}), which bounds its persistent memory however many views pass.
A validator that has fallen behind may therefore never receive the messages of the views it missed.
Instead, once it has received, for a later view $V$ and every proposer $p_\pid$, $f+1$ messages $\langle \ECHO, V, \pid, v \rangle$ with the same value $v$, it catches up by entering view $V$ directly, with the guard built in the same way from these echoes (line~\ref{ln:vc:upon_catch_up}).
One of these echoes comes from a correct validator, which has thus reached view $V$, and values safe for view $V+1$ are safe for $V$ as well.

\providecommand{\LineLabel}[1]{\label{#1}}

\begin{algorithm}[p]
\caption{\coda: Pseudocode (for validator $p_i$)}
\label{alg:fallback}
\begin{algorithmic}[1]
\scriptsize

\State \textbf{Notation:} \LineLabel{ln:core:notation_header}
\State \hskip2em - A \metablock{} $x$ assigns a value $x[j] \in \VALUE \cup \{ \ALLSAFE \}$ to each proposer $p_j$. \LineLabel{ln:core:notation_metablock}
\State \hskip2em - $\VALUE$ is the set of values that validators input for a proposer (in \qchorus, a root or $\bot$), and $\ALLSAFE \notin \VALUE$. \LineLabel{ln:core:const_allsafe}
\State \hskip2em - $\NONE$ denotes absent state (no lock, no message); it is never proposed, voted for, or decided. \LineLabel{ln:core:notation_none}

\fullmedskip
\State \textbf{Global variables (persistent memory):} \LineLabel{ln:core:global_header}
\State \hskip2em $\VIEW$ $\mathit{view}_i$ \BlueComment{$p_i$'s current view} \LineLabel{ln:core:ps_view}
\State \hskip2em Time of entry into the current view and whether $p_i$ has decided

\fullmedskip
\State \textbf{View-specific variables (persistent memory), indexed by view $V$:} \LineLabel{ln:core:vs_header}
\State \hskip2em $\METABLOCK$ $\guardi[V] \gets \NONE$ \BlueComment{the \metablock{} with which $p_i$ entered view $V$, which it proposes if it leads view $V$} \LineLabel{ln:core:pv_guard}
\State \hskip2em $\METABLOCK$ $\mathit{lock}_i[V] \gets \NONE$ \BlueComment{$p_i$'s locked \metablock{} for view $V$, if any} \LineLabel{ln:core:pv_lock}
\State \hskip2em $\MESSAGE$ $\voteone[V]$, $\votetwo[V] \gets \NONE$ \BlueComment{$\VOTEONE$ and $\VOTETWO$ messages broadcast by $p_i$ in view $V$} \LineLabel{ln:core:pv_votes}
\State \hskip2em $\mathsf{Set}(\MESSAGE)$ $\echosone[V] \gets \emptyset$ \BlueComment{$\ECHO$ messages of view $V$ broadcast by $p_i$} \LineLabel{ln:core:pv_echo_set}

\fullmedskip
\State \textbf{Rules:} \LineLabel{ln:core:storage_rule}
\State \hskip2em - Validator $p_i$ maintains (in persistent memory) view-specific variables for views $\viewi$ and $\viewi - 1$ only, and no received \coda messages. \label{line:persistant_memory}
\State \hskip2em - Upon entering view $V$, abandon all ongoing broadcasts of \coda{} messages from views smaller than $V - 1$. \LineLabel{ln:rule_abandon}
\State \hskip2em - A $\VOTETWO$ message is accepted only if its guard is a \metablock{} and its lock is either $\NONE$ or a \metablock{}. \LineLabel{ln:vc:rule_vote2_format}
\State \hskip2em - A received message $\langle \VOTETWO, V, x, \cdot \rangle$ with $x \neq \NONE$ is \emph{accepted} only after $p_i$ has received $f+1$ messages $\langle \VOTEONE, V, x \rangle$. 
\LineLabel{ln:vc:rule_accept_vote2_lock}
\State \hskip2em - A received message $\langle \VOTETWO, V, \NONE, \cdot \rangle$ is \emph{accepted} immediately upon reception. \LineLabel{ln:vc:rule_accept_vote2_bot}
\State \hskip2em - For each view $V$, at most one $\langle \VOTETWO, V, \cdot, \cdot \rangle$ message is accepted from each validator. \LineLabel{ln:vc:rule_one_vote2_per_sender}

\fullmedskip
\State \textbf{Functions:} \LineLabel{ln:core:functions_header}
\State \hskip2em - $\mathsf{echo}(\VIEW \text{ } V, \PROPOSER \text{ } p_j, \VALUE \cup \{ \ALLSAFE \} \text{ } v)$: \BlueComment{broadcasts each echo at most once, and records it} \LineLabel{ln:vc:fn_echo}
\State \hskip4em \textbf{if} $\langle \ECHO, V, j, v \rangle \notin \echosone[V]$: \LineLabel{ln:vc:fn_echo_if}
\State \hskip6em $\echosone[V] \gets \echosone[V] \cup \{ \langle \ECHO, V, j, v \rangle \}$; \textbf{broadcast} $\langle \ECHO, V, j, v \rangle$ \LineLabel{ln:vc:fn_echo_broadcast}

\fullmedskip
\State \hskip2em - $\justified(\METABLOCK \text{ } x)$: \BlueComment{checks whether every entry of $x$ was echoed as safe in the previous view} \LineLabel{ln:core:fn_justified}
\State \hskip4em \textbf{return} whether $p_i$ has received $f+1$ messages $\langle \ECHO, \viewi - 1, j, x[j] \rangle$ for every proposer $p_j$ \LineLabel{ln:core:fn_justified_return}

\fullmedskip
\State \textbf{upon starting:} \LineLabel{ln:core:init_header}
\State \hskip2em $\guardi[1] \gets$ $p_i$'s input \LineLabel{ln:core:init_set_proposal}
\State \hskip2em \textbf{enter} view $1$ \LineLabel{ln:core:init_enter_view1}

\fullmedskip
\State \textbf{upon entering} view $V$: \LineLabel{ln:core:on_enter_header}
\State \hskip2em $\mathit{view}_i \gets V$; record the time of entry \BlueComment{update the current view} \LineLabel{ln:core:on_enter_set_view}
\State \hskip2em \textbf{if} $V = 1$: \BlueComment{view $1$ is the leaderless view} \LineLabel{ln:core:on_enter_dummy_view}
\State \hskip2em\hskip2em $\voteone[1] \gets$ $\langle \VOTEONE, 1, \guardi[1] \rangle$; \textbf{broadcast} $\voteone[1]$ \BlueComment{everyone votes for its own input} \LineLabel{ln:core:on_enter_dummy_vote1}
\State \hskip2em \textbf{else if} $p_i = \leader(V)$: \textbf{broadcast} $\langle \textsc{prepare}, V, \guardi[V] \rangle$ \BlueComment{leader proposes a \metablock} \LineLabel{ln:core:on_enter_if_leader}

\fullmedskip
\State \textbf{upon receipt of} $\langle \PREPARE, \viewi, x \rangle$ \textbf{from} $\leader(\viewi)$ \textbf{and} $\justified(x)$: \LineLabel{ln:core:on_prepare_header}\LineLabel{ln:core:on_prepare_wait_guard}
\State \hskip2em \textbf{if} $\voteone[\viewi] = \NONE$: \LineLabel{ln:core:on_prepare_if_no_vote1}
\State \hskip4em $\voteone[\viewi] \gets \langle \VOTEONE, \viewi, x \rangle$; \textbf{broadcast} $\voteone[\viewi]$ \LineLabel{ln:core:on_prepare_broadcast_vote1}

\fullmedskip
\State \textbf{upon receipt of} $2f + 1$ $\langle \VOTEONE, \viewi, x \rangle$ messages for the same \metablock{} $x$: \LineLabel{ln:core:on_vote1_quorum_header}
\State \hskip2em \textbf{if} $\votetwo[\viewi] = \NONE$: \LineLabel{ln:core:on_vote1_quorum_if_no_vote2}
\State \hskip4em $\locki[\viewi] \gets x$ \BlueComment{lock on the \metablock{} $x$} \LineLabel{ln:core:on_vote1_quorum_set_lock}
\State \hskip4em $\votetwo[\viewi] \gets \langle \VOTETWO, \viewi, \locki[\viewi], \guardi[\viewi] \rangle$; \textbf{broadcast} $\votetwo[\viewi]$ \LineLabel{ln:core:on_vote1_quorum_set_vote2}

\fullmedskip
\State \textbf{upon receipt of} $2f + 1$ $\langle \VOTETWO, V, x, \cdot \rangle$ messages for the same \metablock{} $x \neq \NONE$ and some $V \in \{ \viewi - 1, \viewi \}$: \LineLabel{ln:core:on_vote2_quorum_header}
\State \hskip2em \textbf{if} $p_i$ has not decided yet: \textbf{trigger} $\mathsf{decide}(x)$ \BlueComment{decide the \metablock{} $x$} \LineLabel{ln:core:on_vote2_quorum_if_undecided}

\fullmedskip
\State \textbf{upon}  $5\Delta$ time has elapsed since entering view $\viewi$: \LineLabel{ln:core:on_timer_header}
\State \hskip2em \textbf{if} $\votetwo[\viewi] = \NONE$: \LineLabel{ln:core:on_timer_if_no_vote2}
\State \hskip4em $\votetwo[\viewi] \gets$ $\langle \VOTETWO, \viewi, \NONE, \guardi[\viewi] \rangle$; $\textbf{broadcast}$ $\votetwo[\viewi]$ \BlueComment{timeout: no lock} \LineLabel{ln:core:on_timer_set_vote2}

\fullmedskip
\State \textcolor{blue}{$\triangleright$ the view change: what is safe to carry into the next view, proposer by proposer}
\State \textbf{upon} $5\Delta$ time has elapsed since entering view $\viewi$ \textbf{and} \textbf{existence of} a set $S$ of $2f + 1$ accepted $\langle \VOTETWO, \viewi, \cdot, \cdot \rangle$ messages whose non-$\NONE$ locks are all identical: \BlueComment{fires once for every such set $S$} \LineLabel{ln:vc:on_vote2_quorum_compute_approved}\LineLabel{ln:rule_view_change_gate}\LineLabel{ln:core:notation_every_witness}
\State \hskip2em \textbf{if} there exists a non-$\NONE$ lock $\ell$ in $S$: \textbf{invoke} $\mathsf{echo}(\viewi, p_j, \ell[j])$ for every proposer $p_j$ \BlueComment{a lock exists: echo it} \LineLabel{ln:vc:on_vote2_quorum_set_approved_lock}
\State \hskip2em \textbf{else: for every proposer} $p_j$: \LineLabel{ln:vc:on_vote2_quorum_set_approved_allsafe}
\State \hskip4em \textbf{if} some $G \in \VALUE \cup \{ \ALLSAFE \}$ is the entry for $p_j$ in the guards of $f + 1$ messages in $S$: \LineLabel{ln:vc:on_vote2_quorum_guard_if}
\State \hskip6em \textbf{invoke} $\mathsf{echo}(\viewi, p_j, G)$ \BlueComment{echo the entry for $p_j$ that $f+1$ guards in $S$ share} \LineLabel{ln:vc:on_vote2_quorum_guard_echo}
\State \hskip4em \textbf{else: invoke} $\mathsf{echo}(\viewi, p_j, \ALLSAFE)$ \BlueComment{the guards for $p_j$ diverge: echo $\ALLSAFE$} \LineLabel{ln:vc:on_vote2_quorum_allsafe_echo}

\fullmedskip
\State \textbf{upon} receipt of $2f + 1$ messages $\langle \ECHO, V, j, x[j] \rangle$ \textbf{for every proposer} $p_j$, for some \metablock{} $x$ and some $V \geq \viewi$: \LineLabel{ln:vc:upon_enter_next_view}
\State \hskip2em $\guardi[V + 1] \gets x$ \LineLabel{ln:vc:upon_enter_set_guard}
\State \hskip2em \textbf{enter} view $V + 1$ \LineLabel{ln:vc:upon_enter_enter}

\fullmedskip
\State \textcolor{blue}{$\triangleright$ catching up: a correct validator has reached view $V$, and $x$ is safe to carry into it}
\State \textbf{upon} receipt of $f + 1$ messages $\langle \ECHO, V, j, x[j] \rangle$ \textbf{for every proposer} $p_j$, for some \metablock{} $x$ and some $V > \viewi$: \LineLabel{ln:vc:upon_catch_up}
\State \hskip2em $\guardi[V] \gets x$ \LineLabel{ln:vc:upon_catch_up_set_guard}
\State \hskip2em \textbf{enter} view $V$ \LineLabel{ln:vc:upon_catch_up_enter}

\fullmedskip
\State \textcolor{blue}{$\triangleright$ amplification, for the current view and the one below it}
\State \textbf{upon} receipt of $f + 1$ messages $\langle \ECHO, V, j, v \rangle$ for some $p_j$, some $v \in \VALUE \cup \{ \ALLSAFE \}$, and some $V \in \{ \viewi - 1, \viewi \}$: \LineLabel{ln:vc:upon_echo_amplification}
\State \hskip2em \textbf{invoke} $\mathsf{echo}(V, p_j, v)$ \LineLabel{ln:vc:upon_echo_amplification_invoke}

\end{algorithmic}
\end{algorithm}

\subsection{Guarantees}\label{sec:coda:model}

In this subsection, we state the correctness and complexity guarantees of \coda{}, which \Cref{app:coda} proves.

\coda satisfies the four properties that Step~6 of \Cref{sec:protocol} requires of the fallback consensus (\Cref{lem:coda:agreement,lem:coda:termination,lem:coda:strong_unanimity,lem:coda:obligation}).
In the terms of \coda, in which a validator inputs a \metablock and decides one, they read as follows.
\begin{compactitem}
   \item \emph{Agreement:} No two correct validators decide different \metablocks.
   \item \emph{Termination:} If every correct validator starts \coda, every correct validator eventually decides.
   \item \emph{Strong unanimity (per proposer):} If all correct validators that input a \metablock input the same value other than $\ALLSAFE$ for a proposer $p_\pid$, no correct validator decides a \metablock whose entry for $p_\pid$ is $\ALLSAFE$.
   \item \emph{Obligation (per proposer):} If a correct validator decides a \metablock whose entry for a proposer $p_\pid$ is a value $v \neq \ALLSAFE$, some correct validator input $v$ for $p_\pid$.
\end{compactitem}

Agreement, strong unanimity, and obligation hold in every execution, whatever the timing, and also when correct validators stop taking part in \coda, possibly before inputting (\Cref{app:coda}).
Termination needs every correct validator to start \coda at some point and none to stop taking part.
To simplify its proof, \Cref{app:coda} adopts the \emph{participation convention} that every correct validator starts \coda by GST.
Under this convention, \Cref{lem:coda:termination} also bounds the time. Let $V$ be the first view greater than $1$ that no correct validator enters before GST and whose leader is correct. Then every correct validator decides within $5\Delta$ of the first entry into $V$ by a correct validator.
In \qchorus, a correct validator starts \coda when it adopts a \metablock (Step~5 of \Cref{sec:protocol}), possibly after GST.
The convention costs no generality: if the last correct validator starts \coda at a time $T_{\mathrm{start}}$ after GST, the termination proofs hold with $T_{\mathrm{start}}$ in place of GST, as \Cref{app:coda} shows.
In \qchorus, a correct validator may also halt before it inputs, in which case it never starts \coda.
Once a correct validator halts, the halting gadget of \Cref{sec:protocol} makes every correct validator that participates finalize (\Cref{clm:qv:halting}).

\subsubsection{Complexity}

Finally, we bound, in terms of the number $k$ of proposers, the \emph{message complexity} of \coda, the number of messages that correct validators send per view, retransmissions aside, and the size of the persistent memory of a correct validator. We count memory and message sizes in \emph{words}: a word holds a value (a root, $\bot$, or $\ALLSAFE$), a view number, or a validator identifier, so a \metablock occupies $k$ words, and a message occupies a constant number of words, for its type and view number, beyond the \metablocks or value that it carries.
For every proposer $p_\pid$, let $P_\pid = \{ \mathit{input}_w[\pid] \mid p_w \text{ is correct and inputs} \}$ denote the set of values that correct validators input for $p_\pid$, and let $c = \max_\pid |P_\pid|$.
In general, $|P_\pid| \leq n$, as each correct validator inputs one value per proposer.
When \qchorus{} invokes \coda{} and the proposer $p_\pid$ did not equivocate, every correct validator inputs for $p_\pid$ either the Merkle root of $p_\pid$'s proposal or $\bot$, so $|P_\pid| \leq 2$.
The per-validator broadcast counts and the persistent-memory bound below are stated in terms of $c$: they do not depend on $n$ when $c = O(1)$, and they grow by a factor of $n$ in the worst case.

\begin{restatable}[Message complexity]{lemma}{codamessages}\label{lem:coda:message_complexity}
For every view $V$, each correct validator initiates at most $3 + (c + 1) k$ broadcasts of messages of view $V$.
Hence, retransmissions aside, correct validators send $O(c k n^2)$ messages of each view, which is $O(k n^2)$ when $c = O(1)$.
\end{restatable}

The same accounting bounds the communication of a correct validator in words: a $\PREPARE$ or a $\VOTEONE$ message carries one \metablock, a $\VOTETWO$ message carries two, and an $\ECHO$ message carries one value, for $O(c k n)$ words per view.
We now turn to memory.

\begin{restatable}[Persistent memory]{lemma}{codamemory}\label{lem:coda:memory_complexity}
Each correct validator stores $O(c k)$ words in persistent memory for \coda, which is $O(k)$ words when $c = O(1)$.
\end{restatable}

\subsection{Differences from Forget-IT}\label{sec:coda:differences}

\coda keeps the structure of Forget-IT's views: a $\PREPARE$ message, two rounds of votes, locks, guards carried in $\VOTETWO$ messages, the acceptance rule for these messages, and a view change that reads their locks and guards.
It differs in five main points.
First, every view of Forget-IT has a leader, and its guards start as a special value $\textsc{safe}$, which means that every value is safe.
\coda adds the leaderless view~$1$, and its guards start as the validators' inputs.
Second, a Forget-IT validator votes for a proposal only if it equals its guard or its guard is $\textsc{safe}$, it may relax its guard to $\textsc{safe}$ within a view, and a leader whose guard is $\textsc{safe}$ proposes its own input.
A \coda validator votes for any justified proposal of the leader, its guard stays fixed within a view, and a leader proposes exactly its guard.
Third, Forget-IT passes the outcome of a view change through three phases of echoes, in which $\textsc{safe}$ counts toward a quorum for every value.
\coda has a single phase of echoes per proposer, amplified as in Bracha's broadcast, in which $\ALLSAFE$ is an ordinary value that counts only for itself.
These changes give \coda strong unanimity and obligation per proposer, whereas Forget-IT guarantees validity only when all processes are correct and the network is synchronous from the start.
Fourth, a Forget-IT validator sends at most two echoes per phase and view, the second one for $\textsc{safe}$.
As $\ALLSAFE$ counts for no other value, a \coda validator may instead echo, for one proposer, $\ALLSAFE$ and every value that correct validators input for it, so its messages and persistent memory grow with the number of such values (\Cref{sec:coda:model}).
Fifth, \coda lacks two entry rules of Forget-IT, by which the next leader enters on $f+1$ echoes of the last phase and the other validators on its proposal and $f+1$ echoes of the second phase, and which let Forget-IT decide within three message delays of a correct leader's entry after GST.
In \coda, the first view with a correct leader that starts after GST decides within $5\Delta$ of the first entry into it (\Cref{lem:coda:termination}).

\section{Integration with the \conductor}\label{sec:conductor}

Within \cadence, the slot consensus runs under the \conductor, which opens the slots.
The \conductor groups slots into windows of consecutive slots with evenly spaced deadlines, and a validator opens the next window only once enough slots of the earlier windows are \emph{complete}, which in \cadence means that the validator has finalized the blocks of these slots; this keeps the number of open slots bounded.
Once per window, the validators agree on the first deadline of the next window through an agreement on a core set (ACS)~\cite{cadence}.
We expect \quorus to keep the algorithm of the \conductor and to change two things around it: the ACS and the point at which a validator completes a slot.
This version of the paper outlines these changes but does not specify them in full or analyze them.
In the model of \Cref{sec:model}, in which validators may crash and channels are unreliable, the \conductor also needs changes similar to those made to the slot consensus: a validator keeps retransmitting the messages that others may still need, such as its \textsc{Finalized} messages and those of the ACS, until the \conductor has moved far enough ahead; we leave these changes to a later version.

\para{A signature-free ACS.}
The ACS can be built without signatures, following Ben-Or, Kelmer, and Rabin~\cite{bkr94}: each validator disseminates its proposal with a reliable broadcast~\cite{bracha87}, and one binary agreement per validator decides whether that validator's proposal enters the set.
For the binary agreements, \coda can be used with a single proposer and the values $0$ and $1$, reading an output of $\ALLSAFE$ as $0$.
An ACS instance would end with the same halting rule as a slot (\Cref{sec:qv:description}), with the set in place of the \metablock: a validator would hand the set to the \conductor when it halts in the instance, and then stop running it.
\para{Completion at halting.}
The analysis of the \conductor in \cadence relies on totality: once a correct validator completes a slot after GST, every correct validator does so shortly after.
In \qchorus, finalizing has no such constant bound, but halting does (\Cref{sec:qv:description}): if every correct validator starts participating by time $T$, every correct validator halts by time $\max(T, H, \mathrm{GST}) + 2\Delta$, where $H$ is the time at which the first correct validator halts.
A validator therefore completes a slot when it halts, and it opens a window when it halts in the ACS instance that fixes the window's first deadline.

This bound counts from the last correct validator's start, however, and a validator that is behind does not participate yet.
To bring it up to date, each validator retransmits one cumulative \textsc{FinalizedUpTo} message, which it replaces every $\Delta$ with an up-to-date one, and keeps only the most advanced one from each peer.
The message refers to a single chain that lists the finalized \metablocks in slot order, with the output of each window's ACS instance inserted where the \conductor may open that window, that is, after the earlier slots that must be complete first.
It names the last point of the chain up to which its sender has finalized every slot and every ACS instance, with a hash over the chain up to that point.
These messages need not name the same point: once $f+1$ validators have named a point or a later one, with hashes that fit a single chain, a validator fetches the chain from them, checks it against the hashes, and finalizes every slot and every ACS instance up to that point, and once $2f+1$ have, it halts in them, which completes these slots and opens these windows.
Totality would then hold within a constant time, whatever a validator's lag, for the prefixes of completed slots that the \conductor counts.

\medskip\noindent
With these changes, we expect the guarantees of the \conductor to carry over: the number of open slots would stay bounded at all times, correct validators would open each slot within a constant time of one another after GST, and, for a constant grace period $\mathcal{G}$, every slot whose starting time is at least $\mathrm{GST} + \mathcal{G}$ would be timely (\Cref{sec:model}).

\section{Evaluation}\label{sec:evaluation}

We compare \qchorus with \chorus in simulation, for 200 validators of equal stake, each with a 1\,Gbps link, whose locations follow the geographic distribution of stake in the Monad validator set on October 1, 2026.
The simulation covers propagation, transmission and queueing delays, but no local computation. We simulate only the fast path, without faults or message loss.
Slots start every 100\,ms, and each carries 500\,KB of proposal data, divided equally among 5, 20 or 100 proposers.
We do not fix a signature scheme and instead vary the signature size up to 5000 bytes.
This range covers the parameter sets of ML-DSA~\cite{fips204} and of Falcon~\cite{falcon}, which NIST has selected for standardization.
It excludes the larger SLH-DSA signatures~\cite{fips205}.
Since safety and liveness need no signatures, we also simulate a signature size of zero, which gives the cost of the protocol without signatures.

\begin{figure}[t]
\centering
\includegraphics[width=\linewidth]{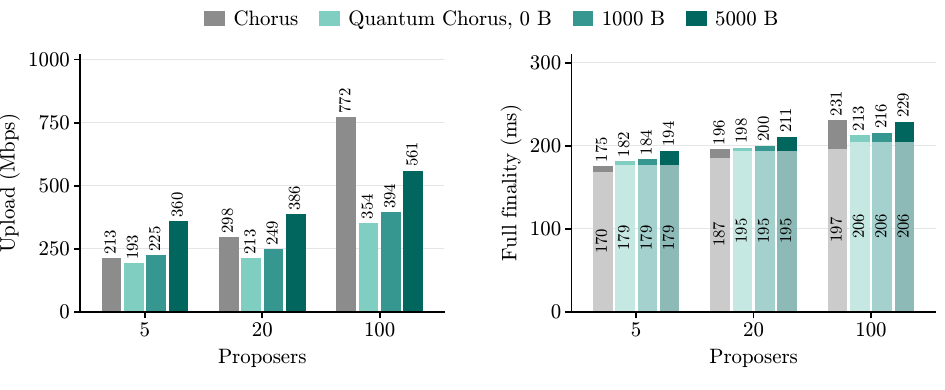}
\caption{Average upload per validator (left) and mean latency to full finality (right), for 200 validators with 1\,Gbps links. \chorus uses BLS signatures, and the signature sizes refer to \qchorus. Each latency bar consists of the latency with unlimited bandwidth (light, value inside the bar) and the increase at 1\,Gbps (dark), with the total above the bar. With unlimited bandwidth, the latency of \qchorus does not depend on the signature size.}
\label{fig:evaluation}
\end{figure}

\Cref{fig:evaluation} shows the average upload of a validator and the mean latency from a proposer's broadcast until a validator finalizes the slot (full finality, as opposed to speculative finality), averaged over all pairs of proposer and validator. Speculative finality comes about 50\,ms earlier: with 20 proposers, for instance, \qchorus reaches it after 145 to 155\,ms on average, depending on the signature size.

In terms of bandwidth, the two protocols differ mainly in the size of their consensus messages. Chunks account for about 180\,Mbps of every validator's upload, in every configuration and for both protocols, and consensus messages for the rest.
Since post-quantum signatures are much larger than BLS signatures, \qchorus with large signatures needs more upload than \chorus when there are few proposers.
However, a \qchorus vote carries one signature, while a \chorus vote carries a signature or certificate for every proposer, so \qchorus scales better with the number of proposers.

Message sizes matter less for latency.
With 1\,Gbps links, they cause transmission and queueing delays that are small relative to the propagation delays: compared with no signatures, even 5000-byte signatures add at most 16\,ms to the latency of \qchorus.
With 5 or 20 proposers, \chorus finalizes up to 19\,ms earlier than \qchorus.
At 100 proposers, the larger messages of \chorus delay it enough that \qchorus finalizes up to 18\,ms earlier.

The two protocols send the same kinds of messages in the same rounds, dissemination followed by two rounds of votes, so one might expect them to have the same latency when bandwidth is unlimited and message sizes do not matter.
In fact, \chorus is then slightly faster (by less than 10\,ms), because its validators broadcast the certificates they form or receive: a validator can cast its fast vote as soon as it receives a \metablock of certificates from another validator, and can finalize as soon as it receives a certificate of $2f+1$ fast votes.
\qchorus forms no compact certificates, so each validator collects the votes itself.
Without this forwarding, the two protocols behave identically from a networking standpoint and have the same latency.

\para{Simulation details.}
Each run covers three overlapping slots, of which we measure the middle one, and we average over 200 trials.
Both protocols use early voting and quorums of 135 validators.
All messages of a validator share its 1\,Gbps uplink, and propagation delays follow RIPE Atlas measurements.
We make optimizations that a deployment would make. Votes do not carry chunks. Instead, as in Deterministic RaptorCast~\cite{mip10}, which \cadence uses in practice~\cite{cadence}, each proposer sends every validator its chunks in separate packets, and validators forward enough of these packets for every validator to decode every proposal. A validator's votes take priority over the chunks it forwards, and, as in the evaluation of \cadence~\cite[Section~8]{cadence}, each proposer starts disseminating at its own lead time before the deadline rather than at the starting time $D - \Delta$, as in \Cref{alg:qv1}.

\section{Related Work}\label{sec:related}

\para{Byzantine consensus.}
Byzantine consensus asks $n$ processes, of which at most $f$ may deviate arbitrarily from the protocol, to agree on a common value~\cite{psl80,lsp82}.
Fischer, Lynch, and Paterson showed that no deterministic protocol solves it under asynchrony, even if only one process may crash~\cite{flp85}.
Deterministic protocols do solve it under synchrony~\cite{ds83} and under partial synchrony~\cite{dls88}, the model of \cadence and of this paper, and randomized protocols solve it under asynchrony~\cite{benor83,rabin83}.
\qchorus and \coda are deterministic.

\para{Byzantine consensus without signatures.}
Protocols that use no cryptography beyond authenticated channels, called information-theoretic or unauthenticated, tolerate a computationally unbounded adversary.
Their safety and liveness thus rest on no signature that a quantum computer could forge.
Since a validator cannot prove a quorum to another without signatures, these protocols often use the amplification of Bracha's reliable broadcast~\cite{bracha87}: a validator that receives $f+1$ matching messages repeats the message, and one that receives $2f+1$ acts on it.
\coda uses this rule for the echoes of its view change (\Cref{sec:coda:description}), and \qchorus for the \textsc{Finalized} messages of its halting gadget (\Cref{sec:qv:description}).
Among the information-theoretic protocols for partial synchrony, IT-HotStuff~\cite{iths20} and TetraBFT~\cite{tetrabft24} keep persistent memory constant, send $O(n^2)$ messages per view, and are optimistically responsive.
Their good-case latency, the latency when the network is synchronous from the start and the leader of the first view is correct, is six and five message delays, respectively.
Forget-IT~\cite{forgetit} reduces it to three, which is optimal~\cite{anrx21}.
\coda builds on Forget-IT (\Cref{sec:coda:differences}).
Simple-IT~\cite{simpleit26} replaces quorum certificates by reliable broadcast and a reliable notification primitive.
With Bracha's reliable broadcast, its good-case latency is four message delays, one more than that of Forget-IT, but its design allows pipelining.

\para{Consensus with post-quantum signatures.}
Protocols can keep their certificates with post-quantum signatures, but at a cost.
Post-quantum signatures can be aggregated, although the aggregates are much larger than a BLS aggregate signature and take much longer to compute~\cite{dqs26,aerie26}.
Certificates can also stay lists of ordinary signatures: in HotStuff-style protocols~\cite{hotstuff19}, the leader collects the votes and disseminates the certificate together with its next proposal, for instance with AVID~\cite{avid05}.
The drawback of these protocols is latency: every round of votes passes through the leader, so a decision typically takes more message delays than in protocols in which validators send their votes to all.
The distributed quorum signature (DQS)~\cite{dqs26} instead stops sending certificates: it turns them from network messages into local events, weak certificates for safety and strong ones for liveness, and is built from ordinary signatures and a Bracha-style broadcast of approvals.
It takes the place of quorum signatures in an existing protocol: a validator that creates a weak certificate broadcasts an approval for it.
Quantumglow, the post-quantum version of Alpenglow, follows this approach~\cite[Section~5]{alpenglow}.

\section*{Use of AI Tools}
The protocols and the main ideas of this paper are our own. We used large language models to help draft and revise parts of the text and the simulator. The authors take full responsibility for the content of this paper.

\bibliographystyle{splncs04}
\bibliography{references}

\clearpage
\appendix
\crefalias{section}{appendix}\crefalias{subsection}{subappendix}

\section{Proofs for \qchorus}\label{app:slot}

This appendix proves the lemmas of \Cref{sec:protocol}.

\subsection{Core Properties}

A correct validator finalizes on the fast path (line~\ref{ln:qv:fastdecide}), on the fallback path (line~\ref{ln:qv:fallbackdecide}), or through the halting gadget (lines~\ref{ln:qv:echo} and~\ref{ln:qv:halt}).
We first show that a finalization through the halting gadget repeats a finalization on one of the two paths.

\begin{claim}[Finalizations through the halting gadget]\label{clm:qv:gadget}
If a correct validator finalizes a \metablock $x$ through the halting gadget at time $t$, then some correct validator finalized $x$ on the fast path or on the fallback path by time $t$.
\end{claim}

\begin{proof}
A correct validator broadcasts $\langle \textsc{Finalized}, x \rangle$ exactly when it finalizes $x$.
Finalizing $x$ through the halting gadget requires $\langle \textsc{Finalized}, x \rangle$ messages from $f+1$ validators, at least one of them correct, which therefore finalized $x$ earlier.
Hence, the first correct validator to finalize $x$ did so on the fast path or on the fallback path, and if a correct validator finalizes $x$ through the halting gadget at time $t$, this first finalization happened by time $t$.
\qed
\end{proof}

The proofs below compare finalized \metablocks entry by entry.
A speculative or fast-path finalization does not take the values of the fast entries as they are: before finalizing, a validator replaces a root whose proposal does not re-encode to it by $\bot$ (lines~\ref{ln:qv:speccheck} and~\ref{ln:qv:fastcheck}).
For a value $v$, a root or $\bot$, we call the result of this check the \emph{checked value} of $v$: $\bot$ if $v$ is a root whose proposal does not re-encode to $v$, and $v$ otherwise.
It does not depend on which chunks are used (see finalizing on the fast path in \Cref{sec:protocol}).

\qvagreement*
\begin{proof}
It suffices to show that correct validators finalize the same \metablock, as a \metablock determines its block.
Fix a proposer $\proc_\pid$, and consider two correct validators finalizing \metablocks $x$ and $y$.
By \Cref{clm:qv:gadget}, some correct validator also finalizes $x$ on the fast path or on the fallback path, and so does one for $y$; we may therefore assume that both finalizations are made on these two paths.
To prove the lemma, we show that  $x[\pid] = y[\pid]$, going through the three possible combinations of paths.
\begin{compactitem}
   \item If both validators finalize on the fast path, quorum intersection settles the matter: there are $2f+1$ path votes whose entry for $\proc_\pid$ is  a fast entry for some value $v$ whose checked value is $x[\pid]$, and $2f+1$ whose entry for $\proc_\pid$ is  a fast entry for some value $w$ whose checked value is $y[\pid]$.
   The two sets of senders intersect in a correct validator, which issued a single path vote (\Cref{lem:qv:complexity}), and thus a single entry for $\proc_\pid$.
   Hence, $v = w$, and $x[\pid] = y[\pid]$.

   \item If both validators finalize on the fallback path,  the agreement property of the fallback consensus gives them the same output. Both finalize this output after replacing its entries $\ALLSAFE$ by $\bot$ (line~\ref{ln:qv:fallbackdecide}), so $x = y$.

   \item The interesting case is the mixed one: without loss of generality, say $x$ is finalized on the fast path and $y$ on the fallback path.
   As $x$ is finalized on the fast path, there exist $2f+1$ path votes whose entry for $\proc_\pid$ is  a fast entry for some value $v$ whose checked value is $x[\pid]$.
   By quorum intersection, any clean set for $\proc_\pid$, being a set of $2f+1$ accepted entries for $\proc_\pid$ from different validators, contains at least one of these fast entries: the two sets of senders intersect in a correct validator, whose unique entry for $\proc_\pid$ is a fast entry for  $v$.
   The adoption rule then leaves little choice.
   If  $v = \bot$, every correct validator adopts $\bot$ for $\proc_\pid$.
   If  $v$ is a root $\mroot$, every correct validator decodes the proposal from the $f+1$ chunks that corroborated the fast entry and re-encodes it. As the result of the check does not depend on the chunks used, either every correct validator adopts $\mroot$ (if the check passes) or every correct validator adopts $\bot$ (if it fails).
   In all cases, all correct validators that input a \metablock{} input the same value, namely the checked value of $v$, for $\proc_\pid$.
   Strong unanimity now ensures that no correct validator outputs $\ALLSAFE$ for $\proc_\pid$, while obligation ensures that the output value for $\proc_\pid$ was input by a correct validator. Hence,  $y[\pid]$ is the checked value of $v$, that is, $y[\pid] = x[\pid]$.
\end{compactitem}
In every case,  $x[\pid] = y[\pid]$: the two validators finalize the same \metablock, and therefore the same block, which concludes the proof.
\qed
\end{proof}

\qvspeculation*
\begin{proof}
Suppose a correct validator $\proc_s$ speculatively finalizes a \metablock $x$, and fix a proposer $\proc_\pid$. Let $v$ be the value of the fast entry of $\proc_s$ for $\proc_\pid$, a Merkle root or $\bot$, and let $\hat{v}$ be its checked value, so that $x[\pid] = \hat{v}$ (line~\ref{ln:qv:speculate}).
Throughout the proof, we say that a \textsc{Vote} message \emph{matches} $v$ if its entry for $\proc_\pid$ is $\langle \pid, v \rangle$: a positive entry for $v$ when $v$ is a root, and a negative entry when $v = \bot$.
The speculation rests on $2f+1$ \textsc{Vote} messages matching $v$, sent by $2f+1$ distinct validators, and $\proc_s$ holds all of them (lines~\ref{ln:qv:fastroot}, \ref{ln:qv:fastbot}, and~\ref{ln:qv:speculate}).

A finalization through the halting gadget repeats one on the fast path or on the fallback path (\Cref{clm:qv:gadget}), so it suffices to consider these two paths.

Observe first that every fast-path finalization rests, for $\proc_\pid$, on fast entries for $v$, and thus finalizes $\hat{v}$ for $\proc_\pid$.
Namely, a fast-path finalization rests on $2f+1$ \textsc{PathVote}s whose entries for $\proc_\pid$ are fast entries for the same value, and at least one of them comes from a correct validator. That validator issued its fast entry upon receiving $2f+1$ \textsc{Vote} messages matching the entry's value, while $\proc_s$ holds $2f+1$ \textsc{Vote} messages matching $v$. The two sets of senders intersect in a correct validator, which sent a single \textsc{Vote} (\Cref{lem:qv:complexity}), so the fast entry is for $v$.
A reversion, therefore, can only originate from the fallback path.

On the fallback path, the proof relies on signatures. A \textsc{Vote} message is signed by its sender as a whole: it carries a single signature over the entire vector of entries.
Consequently, two \textsc{Vote} messages from the same validator whose entries for $\proc_\pid$ differ are two distinct messages signed by that validator, which is a transferable proof of its misbehavior.
In particular, if $f+1$ distinct validators send \textsc{Vote} messages that do \emph{not} match $v$, then, as $(f+1) + (2f+1) = 3f+2 > n$, one of them also sent one of the $2f+1$ \textsc{Vote}s matching $v$ that $\proc_s$ holds, and whoever holds the two messages holds the proof.
The argument below produces exactly such a pair of sets: the $2f+1$ matching \textsc{Vote}s held by $\proc_s$, and $f+1$ non-matching \textsc{Vote}s held by some correct validator.

We claim that, short of provable misbehavior, no correct validator inputs to the fallback consensus a value $y$ for $\proc_\pid$ with $y \neq \hat{v}$.
Strong unanimity and obligation, applied to proposer $\proc_\pid$, then ensure that the fallback consensus outputs $\hat{v}$, and the lemma follows: every correct validator that inputs a \metablock inputs $\hat{v}$ for $\proc_\pid$, so strong unanimity rules out $\ALLSAFE$, and obligation gives $\hat{v}$. As $\proc_\pid$ is arbitrary and $x[\pid] = \hat{v}$, every \metablock that a correct validator finalizes is $x$ unless misbehavior can be proven, and its block is then the block that $\proc_s$ speculatively finalized.

To prove this claim, suppose that a correct validator $\proc_i$ inputs a value $y$ for $\proc_\pid$ with $y \neq \hat{v}$, and recall the three ways in which it may have adopted $y$ (line~\ref{ln:qv:adopt}).
This implies that $y \neq v$: a root that a correct validator inputs passes the check, as the paragraph on data availability shows, so $y = v$ would give $\hat{v} = v = y$.
In each of the three scenarios below, we show that some correct validator holds $f+1$ \textsc{Vote} messages that do not match $v$, which gives the proof of misbehavior described above.
\begin{compactitem}
   \item In the first scenario, $\proc_i$'s clean set for $\proc_\pid$ holds a \emph{fast entry} for $y$, and $y$ is either $\bot$ or a root that passed the decode--re-encode check.
   For that fast entry to enter the clean set, it had to be accepted: $\proc_i$ received $f+1$ \textsc{Vote} messages matching $y$ (line~\ref{ln:qv:accept}).
   As $y \neq v$, these messages do not match $v$, and $\proc_i$ holds them.

   \item In the second scenario, $\proc_i$'s clean set holds a fast entry for a root $\mroot$ that \emph{failed} the check, so $\proc_i$ adopted $y = \bot$.
   If $\mroot \neq v$, the $f+1$ \textsc{Vote}s that corroborated the fast entry do not match $v$, and $\proc_i$ holds them.
   If $\mroot = v$, then $v$ fails the check as well, as its result does not depend on the chunks used, so $\hat{v} = \bot = y$, which contradicts $y \neq \hat{v}$.

   \item In the third scenario, $\proc_i$'s clean set consists of \emph{fallback entries} alone.
   If $y$ is a root, the clean set holds $f+1$ fallback entries for $y$, at least one issued by a correct validator $\proc_w$, which issued it only upon receiving $f+1$ \textsc{Vote} messages with a positive entry for $y$ (line~\ref{ln:qv:decode}).
   As $y \neq v$, these $f+1$ messages do not match $v$, and $\proc_w$ holds them.

   If $y = \bot$, then $\hat{v} \neq \bot$: $v$ is a root $\mroot$ with a valid encoding.
   Moreover, fewer than $f+1$ fallback entries of the clean set carry $\mroot$, for otherwise $\proc_i$ would have adopted $\mroot$.
   The clean set holds $2f+1$ entries from distinct validators, of which at least $f+1$ come from correct validators, so some correct validator $\proc_w$ issued a fallback entry for $\proc_\pid$ that does not carry $\mroot$.
   Consider the \textsc{Vote} messages that $\proc_w$ held when it issued its \textsc{PathVote}, at least $2f+1$ of them (line~\ref{ln:qv:pathvote}). We claim that at least $f+1$ of them are not positive for $\mroot$.
   If they hold fewer than $f+1$ positive entries for $\mroot$, the claim is immediate.
   Otherwise, they hold $f+1$ positive entries for $\mroot$.
   As the encoding is valid, decoding them would have yielded a fallback entry for $\mroot$, which $\proc_w$ did not issue. The rule (line~\ref{ln:qv:decode}) allows this only if $\proc_w$ decoded another root instead, that is, only if some $\mroot' \neq \mroot$ also gathered $f+1$ positive entries at $\proc_w$.
   Those $f+1$ messages, positive for $\mroot'$, are not positive for $\mroot$.
   In both cases, $\proc_w$ holds $f+1$ \textsc{Vote} messages that do not match $v$, as required.
\end{compactitem}
In every case that does not contradict $y \neq \hat{v}$, a correct validator holds $f+1$ signed \textsc{Vote} messages that do not match $v$, while $\proc_s$ holds $2f+1$ signed \textsc{Vote} messages that do.
Both sets remain in persistent memory across crashes.
Together, the two sets expose a validator that signed two different \textsc{Vote} messages, which is precisely the misbehavior the lemma promises to prove.
\qed
\end{proof}

\qvcensorship*
\begin{proof}
Fix a correct proposer $p_j$.
Let $\mroot_\pid$ denote the Merkle root of $\proc_\pid$'s proposal.
Every correct validator receives its chunk, with a valid Merkle proof under $\mroot_\pid$ and a header signed by $\proc_\pid$, by the deadline, while it participates: as the slot is timely, $\proc_\pid$ initiates its send upon starting, by time $D - \Delta \geq \mathrm{GST}$ (line~\ref{ln:qv:disseminate}), and never abandons it, and every correct validator participates from $D - \Delta$ on, so the delay bound of \Cref{sec:model} delivers the chunk during $[D - \Delta, D]$, and the validator records it (line~\ref{ln:qv:chunk}) unless it has halted by then.
As channels are authenticated and $\proc_\pid$ sends each validator a single chunk, every \textsc{Vote} that a correct validator broadcasts thus includes a positive entry $\langle \pid, \mroot_\pid \rangle$. Hence, every fast entry for $\proc_\pid$ that a correct validator issues names $\mroot_\pid$: it rests on $2f+1$ \textsc{Vote}s with matching entries for $\proc_\pid$, at least $f+1$ of them from correct validators.
Moreover, as $\proc_\pid$ is correct, its encoding is valid: any $f+1$ chunks under $\mroot_\pid$ decode to its proposal, which re-encodes to $\mroot_\pid$.
It suffices to show that every \metablock $x$ that a correct validator finalizes satisfies $x[\pid] = \mroot_\pid$, as the block of such a \metablock contains the proposal of $\proc_\pid$.

Consider first a correct validator finalizing a \metablock $x$ on the fast path.
It holds $2f+1$ \textsc{PathVote}s whose entry for $\proc_\pid$ is  a fast entry for some value $v$ whose checked value is $x[\pid]$. At least one of them comes from a correct validator, whose fast entry for $\proc_\pid$ names $\mroot_\pid$.
Hence,  $v = \mroot_\pid$, and, as the encoding is valid, $x[\pid] = \mroot_\pid$.

The fallback path is no more accommodating.
First, no correct validator ever accepts a fast entry for $\proc_\pid$ naming a value $y \neq \mroot_\pid$: acceptance requires $f+1$ \textsc{Vote}s whose entry for $\proc_\pid$ matches $y$, whereas every correct validator's entry for $\proc_\pid$ is positive for $\mroot_\pid$, leaving at most $f$ such \textsc{Vote}s.
Second, every fallback entry for $\proc_\pid$ that a correct validator issues names $\mroot_\pid$: the validator holds at least $2f+1$ \textsc{Vote}s, at least $f+1$ of them from correct validators and thus positive for $\mroot_\pid$, and only the at most $f$ Byzantine validators can send \textsc{Vote}s that are not positive for $\mroot_\pid$, whether negative or positive for another root (a \textsc{Vote} carries no header signed by the proposer, so the checks of line~\ref{ln:qv:allornothing} allow any root), so no other root reaches $f+1$ positive entries. It therefore decodes the proposal from chunks under $\mroot_\pid$, and the proposal re-encodes to $\mroot_\pid$.
Consequently, every set of $2f+1$ accepted entries for $\proc_\pid$ from different validators is clean, its fast entries name $\mroot_\pid$, and at least $f+1$ of its entries, those from correct validators, are fast or fallback entries for $\mroot_\pid$.
The adoption rule therefore returns $\mroot_\pid$ for every correct validator's clean set for $\proc_\pid$: if the set holds a fast entry, the entry names $\mroot_\pid$, and the decode--re-encode check passes, the encoding being valid; otherwise, at least $f+1$ of its fallback entries, those from correct validators, name $\mroot_\pid$, and only the at most $f$ entries from Byzantine validators can name another root or $\bot$, so no other root reaches $f+1$ entries.
Every correct validator thus adopts $\mroot_\pid$ for $\proc_\pid$, and inputs a \metablock whose entry for $\proc_\pid$ is $\mroot_\pid$, if it inputs one.
Strong unanimity for proposer $\proc_\pid$ ensures that no correct validator outputs $\ALLSAFE$ for $\proc_\pid$, and obligation ensures that the output value for $\proc_\pid$ was input by a correct validator. Hence, the output \metablock $x$ satisfies $x[\pid] = \mroot_\pid$.
A finalization through the halting gadget repeats one on the two paths (\Cref{clm:qv:gadget}), so its \metablock satisfies this as well.

In both cases, every \metablock $x$ that a correct validator finalizes satisfies $x[\pid] = \mroot_\pid$, and its block contains the proposal of $\proc_\pid$.
\qed
\end{proof}

Once a correct validator halts, the halting gadget makes every correct validator that participates finalize; the proofs of termination and accountability below use this.

\begin{claim}[Finalization after a halt]\label{clm:qv:halting}
If a correct validator halts at time $H$, then every correct validator that starts participating at time $S$ finalizes by time $\max(S, H, \mathrm{GST}) + \Delta$.
\end{claim}

\begin{proof}
By time $H$, the halting validator received $\langle \textsc{Finalized}, x \rangle$ messages, for some \metablock $x$, from $2f+1$ validators, at least $f+1$ of them correct (line~\ref{ln:qv:halt}).
These correct validators initiated their broadcasts by time $H$ and never abandon them. By the delay bound of \Cref{sec:model} with $\tau = \max(S, H)$, a correct validator that starts participating at time $S$ therefore receives their messages during $[\max(S, H, \mathrm{GST}), \max(S, H, \mathrm{GST}) + \Delta]$.
This interval starts no earlier than $S$, so the validator participates when the messages arrive, unless it has halted.
A correct validator that has halted by then has finalized, as it finalizes when it halts.
Any other correct validator keeps the messages in persistent memory, and thus holds $f+1$ of them by time $\max(S, H, \mathrm{GST}) + \Delta$ and finalizes by then (line~\ref{ln:qv:echo}).
\qed
\end{proof}

\qvtermination*
\begin{proof}
 If some correct validator halts, every correct validator finalizes by \Cref{clm:qv:halting}, as it eventually participates. Otherwise, no correct validator halts, so that all of them keep taking part in the fallback consensus, and it suffices to prove that every correct validator eventually inputs a \metablock to the fallback consensus: from there, the termination of the fallback consensus carries the day, as a correct validator finalizes the output \metablock (line~\ref{ln:qv:fallbackdecide}).
We prove this below.

Every correct validator eventually participates and broadcasts a \textsc{Vote} (line~\ref{ln:qv:vote}), so every correct validator eventually holds $2f+1$ \textsc{Vote}s and, once time reaches $D + \Delta$ (or earlier, see Step~3), forms an entry for every proposer and broadcasts its \textsc{PathVote} (line~\ref{ln:qv:pathvote}).
Fix a proposer $\proc_\pid$.
Crucially, if a correct validator issues a fast entry for $\proc_\pid$, then every correct validator eventually accepts it.
Namely, the issuer formed its fast entry upon receiving $2f+1$ \textsc{Vote}s whose entry for $\proc_\pid$ matches the fast entry, at least $f+1$ of which come from correct validators, and correct validators keep their \textsc{Vote} broadcasts forever-alive.
The $f+1$ corroborating \textsc{Vote}s thus eventually reach every correct validator, and the fast entry is accepted (line~\ref{ln:qv:accept}).
Fallback entries, recall, are accepted the moment they are received.
Consequently, every correct validator eventually accepts the entries for $\proc_\pid$ of \emph{all} correct validators, of which there are at least $2f+1$.

Once that happens, if not before, the accepted entries for $\proc_\pid$ contain a clean set.
Indeed, no two correct validators issue fast entries for $\proc_\pid$ naming different values: each rests on $2f+1$ \textsc{Vote}s with matching entries for $\proc_\pid$, and two such sets of \textsc{Vote}s intersect in a correct validator, which sent a single \textsc{Vote} and thus a single entry for $\proc_\pid$.
The entries of the correct validators for $\proc_\pid$ thus form a set of $2f+1$ accepted entries whose fast entries all agree: a clean set.
Hence, every correct validator eventually holds a clean set for every proposer, and, once time has reached $D + 2\Delta$, adopts a value for each of them (line~\ref{ln:qv:adopt}).
The adoption rule is a local computation that always completes: where it calls for decoding, the $f+1$ chunks it needs arrived with the \textsc{Vote}s that corroborated the fast entry in question.
Every correct validator therefore inputs a \metablock to the fallback consensus, which is all we needed.
\qed
\end{proof}

\qveventualhalting*
\begin{proof}
For the time bound, suppose that every correct validator has finalized by time $T$.
All correct validators finalize the same \metablock $x$ (see the proof of \Cref{lem:qv:agreement}), and each of them broadcasts $\langle \textsc{Finalized}, x \rangle$ when it finalizes.
By time $T$, the at least $2f+1$ correct validators have thus initiated these broadcasts, which are forever-alive and so reach every correct validator by time $\max(T, \mathrm{GST}) + \Delta$.
A correct validator that has not halted by then holds $\langle \textsc{Finalized}, x \rangle$ from $2f+1$ validators, participates, as it has finalized, and halts by that time (line~\ref{ln:qv:halt}).

For the first statement, if every correct validator starts participating in the slot, every correct validator eventually finalizes (\Cref{lem:qv:termination}), so there is a time $T$ by which all of them have, and the bound applies.
\qed
\end{proof}

\subsection{Complexity}

\qvcomplexity*
\begin{proof}
We inspect the three message types (\Cref{alg:qv1}).
A \textsc{Chunk} carries the proposer's signature on its header $\langle \pid, \mroot_\pid \rangle$, together with a symbol and a Merkle proof.
A \textsc{Vote} carries the sender's single signature over its entire vector of entries.
Its entries are unsigned, and the chunks attached to positive entries consist of a root, a symbol, and a Merkle proof: the proposer's signed header, which the sender verified upon receiving the chunk, is stored locally and never forwarded (line~\ref{ln:qv:vote}).
A \textsc{PathVote} only carries the sender's single signature over its entire vector of entries (line~\ref{ln:qv:pathvote}).
As for the number of messages, Step~1 runs at most once at each proposer and sends one \textsc{Chunk} to each validator (line~\ref{ln:qv:disseminate}). Step~2 runs at most once and broadcasts one \textsc{Vote} (line~\ref{ln:qv:vote}). Step~3 runs at most once and broadcasts one \textsc{PathVote} (line~\ref{ln:qv:pathvote}).
No rule initiates any of these sends or broadcasts a second time, also after recovery. Retransmissions repeat a send or broadcast that was already initiated, and the counts exclude them. The halting gadget adds at most one \textsc{Finalized} message per correct validator, which carries no signature. \qed
\end{proof}

\subsection{Accountability}

\qvacchalf*
\begin{proof}
As the slot is timely, every correct validator starts participating in it by time $D - \Delta \geq \mathrm{GST}$, and none crashes from then on (\Cref{sec:model}).
Let $t \leq \lfloor f/2 \rfloor$ denote the number of Byzantine validators, so that at least $n - t \geq 2f + 1 + \lceil f/2 \rceil$ validators are correct.

We first show that halting does not interfere with the argument below.
If some correct validator halts by time $D + \Delta$, every correct validator finalizes by $D + 2\Delta$ (\Cref{clm:qv:halting}, as $\mathrm{GST} \leq D - \Delta$ and every correct validator starts participating by $D - \Delta$), and there is nothing to prove. Otherwise, no correct validator halts by $D + \Delta$, and every correct validator broadcasts its \textsc{PathVote} by then, as used below. A correct validator that halts later has finalized before halting; in particular, a correct validator that has not finalized by $D + 2\Delta$ has not halted by then.

Next, we show that some correct validator $\proc_w$ issued a fallback entry for some proposer $\proc_\pid$.
Let $\proc_i$ be a correct validator that has not finalized by time $D + 2\Delta$, and hence has not halted by then.
As the slot is timely, $\proc_i$ has received, by time $D + 2\Delta$, the \textsc{PathVote}s of all correct validators, broadcast by time $D + \Delta$.
If, for every proposer, the entries of all correct validators were fast entries naming the same value, $\proc_i$ would hold, for every proposer, $n - t \geq 2f + 1$ \textsc{PathVote}s with matching fast entries, and would finalize by time $D + 2\Delta$ (line~\ref{ln:qv:fastdecide}), the chunks it needs having arrived with the \textsc{Vote}s by $D + \Delta$, contrary to its choice.
Hence, for some proposer $\proc_\pid$, either two correct validators issued fast entries for $\proc_\pid$ naming different values, or some correct validator issued a fallback entry for $\proc_\pid$.
The former is impossible: each of the two fast entries rests on $2f + 1$ \textsc{Vote}s whose entries for $\proc_\pid$ match its value, the two sets of senders share at least $(2f+1) + (2f+1) - n = f + 1 > t$ validators, hence a correct one, and a correct validator sends a single \textsc{Vote}, and thus a single entry for $\proc_\pid$.
Therefore, some correct validator $\proc_w$ issued a fallback entry for $\proc_\pid$.

We now examine what $\proc_w$ saw.
As $\proc_w$ issued a fallback entry, it formed its path vote at time $D + \Delta$, since a path vote formed earlier holds fast entries only.
As the slot is timely, by time $D + \Delta$, $\proc_w$ has received the \textsc{Vote}s of all $n - t$ correct validators, broadcast by time $D$.
Let $P$ denote the number of these \textsc{Vote}s whose entry for $\proc_\pid$ is positive, and $N$ the number whose entry for $\proc_\pid$ is negative, so that $P + N = n - t \geq 2f + 1 + \lceil f/2 \rceil$.
As $\proc_w$ issued a fallback entry for $\proc_\pid$, it observed fewer than $2f + 1$ negative entries for $\proc_\pid$, and, for every root $\mroot$, fewer than $2f + 1$ positive entries naming $\mroot$ (lines~\ref{ln:qv:fastroot} and~\ref{ln:qv:fastbot}).
In particular, $N \leq 2f$, hence $P \geq \lceil f/2 \rceil + 1 \geq 1$: some correct validator formed a positive entry for $\proc_\pid$, which it does only upon receiving a \textsc{Chunk} whose header is signed by $\proc_\pid$ (line~\ref{ln:qv:chunk}), and which it stores.
The correct validators thus hold a \textsc{Chunk} signed by $\proc_\pid$.
Two cases remain.
\begin{compactitem}
   \item If the positive entries of correct validators for $\proc_\pid$ name two different roots, then two correct validators hold \textsc{Chunk} messages signed by $\proc_\pid$ with different roots: this is item~(1).

   \item Otherwise, all $P$ positive entries name the same root $\mroot$, so $P \leq 2f$, and $N = (n - t) - P \geq 2f + 1 + \lceil f/2 \rceil - 2f = \lceil f/2 \rceil + 1 \geq \lfloor f/2 \rfloor + 1$.
   These $N$ \textsc{Vote}s come from different correct validators, are signed by their senders, and $\proc_w$ holds them, as it received the \textsc{Vote}s of all correct validators by $D + \Delta$: together with the \textsc{Chunk} above, this is item~(2). \qed
\end{compactitem}
\end{proof}

\qvaccf*
\begin{proof}
As the slot is timely, every correct validator starts participating in it by time $D - \Delta \geq \mathrm{GST}$, and none crashes from then on (\Cref{sec:model}).
At least $2f+1$ validators are correct.
We first show that halting does not interfere with the argument below.
If some correct validator halts by time $D + 3\Delta$, every correct validator finalizes by $D + 4\Delta$ (\Cref{clm:qv:halting}, as $\mathrm{GST} \leq D - \Delta$ and every correct validator starts participating by $D - \Delta$), and there is nothing to prove. Otherwise, no correct validator halts by $D + 3\Delta$, and the broadcasts below are made by then. A correct validator that halts later has finalized before halting; in particular, a correct validator that has not finalized by $D + 4\Delta$ has not halted by then.

Next, we pin down the timeline.
Every correct validator broadcasts its \textsc{Vote} by time $D$ and its \textsc{PathVote} by time $D + \Delta$, and every correct validator receives them by times $D + \Delta$ and $D + 2\Delta$, respectively.
By time $D + 2\Delta$, moreover, every correct validator has accepted every entry of every correct \textsc{PathVote}: fallback entries are accepted on receipt, and a correct fast entry is corroborated by $f+1$ correct \textsc{Vote}s that arrived by $D + \Delta$.
As in the proof of termination, the entries of the correct validators form a clean set for every proposer, so every correct validator adopts a \metablock and inputs it to the fallback consensus at time $D + 2\Delta$ (line~\ref{ln:qv:input}). The pre-fallback broadcasts are forever-alive, and the view-$1$ broadcasts of \coda remain alive through $D + 4\Delta$: as all correct validators enter view~$1$ at time $D + 2\Delta$, none enters a higher view before $D + 7\Delta$ (by the proof of \Cref{clm:coda:view_floor}, which applies unchanged when validators halt), and a validator abandons a broadcast of \coda only upon entering a higher view (line~\ref{ln:rule_abandon}), not when it decides or halts.

Suppose that all correct validators input the same \metablock $a$.
Upon entering the leaderless view~$1$ at time $D + 2\Delta$, every correct validator broadcasts $\langle \VOTEONE, 1, a \rangle$ (\Cref{alg:fallback}).
By time $D + 3\Delta$, every correct validator has received $2f+1$ such messages, locks $a$, and broadcasts $\langle \VOTETWO, 1, a, \cdot \rangle$.
By time $D + 4\Delta$, every correct validator has received $2f+1$ $\VOTETWO$ messages with lock $a$ (line~\ref{ln:core:on_vote2_quorum_header}), and finalizes.
Consequently, if some correct validator has not finalized by time $D + 4\Delta$, the inputs of the correct validators are not all identical: two correct validators $\proc_i$ and $\proc_k$ adopted different values for some proposer $\proc_\pid$.
We fix this $\proc_\pid$ and distinguish two cases, according to the fast entries for $\proc_\pid$ issued by correct validators.
Recall that all such fast entries name the same value $v$: each rests on $2f+1$ \textsc{Vote}s matching it, and two such sets of \textsc{Vote}s intersect in a correct validator, which sent a single \textsc{Vote}.

\begin{compactitem}
   \item \emph{No correct validator issued a fast entry for $\proc_\pid$.}
   Then all $2f+1$ or more correct validators issued fallback entries for $\proc_\pid$, and, by time $D + 2\Delta$, every correct validator holds their signed \textsc{PathVote}s: the conclusion of the lemma.

   \item \emph{Some correct validator $\proc_c$ issued a fast entry for $\proc_\pid$, naming $v$.}
   If the clean set of $\proc_i$ or of $\proc_k$ for $\proc_\pid$ consists of fallback entries alone, that validator holds $2f+1$ signed \textsc{PathVote}s, from different validators, with fallback entries for $\proc_\pid$: the conclusion of the lemma.
   Otherwise, both clean sets hold fast entries, naming values $u_i$ and $u_k$.
   Adopting from a fast entry is a deterministic function of the value it names (line~\ref{ln:qv:adopt}): $\bot$ is adopted as $\bot$, and a root is adopted as itself if the decode--re-encode check passes and as $\bot$ if it fails, where the result of the check depends only on the root: every $f+1$ chunks under a valid encoding decode to the same proposal, and no chunks under an invalid one re-encode to it.
   Had $u_i = u_k$, the two validators would therefore have adopted the same value, decoding failure or not. Hence $u_i \neq u_k$, and at least one of them, say $u_i$, differs from $v$.
   The fast entry for $u_i$ was accepted by $\proc_i$, which therefore holds $f+1$ \textsc{Vote}s, from different validators, matching $u_i$ (line~\ref{ln:qv:accept}). Meanwhile, $\proc_c$ holds $2f+1$ \textsc{Vote}s, from different validators, matching $v$.
   As $(f+1) + (2f+1) = 3f+2 > n$, some validator sent both a \textsc{Vote} matching $u_i$ and a \textsc{Vote} matching $v$, that is, two \textsc{Vote}s with different entries for $\proc_\pid$. Validators $\proc_i$ and $\proc_c$ together hold them: two \textsc{Vote}s signed by the same validator that differ, a transferable proof of that validator's misbehavior. \qed
\end{compactitem}
\end{proof}

\section{Proofs for \coda}\label{app:coda}
This appendix proves the guarantees of \coda stated in \Cref{sec:coda:model}, under the model and conventions of \Cref{sec:coda:description}.

\subsection{Core Properties}

The proofs of agreement, strong unanimity, and obligation allow correct validators to stop taking part in \coda (\Cref{sec:coda:description}). They reason only about the events that correct validators cause: the messages they broadcast and accept, the views they enter, the variables they write, and their decisions. They trace each such event to the rule that caused it and to messages received, or events that occurred, before it, and none of them uses that a correct validator keeps running, which only the termination proofs assume. A validator that stops taking part causes no events from then on: it executes no handler, and its retransmissions repeat messages it broadcast before.

We start by proving that \coda satisfies agreement.

\paragraph{Agreement.}

We begin with a few structural observations about \coda, and then prove agreement.
The first observation is that the current view of a correct validator only grows.

\begin{claim} \label{clm:coda:monotonic}
If a correct validator $p_i$ enters a view $V_1 \geq 1$ at some time, and enters a view $V_2 \geq 1$ at a later time, then $V_2 > V_1$.
\end{claim}
\begin{proof}
Validator $p_i$ updates $\viewi$ exactly when it enters a view, setting $\viewi$ to the entered view (line~\ref{ln:core:on_enter_set_view}).
Views are entered at three places only. The first is line~\ref{ln:core:init_enter_view1}, upon starting, which happens once and precedes every other entry. The second is line~\ref{ln:vc:upon_enter_enter}, which enters a view $V + 1$ under the condition $V \geq \viewi$ (line~\ref{ln:vc:upon_enter_next_view}). The third is line~\ref{ln:vc:upon_catch_up_enter}, which enters a view $V$ under the condition $V > \viewi$ (line~\ref{ln:vc:upon_catch_up}).
The last two thus enter a view greater than $\viewi$. Recovery preserves $\viewi$ and does not repeat initialization.
Hence, $V_2 > V_1$.
\qed
\end{proof}

Next, we show that a correct validator commits to a single vote of each kind per view, and that its lock, once set, is the one its $\VOTETWO$ message carries.

\begin{claim} \label{clm:coda:unique_votes}
Let $p_i$ be a correct validator, and let $V \geq 1$ be a view.
Then:
\begin{compactenum}
    \item Validator $p_i$ writes $\voteone[V]$, $\votetwo[V]$, and $\locki[V]$ only at times at which $\viewi = V$, and at most once each. \label{item:coda-write-once}
    \item Validator $p_i$ broadcasts at most one $\langle \VOTEONE, V, \cdot \rangle$ and at most one $\langle \VOTETWO, V, \cdot, \cdot \rangle$ message, each at a time at which $\viewi = V$. \label{item:coda-one-message}
    \item If $p_i$ broadcasts $\langle \VOTETWO, V, x, G \rangle$, then $x$ is the value of $\locki[V]$ and $G$ is the value of $\guardi[V]$ at the time of the broadcast, and $\locki[V] = x$ holds from that time onward. \label{item:coda-vote2-carries-lock}
\end{compactenum}
\end{claim}
\begin{proof}
We first prove part~(\ref{item:coda-write-once}).
The variable $\voteone[V]$ is written at line~\ref{ln:core:on_enter_dummy_vote1}, within the handler that enters view $V = 1$ and hence after $\viewi$ is set to $1$ (line~\ref{ln:core:on_enter_set_view}), and at line~\ref{ln:core:on_prepare_broadcast_vote1}, which acts on view $\viewi$.
The two sites are mutually exclusive: the second belongs to a rule triggered by a $\PREPARE$ message from $\leader(\viewi)$ (line~\ref{ln:core:on_prepare_header}), whereas view $1$ has no leader (line~\ref{ln:core:on_enter_dummy_view}).
The first site is reached only once, as view $1$ is entered only once (\Cref{clm:coda:monotonic}), and the second is guarded by $\voteone[\viewi] = \NONE$ (line~\ref{ln:core:on_prepare_if_no_vote1}).
The variable $\votetwo[V]$ is written at line~\ref{ln:core:on_vote1_quorum_set_vote2} and at line~\ref{ln:core:on_timer_set_vote2}. Both sites act on view $\viewi$, and both are guarded by $\votetwo[\viewi] = \NONE$ (lines~\ref{ln:core:on_vote1_quorum_if_no_vote2} and~\ref{ln:core:on_timer_if_no_vote2}).
The variable $\locki[V]$ is written at line~\ref{ln:core:on_vote1_quorum_set_lock} alone, within the same handler and under the same guard.
It remains to observe that these variables are never reset: the only mechanism that discards them is the memory rule (line~\ref{line:persistant_memory}), which discards the variables of a view once it falls below $\viewi - 1$. By \Cref{clm:coda:monotonic}, such a view is never entered again, and, as shown above, the variables of a view other than $\viewi$ are never written.

Part~(\ref{item:coda-one-message}) follows: each of these write sites broadcasts the message it has just written, while $\viewi = V$ by part~(\ref{item:coda-write-once}) (lines~\ref{ln:core:on_enter_dummy_vote1}, \ref{ln:core:on_prepare_broadcast_vote1}, \ref{ln:core:on_vote1_quorum_set_vote2}, and~\ref{ln:core:on_timer_set_vote2}), and no other line of \coda broadcasts a $\VOTEONE$ or a $\VOTETWO$ message.
Part~(\ref{item:coda-vote2-carries-lock}) follows as well. Both sites that create a $\VOTETWO$ message for view $V$ fill it with the current value of $\guardi[V]$. The site at line~\ref{ln:core:on_vote1_quorum_set_vote2} fills in the current value of $\locki[V]$, and the site at line~\ref{ln:core:on_timer_set_vote2} fills in $\NONE$, which is the current value of $\locki[V]$ there: $\locki[V]$ is written only in the handler that also writes $\votetwo[V]$ (line~\ref{ln:core:on_vote1_quorum_set_lock}), and $\votetwo[V] = \NONE$ at the timeout site (line~\ref{ln:core:on_timer_if_no_vote2}). By part~(\ref{item:coda-write-once}), $\locki[V]$ is never written again.
\qed
\end{proof}

The next two claims trace the values that correct validators vote for, and the locks that they accept, back to their origin.

\begin{claim} \label{clm:coda:vote1_guard}
If a correct validator $p_i$ broadcasts $\langle \VOTEONE, V, x \rangle$ for a view $V > 1$, then, for every proposer $p_j$, validator $p_i$ has received $f+1$ messages $\langle \ECHO, V - 1, j, x[j] \rangle$ from distinct validators.
\end{claim}
\begin{proof}
As established in the proof of \Cref{clm:coda:unique_votes}, the only site that creates a $\VOTEONE$ message for a view $V > 1$ is line~\ref{ln:core:on_prepare_broadcast_vote1}, and its rule is triggered only by a message $\langle \PREPARE, \viewi, x \rangle$ that is justified (line~\ref{ln:core:on_prepare_wait_guard}), where $\viewi = V$.
A \metablock{} $x$ is justified only if, for every proposer $p_j$, the validator has received $f+1$ messages $\langle \ECHO, V - 1, j, x[j] \rangle$ from distinct validators (line~\ref{ln:core:fn_justified_return}).
\qed
\end{proof}

\begin{claim} \label{clm:coda:accepted_lock}
If a correct validator accepts a message $\langle \VOTETWO, V, x, \cdot \rangle$ with $x \neq \NONE$, then some correct validator has broadcast $\langle \VOTEONE, V, x \rangle$.
\end{claim}
\begin{proof}
Such a message is accepted only once its recipient has received $f+1$ messages $\langle \VOTEONE, V, x \rangle$ from distinct validators (line~\ref{ln:vc:rule_accept_vote2_lock}), at most $f$ of which are faulty.
\qed
\end{proof}

We now introduce the first decision of the execution.
If no correct validator ever decides, agreement holds vacuously. We therefore assume throughout the remainder of the agreement proof that some correct validator does decide.

We say that a correct validator \emph{decides in view $V$} if it triggers $\mathsf{decide}$ on the strength of $2f+1$ $\langle \VOTETWO, V, \cdot, \cdot \rangle$ messages (line~\ref{ln:core:on_vote2_quorum_header}). Its current view is then $V$ or $V+1$.

\begin{definition} [View $\vdec$]
Let $\vdec$ denote the smallest view in which a correct validator decides.
\end{definition}

Before naming the \metablock{} that view $\vdec$ decides, we check that there is only one.

\begin{claim} \label{clm:coda:unique_in_vdec}
If a correct validator decides a \metablock{} $x_1$ in view $\vdec$, and a correct validator decides a \metablock{} $x_2$ in view $\vdec$, then $x_1 = x_2$.
\end{claim}
\begin{proof}
Each of the two decisions requires the receipt of $2f+1$ $\langle \VOTETWO, \vdec, \cdot, \cdot \rangle$ messages carrying the decided \metablock, from distinct senders (line~\ref{ln:core:on_vote2_quorum_header}).
As $n = 3f+1$, the two sets of senders share at least $2(2f+1) - n = f+1$ validators, at least one of which is correct.
By \Cref{clm:coda:unique_votes}(\ref{item:coda-one-message}), that validator broadcasts a single $\VOTETWO$ message for view $\vdec$, which implies $x_1 = x_2$.
\qed
\end{proof}

We may therefore speak of \emph{the} \metablock{} that view $\vdec$ decides.

\begin{definition} [\Metablock{} $\valdec$]
Let $\valdec$ denote the \metablock{} decided by a correct validator in view $\vdec$.
\end{definition}

The following claim is the counting heart of the proof.
A decision in view $\vdec$ leaves behind a residue of correct $\VOTETWO$ messages so large that no quorum of $\VOTETWO$ messages for view $\vdec$ can avoid it, however it is assembled.

\begin{claim} \label{clm:coda:residue}
Let $D$ denote the set of correct validators that broadcast $\langle \VOTETWO, \vdec, \valdec, \cdot \rangle$.
Then:
\begin{compactitem}
    \item $|D| \geq f+1$.
    \item Every set of $2f+1$ $\langle \VOTETWO, \vdec, \cdot, \cdot \rangle$ messages from distinct senders contains a message broadcast by a member of $D$.
\end{compactitem}
\end{claim}
\begin{proof}
Let $b \leq f$ denote the number of validators that are faulty in the execution under consideration.
The decision on $\valdec$ in view $\vdec$ requires the receipt of $2f+1$ $\langle \VOTETWO, \vdec, \valdec, \cdot \rangle$ messages from distinct senders (line~\ref{ln:core:on_vote2_quorum_header}), at most $b$ of which are faulty. The rest were broadcast by correct validators, which therefore belong to $D$.
Thus, $|D| \geq 2f + 1 - b \geq f + 1$.

Now, consider any set $S$ of $2f+1$ $\langle \VOTETWO, \vdec, \cdot, \cdot \rangle$ messages from distinct senders.
At most $f + b \leq 2f$ validators lie outside $D$: the $b$ faulty ones, and the at most $(n - b) - |D| \leq (3f + 1 - b) - (2f + 1 - b) = f$ correct ones.
As $S$ has $2f+1$ distinct senders, at least one of them belongs to $D$. By \Cref{clm:coda:unique_votes}(\ref{item:coda-one-message}), that validator broadcasts a single $\VOTETWO$ message for view $\vdec$, namely $\langle \VOTETWO, \vdec, \valdec, \cdot \rangle$, which is thus the message of $S$ that it sent.
\qed
\end{proof}

We are ready for the central claim of the agreement proof: once view $\vdec$ has decided $\valdec$, no correct validator ever echoes, or votes for, anything else.

\begin{claim} \label{clm:coda:messages_after_vdec}
Let $m$ be a message broadcast by a correct validator $p_i$.
Then:
\begin{compactenum}
    \item If $m = \langle \ECHO, V, j, v \rangle$ with $V \geq \vdec$, then $v = \valdec[j]$. \label{item:coda-echo-vdec}
    \item If $m = \langle \VOTEONE, V, x \rangle$ with $V > \vdec$, then $x = \valdec$. \label{item:coda-vote1-vdec}
\end{compactenum}
Moreover, $\guardi[V] = \valdec$ holds at every time at which $\viewi = V > \vdec$.
\end{claim}
\begin{proof}
We proceed by induction on the order of the events of interest at correct validators, namely the broadcasts covered by the claim and the entries into views greater than $\vdec$.
Fix such an event, and suppose that the claim holds for every event that precedes it.
We establish the claim through three subclaims.

\begin{subclaim} \label{clm:coda:guard_after_vdec}
If a correct validator $p_i$ enters a view $W > \vdec$, then $\guardi[W] = \valdec$ from that moment onward.
\end{subclaim}
\begin{subproof}[Proof of \Cref{clm:coda:guard_after_vdec}]
As $W > \vdec \geq 1$, validator $p_i$ enters view $W$ at line~\ref{ln:vc:upon_enter_enter} or at line~\ref{ln:vc:upon_catch_up_enter}, and sets its guard for $W$ to a \metablock{} $x$ (line~\ref{ln:vc:upon_enter_set_guard} or~\ref{ln:vc:upon_catch_up_set_guard}).
In the former case, $W = V + 1$ for a view $V \geq \vdec$, and, for every proposer $p_j$, a set of $2f+1$ messages $\langle \ECHO, V, j, x[j] \rangle$ from distinct senders exists (line~\ref{ln:vc:upon_enter_next_view}), at least $f+1$ of them broadcast by correct validators.
In the latter case, for every proposer $p_j$, a set of $f+1$ messages $\langle \ECHO, W, j, x[j] \rangle$ from distinct senders exists (line~\ref{ln:vc:upon_catch_up}), at least one of them broadcast by a correct validator, and $W \geq \vdec$.
Either way, for every proposer $p_j$, some correct validator broadcast an echo for $x[j]$, for a view at least $\vdec$, before the entry. The induction hypothesis (part~\ref{item:coda-echo-vdec}) yields $x[j] = \valdec[j]$.
As this holds for every proposer, $x = \valdec$, and the entry sets $\guardi[W] \gets \valdec$.
No other line writes $\guardi[W]$: guards are written upon entering a view alone, and view $W$ is entered once (\Cref{clm:coda:monotonic}), which concludes the proof.
\end{subproof}

\begin{subclaim} \label{clm:coda:vote1_after_vdec}
If $m = \langle \VOTEONE, V, x \rangle$ with $V > \vdec$, then $x = \valdec$.
\end{subclaim}
\begin{subproof}[Proof of \Cref{clm:coda:vote1_after_vdec}]
As $V > \vdec \geq 1$, \Cref{clm:coda:vote1_guard} applies: for every proposer $p_j$, validator $p_i$ has received $f+1$ messages $\langle \ECHO, V - 1, j, x[j] \rangle$ from distinct validators, at least one of which was broadcast by a correct validator.
As $V - 1 \geq \vdec$, the induction hypothesis (part~\ref{item:coda-echo-vdec}) gives $x[j] = \valdec[j]$ for every proposer $p_j$, that is, $x = \valdec$.
\end{subproof}

\begin{subclaim} \label{clm:coda:echo_after_vdec}
If $m = \langle \ECHO, V, j, v \rangle$ with $V \geq \vdec$, then $v = \valdec[j]$.
\end{subclaim}
\begin{subproof}[Proof of \Cref{clm:coda:echo_after_vdec}]
Echo messages are broadcast at line~\ref{ln:vc:fn_echo_broadcast} alone, so $p_i$ broadcasts $m$ upon an invocation of $\mathsf{echo}(V, p_j, v)$, of which \coda has two sources.

\begin{compactitem}
\item The \emph{amplification rule} (line~\ref{ln:vc:upon_echo_amplification}) requires a set of $f+1$ messages $\langle \ECHO, V, j, v \rangle$ from distinct senders, at least one of which was broadcast by a correct validator, and thus earlier. The induction hypothesis (part~\ref{item:coda-echo-vdec}) gives $v = \valdec[j]$.

\item The \emph{view-change rule} (line~\ref{ln:vc:on_vote2_quorum_compute_approved}) is triggered while $\viewi = V$ by a set $S$ of $2f+1$ accepted $\langle \VOTETWO, V, \cdot, \cdot \rangle$ messages, from distinct senders, whose non-$\NONE$ locks are all identical.
The rule echoes for every proposer at once. We fix the proposer $p_j$ of $m$ and distinguish two cases.
\begin{compactitem}
    \item \emph{Case $V = \vdec$.}
    By \Cref{clm:coda:residue}, the set $S$ contains a message $\langle \VOTETWO, \vdec, \valdec, \cdot \rangle$ broadcast by a correct validator.
    Hence, $S$ contains a non-$\NONE$ lock, so the rule takes its first branch (line~\ref{ln:vc:on_vote2_quorum_set_approved_lock}) and echoes $\ell[j]$ for the lock $\ell$ it exhibits. As all non-$\NONE$ locks in $S$ are identical, $\ell = \valdec$, and thus $v = \valdec[j]$.

    \item \emph{Case $V > \vdec$.}
    At least $f+1$ messages of $S$ were broadcast by correct validators. Each was broadcast at a time at which its sender's current view was $V$ (\Cref{clm:coda:unique_votes}(\ref{item:coda-one-message})), and thus carries the guard of its sender at that time, whose entry for $p_j$ is $\valdec[j]$ (\Cref{clm:coda:unique_votes}(\ref{item:coda-vote2-carries-lock}) and the induction hypothesis, applied to that sender's entry into view $V$).
    We distinguish the branch the rule takes.
    \begin{compactitem}
        \item The first branch (line~\ref{ln:vc:on_vote2_quorum_set_approved_lock}) echoes $\ell[j]$ for a non-$\NONE$ lock $\ell$ of an accepted message of $S$.
        By \Cref{clm:coda:accepted_lock}, some correct validator broadcast $\langle \VOTEONE, V, \ell \rangle$, so $\ell = \valdec$ by the induction hypothesis (part~\ref{item:coda-vote1-vdec}), and thus $v = \valdec[j]$.

        \item The second branch (line~\ref{ln:vc:on_vote2_quorum_guard_if}) is reached otherwise.
        Its condition asks for a value that is the entry for $p_j$ in the guards of $f+1$ messages of $S$, and $\valdec[j]$ is such a value: as established above, at least $f+1$ messages of $S$ carry a guard whose entry for $p_j$ is $\valdec[j]$. It is the only such value, as two values would require $2f+2 > |S|$ messages; hence $v = \valdec[j]$.
        The condition therefore holds, and the last branch, which echoes $\ALLSAFE$, is not taken (line~\ref{ln:vc:on_vote2_quorum_allsafe_echo}).\subqedhere
    \end{compactitem}
\end{compactitem}
\end{compactitem}
\end{subproof}

The three subclaims cover every event of interest, which completes the induction.
\qed
\end{proof}

Agreement now follows.

\begin{lemma} [Agreement] \label{lem:coda:agreement}
\coda (\Cref{alg:fallback}) satisfies agreement.
\end{lemma}
\begin{proof}
As $\vdec$ is the smallest view in which a correct validator decides, it suffices to prove that no correct validator decides a \metablock{} other than $\valdec$ in a view $V \geq \vdec$.
For $V = \vdec$, this is \Cref{clm:coda:unique_in_vdec}.

Hence, let $V > \vdec$, and suppose that a correct validator decides a \metablock{} $y \neq \NONE$ in view $V$.
The decision requires the receipt of $2f+1$ messages $\langle \VOTETWO, V, y, \cdot \rangle$ from distinct senders (line~\ref{ln:core:on_vote2_quorum_header}), at least one of which was broadcast by a correct validator $p_w$.
By \Cref{clm:coda:unique_votes}(\ref{item:coda-vote2-carries-lock}), $\mathit{lock}_w[V] = y \neq \NONE$, which is set only upon the receipt of $2f+1$ messages $\langle \VOTEONE, V, y \rangle$ from distinct senders (lines~\ref{ln:core:on_vote1_quorum_header} and~\ref{ln:core:on_vote1_quorum_set_lock}), at least one of which was broadcast by a correct validator.
By \Cref{clm:coda:messages_after_vdec}(\ref{item:coda-vote1-vdec}), $y = \valdec$, which concludes the proof.
\qed
\end{proof}

\paragraph{Termination.}
Termination assumes that every correct validator starts \coda and none stops taking part, and its proofs also assume that every correct validator starts by GST (the participation convention of \Cref{sec:coda:model}).
We first establish that correct validators keep advancing through views, and then that the first post-GST view with a correct leader decides.
Where the proofs below speak of the $2f+1$ correct validators, they mean an arbitrary but fixed set of $2f+1$ correct validators, in case there are more.
The following two claims tie echoes and view entries to one another.

\begin{claim} \label{clm:coda:echo_implies_entry}
If a correct validator broadcasts $\langle \ECHO, V, \pid, v \rangle$ for a view $V \geq 1$, then some correct validator has previously entered view $V$.
\end{claim}
\begin{proof}
We first observe that a correct validator broadcasts an $\ECHO$ message for view $V$ only while its current view is at least $V$: the view-change rule (line~\ref{ln:vc:on_vote2_quorum_compute_approved}) echoes for view $\viewi$, and the amplification rule (line~\ref{ln:vc:upon_echo_amplification}) echoes only for views $\viewi - 1$ and $\viewi$.
As $\viewi$ is set to a view only upon entering it (line~\ref{ln:core:on_enter_set_view}), the sender of such a message has entered a view at least $V$ beforehand.

Let $\mathcal{E}^{\star}$ denote the first entry by a correct validator into a view at least $V$, and let $Y \geq V$ be the entered view. It suffices to show that $Y = V$.
If $Y = 1$, then $\mathcal{E}^{\star}$ is an initial entry (line~\ref{ln:core:init_enter_view1}), and $Y = V$ follows from $1 = Y \geq V \geq 1$.
Otherwise, $Y > 1$, so $\mathcal{E}^{\star}$ occurs at line~\ref{ln:vc:upon_enter_enter} or at line~\ref{ln:vc:upon_catch_up_enter}, the only two places where a view other than view $1$ is entered.
It does not occur at the latter: the catch-up rule requires $f+1$ $\langle \ECHO, Y, \cdot, \cdot \rangle$ messages from distinct senders (line~\ref{ln:vc:upon_catch_up}), one of which was broadcast by a correct validator, and thus before $\mathcal{E}^{\star}$, while that validator's view was at least $Y \geq V$. That validator thus entered a view at least $V$ before $\mathcal{E}^{\star}$, contradicting the minimality of $\mathcal{E}^{\star}$.
Hence, $\mathcal{E}^{\star}$ occurs at line~\ref{ln:vc:upon_enter_enter}, which requires $2f+1$ $\langle \ECHO, Y - 1, \cdot, \cdot \rangle$ messages from distinct senders (line~\ref{ln:vc:upon_enter_next_view}), at least $f+1$ of them broadcast by correct validators before $\mathcal{E}^{\star}$, while their views were at least $Y - 1$.
By the minimality of $\mathcal{E}^{\star}$, $Y - 1 < V$, and, as $Y \geq V$, $Y = V$.
\qed
\end{proof}

\begin{claim} \label{clm:coda:entry_downward_closed}
If a correct validator enters a view $W > 1$, then the first such entry occurs at line~\ref{ln:vc:upon_enter_enter}, and some correct validator has previously entered view $W - 1$.
\end{claim}
\begin{proof}
Let $\mathcal{E}$ denote the first entry by a correct validator into view $W$. As $W > 1$, it occurs at line~\ref{ln:vc:upon_enter_enter} or at line~\ref{ln:vc:upon_catch_up_enter}.
It does not occur at the latter: the catch-up rule requires $f+1$ $\langle \ECHO, W, \cdot, \cdot \rangle$ messages from distinct senders (line~\ref{ln:vc:upon_catch_up}), one of which was broadcast by a correct validator, and, by \Cref{clm:coda:echo_implies_entry}, some correct validator entered view $W$ before that broadcast, hence before $\mathcal{E}$, a contradiction.
Therefore, $\mathcal{E}$ occurs at line~\ref{ln:vc:upon_enter_enter}, which requires a set of $2f+1$ $\langle \ECHO, W - 1, \cdot, \cdot \rangle$ messages from distinct senders (line~\ref{ln:vc:upon_enter_next_view}), at least $f+1$ of which were broadcast by correct validators before $\mathcal{E}$. \Cref{clm:coda:echo_implies_entry} then places an entry into view $W - 1$ before those broadcasts.
\qed
\end{proof}

Next, we record which messages survive on an unreliable network.
Recall that a broadcast is abandoned only upon a view entry, and only for messages outside the two-view window that the entry retains (line~\ref{ln:rule_abandon}).

\begin{claim} \label{clm:coda:forever_alive}
Let $V \geq 1$ be a view such that no correct validator ever enters a view greater than $V$.
Then every $\VOTEONE$, $\VOTETWO$ and $\ECHO$ message for view $V$ or view $V - 1$ broadcast by a correct validator is forever-alive, and is therefore eventually received by every correct validator.
\end{claim}
\begin{proof}
Let $m$ be such a message, for a view $W \in \{ V - 1, V \}$, broadcast by a correct validator $p_i$.
Its broadcast is initiated by hypothesis, so it suffices to show that it is never abandoned.
The only mechanism that abandons broadcasts is the rule at line~\ref{ln:rule_abandon}, which fires upon a view entry: entering a view $W'$ abandons the ongoing broadcasts of the messages of views smaller than $W' - 1$.
Consider the moment at which $p_i$ broadcasts $m$.
Its current view is then $W$ or $W + 1$: a $\VOTEONE$ or $\VOTETWO$ message for view $W$ is broadcast while $\viewi = W$ (\Cref{clm:coda:unique_votes}(\ref{item:coda-one-message})), and an $\ECHO$ message for view $W$ is broadcast by the view-change rule while $\viewi = W$ (line~\ref{ln:vc:on_vote2_quorum_compute_approved}), or by the amplification rule while $W \in \{ \viewi - 1, \viewi \}$ (line~\ref{ln:vc:upon_echo_amplification}).
As entered views only grow (\Cref{clm:coda:monotonic}), every later entry of $p_i$ is into a view $W' \geq W + 1$, and such an entry abandons the broadcast of $m$ exactly when $W$ is smaller than $W' - 1$.
An entry into view $W + 1$ therefore spares $m$, and only an entry into a view $W' \geq W + 2$ abandons it.
Since $W \geq V - 1$, such a view satisfies $W' \geq V + 1$, and no correct validator ever enters a view greater than $V$ by assumption.
Hence, $m$ is forever-alive, and its retransmissions eventually reach every correct validator, also after recovery or the start of participation.
\qed
\end{proof}

The next claim is the engine of progress: whenever correct validators would otherwise be stranded in a view, they all reach the same verdict about it and move on.

\begin{claim} \label{clm:coda:all_enter}
For every view $V \geq 1$, some correct validator eventually enters view $V$.
\end{claim}
\begin{proof}
Suppose, by contradiction, that some view is never entered by any correct validator.
By \Cref{clm:coda:entry_downward_closed}, the set of views entered by correct validators is closed downward, and every correct validator enters view $1$ upon starting (line~\ref{ln:core:init_enter_view1}). Let $\vmax$ be the greatest view that some correct validator enters.
No correct validator ever enters a view greater than $\vmax$, so \Cref{clm:coda:forever_alive} applies to views $\vmax$ and $\vmax - 1$ throughout.

\proofstep{Step 1: Every correct validator eventually enters view $\vmax$.}
If $\vmax = 1$, this holds by initialization.
Otherwise, by \Cref{clm:coda:entry_downward_closed}, the first entry into $\vmax$ occurs at line~\ref{ln:vc:upon_enter_enter}, and thus rests on a set of $2f+1$ messages $\langle \ECHO, \vmax - 1, \pid, x[\pid] \rangle$ for every proposer $p_\pid$ and some \metablock{} $x$, at least $f+1$ of them, for every proposer, from correct validators.
These are forever-alive, so every correct validator eventually receives $f+1$ of them for every proposer $p_\pid$ after its last recovery.
A correct validator whose view is then at most $\vmax - 2$ enters view $\vmax - 1$ on the spot, through the catch-up rule (line~\ref{ln:vc:upon_catch_up}, with $V = \vmax - 1 > \viewi$).
Hence, every correct validator $p_i$ is eventually in view $\vmax - 1$ or view $\vmax$, so $\vmax - 1 \in \{ \viewi - 1, \viewi \}$ and the amplification rule fires (line~\ref{ln:vc:upon_echo_amplification}): $p_i$ broadcasts $\langle \ECHO, \vmax - 1, \pid, x[\pid] \rangle$ for every proposer $p_\pid$, unless it has done so already.
Hence all $2f+1$ correct validators broadcast these echoes, which are forever-alive, so every correct validator eventually holds $2f+1$ of them for every proposer.
A correct validator whose view is then $\vmax - 1$ enters view $\vmax$ at line~\ref{ln:vc:upon_enter_enter} (with $V = \vmax - 1 \geq \viewi$), and one whose view is $\vmax$ is there already.

\proofstep{Step 2: Every correct validator broadcasts a single $\VOTETWO$ message for view $\vmax$, and every correct validator eventually accepts all $2f+1$ of these messages.}
Every correct validator is in view $\vmax$ forever, and its clock advances at the rate of real time from GST onward, so the timer of line~\ref{ln:core:on_timer_header} eventually reaches $5\Delta$ and it broadcasts a $\VOTETWO$ message for $\vmax$ (line~\ref{ln:core:on_timer_set_vote2}), unless it has broadcast one earlier (line~\ref{ln:core:on_vote1_quorum_set_vote2}). By \Cref{clm:coda:unique_votes}(\ref{item:coda-one-message}) this message is unique, and it is forever-alive.
Let $S^{\star}$ denote the set of these $2f+1$ messages.
Every member of $S^{\star}$ is eventually accepted by every correct validator: those carrying $\NONE$ upon reception (line~\ref{ln:vc:rule_accept_vote2_bot}), and one carrying a lock $x \neq \NONE$ once $f+1$ $\langle \VOTEONE, \vmax, x \rangle$ messages arrive (line~\ref{ln:vc:rule_accept_vote2_lock}), which they do: its sender locked $x$ upon receiving $2f+1$ such messages (line~\ref{ln:core:on_vote1_quorum_set_lock}), at least $f+1$ of them broadcast by correct validators and hence forever-alive.
So every correct validator eventually holds $S^{\star}$, with every member accepted.

\proofstep{Step 3: All correct validators echo a common value for every proposer.}
The set $S^{\star}$ satisfies the condition of the view-change rule (line~\ref{ln:vc:on_vote2_quorum_compute_approved}): it consists of $2f+1$ accepted $\langle \VOTETWO, \vmax, \cdot, \cdot \rangle$ messages, and all of its non-$\NONE$ locks are identical.
The latter holds because these locks are the locks of correct validators, and no two correct validators lock different \metablock{}s in the same view.
Indeed, two locks on different \metablocks{} would require two sets of $2f+1$ $\VOTEONE$ messages for view $\vmax$, from distinct senders each, sharing $2(2f+1) - n = f+1$ senders, at least one of them correct and hence broadcasting two $\VOTEONE$ messages for $\vmax$, contradicting \Cref{clm:coda:unique_votes}(\ref{item:coda-one-message}).
Moreover, the gate of line~\ref{ln:rule_view_change_gate} is eventually open at every correct validator, as every correct validator is in view $\vmax$ forever and its clock advances after GST.
The rule fires for every set that satisfies its condition (line~\ref{ln:core:notation_every_witness}), so it fires for $S^{\star}$ at every correct validator, and $S^{\star}$ is the same set at every correct validator.
The values it echoes are determined by $S^{\star}$ alone.
If $S^{\star}$ contains a lock $\ell$, every correct validator echoes $\ell[\pid]$ for every proposer $p_\pid$.
Otherwise, for every proposer $p_\pid$, at most one value is the entry for $p_\pid$ in the guards of $f+1$ messages of $S^{\star}$, as two such values would require $2f+2 > |S^{\star}|$ messages, and every correct validator echoes that value, or $\ALLSAFE$ if there is none.

\proofstep{Contradiction.}
These echoes are forever-alive, so every correct validator eventually holds $2f+1$ of them for every proposer and enters view $\vmax + 1$ (line~\ref{ln:vc:upon_enter_enter}), contradicting the maximality of $\vmax$.
\qed
\end{proof}

For every view $V \geq 1$, let $\tau(V)$ denote the earliest time at which a correct validator enters view $V$. It is well defined by \Cref{clm:coda:all_enter}, and $\tau(V) < \tau(V+1)$ by \Cref{clm:coda:entry_downward_closed}.
We call a view $V$ a \emph{post-GST} view if and only if $\tau(V) \geq \mathrm{GST}$.
The gate of line~\ref{ln:rule_view_change_gate} now buys us a floor on how quickly post-GST views can pass.

\begin{claim} \label{clm:coda:view_floor}
For every post-GST view $V \geq 1$, no correct validator broadcasts an $\ECHO$ message for view $V$ before time $\tau(V) + 5\Delta$. Hence $\tau(V + 1) \geq \tau(V) + 5\Delta$.
\end{claim}
\begin{proof}
By \Cref{clm:coda:entry_downward_closed}, the entry into view $V + 1$ at time $\tau(V+1)$ occurs at line~\ref{ln:vc:upon_enter_enter}, and thus rests on $2f+1$ $\langle \ECHO, V, \cdot, \cdot \rangle$ messages from distinct senders, at least one of which was broadcast by a correct validator before $\tau(V+1)$.
The first $\ECHO$ message for view $V$ broadcast by a correct validator is not produced by the amplification rule (line~\ref{ln:vc:upon_echo_amplification}), which presupposes an earlier such message. It is produced by the view-change rule, which its sender evaluates for view $V$ only once $5\Delta$ has elapsed since it entered $V$ (line~\ref{ln:rule_view_change_gate}).
That entry occurs at or after $\tau(V) \geq \mathrm{GST}$, and clocks exhibit no drift from GST onward, so no correct validator broadcasts an $\ECHO$ message for view $V$ before $\tau(V) + 5\Delta$, and the entry at $\tau(V+1)$, which follows such a message, is not earlier.
\qed
\end{proof}

\begin{claim} \label{clm:coda:prompt_entry}
For every post-GST view $V \geq 1$, every correct validator is in view $V$ at time $\tau(V) + 2\Delta$.
\end{claim}
\begin{proof}
By \Cref{clm:coda:view_floor} and repeated applications of \Cref{clm:coda:entry_downward_closed}, no correct validator enters a view greater than $V$ before $\tau(V + 1) \geq \tau(V) + 5\Delta$. Since views never decrease (\Cref{clm:coda:monotonic}), a validator that has entered view $V$ by time $\tau(V) + 2\Delta$, and that cannot have entered a view above $V$ by then, is in view $V$ at that time.
It therefore suffices to show that every correct validator enters view $V$ by time $\tau(V) + 2\Delta$.

If $V = 1$, every correct validator starts by GST under the participation convention of \Cref{sec:coda:model} and enters view $1$ upon starting, so $\tau(1) \leq \mathrm{GST}$. As view $1$ is post-GST, $\tau(1) = \mathrm{GST}$, and every correct validator enters view $1$ at that time.
Hence, let $V > 1$.
By \Cref{clm:coda:entry_downward_closed}, the entry at time $\tau(V)$, being the first into view $V$, occurs at line~\ref{ln:vc:upon_enter_enter}, and thus rests, for every proposer $p_\pid$, on $2f+1$ messages $\langle \ECHO, V - 1, \pid, x[\pid] \rangle$ from distinct senders for some \metablock{} $x$ (line~\ref{ln:vc:upon_enter_next_view}), at least $f+1$ of them broadcast by correct validators before $\tau(V)$.
Each of those broadcast by a correct validator is abandoned only upon its sender entering a view greater than $V$, which does not happen before $\tau(V) + 5\Delta$. These $f+1$ echoes per proposer are therefore alive throughout $[\tau(V), \tau(V) + 2\Delta]$ and, as $\mathrm{GST} \leq \tau(V)$, received by every correct validator by time $\tau(V) + \Delta$.
A correct validator whose view is then at most $V - 2$ enters view $V - 1$ on the spot, through the catch-up rule (line~\ref{ln:vc:upon_catch_up}).
Hence, at time $\tau(V) + \Delta$, every correct validator is in view $V - 1$ or view $V$, so $V - 1 \in \{ \viewi - 1, \viewi \}$ and the amplification rule fires (line~\ref{ln:vc:upon_echo_amplification}): every correct validator broadcasts $\langle \ECHO, V - 1, \pid, x[\pid] \rangle$ for every proposer $p_\pid$ by time $\tau(V) + \Delta$, unless it has done so already.
These $2f+1$ echoes are alive until $\tau(V) + 5\Delta$ for the same reason, hence received by every correct validator by time $\tau(V) + 2\Delta$, at which point every correct validator holds $2f+1$ echoes for every proposer and has entered view $V$.
\qed
\end{proof}

We are ready to prove termination.

\begin{lemma} [Termination] \label{lem:coda:termination}
Assume that every correct validator starts \coda by GST and that no correct validator stops taking part.
\coda{} (\Cref{alg:fallback}) satisfies termination.
Moreover, every correct validator decides by time $\tau(V^{\star}) + 5\Delta$, where $V^{\star}$ is the first post-GST view greater than $1$ whose leader is correct.
\end{lemma}
\begin{proof}
The view $V^{\star}$ exists: every correct validator performs finitely many steps by any given time, so only finitely many views are entered before GST and post-GST views exist (\Cref{clm:coda:all_enter}). Moreover, the leader rotates round-robin and at least $2f+1$ validators are correct.
Write $T = \tau(V^{\star})$ and $p_{\ell} = \leader(V^{\star})$.
Throughout, we use \Cref{clm:coda:view_floor}, which gives $\tau(V^{\star} + 1) \geq T + 5\Delta$ and $\tau(V^{\star} + 2) \geq T + 10\Delta$, so that no correct validator is in a view greater than $V^{\star}$ before time $T + 5\Delta$, nor in a view greater than $V^{\star} + 1$ before time $T + 10\Delta$.

\proofstep{Step 1: By time $T + 3\Delta$, every correct validator holds a justified proposal $\langle \PREPARE, V^{\star}, x \rangle$ from $p_{\ell}$.}
By \Cref{clm:coda:prompt_entry}, every correct validator, $p_{\ell}$ included, is in view $V^{\star}$ at time $T + 2\Delta$.
Upon entering $V^{\star}$, the correct $p_{\ell}$ broadcasts $\langle \PREPARE, V^{\star}, x \rangle$ for $x = \mathit{guard}_{\ell}[V^{\star}]$ (line~\ref{ln:core:on_enter_if_leader}), the \metablock{} with which it entered the view. This broadcast is abandoned only upon $p_{\ell}$ entering view $V^{\star} + 2$, impossible before $T + 10\Delta$, so every correct validator receives the proposal by time $T + 3\Delta$.
The entry of $p_{\ell}$ into $V^{\star}$ is not a catch-up entry (line~\ref{ln:vc:upon_catch_up_enter}): that would require an $\ECHO$ message for view $V^{\star}$ from a correct validator, none of which exists before $T + 5\Delta$ (\Cref{clm:coda:view_floor}), whereas $p_{\ell}$ enters by $T + 2\Delta$.
Hence, $p_{\ell}$ entered $V^{\star}$ with guard $x$ on the strength of $2f+1$ messages $\langle \ECHO, V^{\star} - 1, \pid, x[\pid] \rangle$ for every proposer $p_\pid$, at least $f+1$ of them from correct validators. These are alive until their senders enter view $V^{\star} + 1$, hence received by every correct validator by time $T + 3\Delta$, at which point $\justified(x)$ holds everywhere (line~\ref{ln:core:fn_justified_return}).

\proofstep{Step 2: By time $T + 4\Delta$, every correct validator has broadcast $\langle \VOTETWO, V^{\star}, x, \cdot \rangle$.}
Consider a correct validator $p_i$ at time $T + 3\Delta$, holding the justified proposal of Step 1 and still in view $V^{\star}$, as no correct validator enters $V^{\star} + 1$ before $T + 5\Delta$.
If $\voteone[V^{\star}] = \NONE$, the rule at line~\ref{ln:core:on_prepare_header} fires and $p_i$ broadcasts $\langle \VOTEONE, V^{\star}, x \rangle$. Otherwise, $p_i$ has already broadcast a $\VOTEONE$ message for $V^{\star}$, and that message carries $x$ as well, since in view $V^{\star} > 1$ a correct validator votes only for a proposal received from $p_{\ell}$ (line~\ref{ln:core:on_prepare_header}), and the correct $p_{\ell}$ broadcasts a single one.
All $2f+1$ correct validators thus broadcast $\langle \VOTEONE, V^{\star}, x \rangle$ by time $T + 3\Delta$. These are alive until their senders enter view $V^{\star} + 2$, impossible before $T + 10\Delta$, so every correct validator receives them by time $T + 4\Delta$. At that time $p_i$ is still in view $V^{\star}$, again because no correct validator enters $V^{\star} + 1$ before $T + 5\Delta$, so the rule at line~\ref{ln:core:on_vote1_quorum_header}, which acts on the current view only, fires for $V^{\star}$. This is where the gate of line~\ref{ln:rule_view_change_gate} earns its length: a shorter gate would let a correct validator leave $V^{\star}$ before its $\VOTEONE$ quorum for $V^{\star}$ completes, and it would then never broadcast a $\VOTETWO$ message for $V^{\star}$.
At that moment, either $\votetwo[V^{\star}] = \NONE$, in which case $p_i$ locks $x$ and broadcasts $\langle \VOTETWO, V^{\star}, x, \cdot \rangle$ (lines~\ref{ln:core:on_vote1_quorum_set_lock} and~\ref{ln:core:on_vote1_quorum_set_vote2}), or $\votetwo[V^{\star}]$ was set earlier at line~\ref{ln:core:on_vote1_quorum_set_vote2}, the timeout of line~\ref{ln:core:on_timer_header} firing no earlier than $T + 5\Delta$. In the latter case, the lock is a \metablock{} backed by $2f+1$ $\VOTEONE$ messages for $V^{\star}$, at least one of them from a correct validator, and is therefore $x$.
Either way, $p_i$ has broadcast $\langle \VOTETWO, V^{\star}, x, \cdot \rangle$ by time $T + 4\Delta$.

\proofstep{Step 3: Every correct validator decides by time $T + 5\Delta$.}
By Step 2, all $2f+1$ correct validators broadcast $\langle \VOTETWO, V^{\star}, x, \cdot \rangle$ by time $T + 4\Delta$. These messages are alive until their senders enter view $V^{\star} + 2$, impossible before $T + 10\Delta$, so every correct validator receives all of them by time $T + 5\Delta$.
As $x \neq \NONE$, the rule at line~\ref{ln:core:on_vote2_quorum_header} then fires, provided $V^{\star} \in \{ \viewi - 1, \viewi \}$ at that moment: this holds, since no correct validator is in a view greater than $V^{\star} + 1$ before $T + 10\Delta$.
Hence, every correct validator decides by time $T + 5\Delta$.
\qed
\end{proof}

The participation convention, that every correct validator starts \coda by GST, costs no generality.
Suppose that every correct validator starts \coda and none stops taking part, and let $T_0 = \max(\mathrm{GST}, T_{\mathrm{start}})$, where $T_{\mathrm{start}}$ is the time at which the last correct validator starts.
From $T_0$ on, every correct validator runs \coda, message delays are bounded, clocks do not drift, and no correct validator crashes.
A late starter may miss broadcasts that others abandoned before it started (line~\ref{ln:rule_abandon}), but the proofs above rely on the delivery only of messages of correct validators that stay alive long enough after $T_0$ (\Cref{clm:coda:forever_alive,clm:coda:all_enter,clm:coda:view_floor,clm:coda:prompt_entry} and the proof of \Cref{lem:coda:termination}), and the model delivers these to late starters as well, also if they dropped an earlier copy before starting.
The termination proofs thus hold with $T_0$ in place of GST: every correct validator decides by time $\tau(V^{\star}) + 5\Delta$, where $V^{\star}$ is the first view greater than $1$ that no correct validator enters before $T_0$ and whose leader is correct.

\paragraph{Strong unanimity.}
We now turn to strong unanimity, which \coda{} owes to each proposer separately.
The proof rests on a single invariant, the per-proposer counterpart of \Cref{clm:coda:messages_after_vdec}: once the correct validators agree on a proposer's entry, that entry never departs from the agreed value, in any guard and in any message.

\begin{claim} \label{clm:coda:unanimity_invariant}
Let $p_\pid$ be a proposer and let $v \in \VALUE$ be a value such that $\mathit{input}_w[\pid] = v$ for every correct validator $p_w$ that inputs.
Let $m$ be a message broadcast by a correct validator $p_i$.
Then:
\begin{compactenum}
    \item If $m = \langle \ECHO, V, \pid, x \rangle$, then $x = v$. \label{item:coda-unanimity-echo}
    \item If $m = \langle \VOTEONE, V, x \rangle$, then $x[\pid] = v$. \label{item:coda-unanimity-vote1}
\end{compactenum}
Moreover, $\guardi[V][\pid] = v$ holds at every time at which $\viewi = V$.
\end{claim}
\begin{proof}
As in the proof of \Cref{clm:coda:messages_after_vdec}, we proceed by induction on the order of the events of interest at correct validators, namely the broadcasts covered by the claim and the entries into views.
Fix such an event, and suppose that the claim holds for every event that precedes it.
We establish the claim through three subclaims.

\begin{subclaim} \label{clm:coda:unanimity_guard}
If a correct validator $p_i$ enters a view $W$, then $\guardi[W][\pid] = v$ from that moment onward.
\end{subclaim}
\begin{subproof}[Proof of \Cref{clm:coda:unanimity_guard}]
If $W = 1$, the entry sets $\guardi[1]$ to $\propi$ (line~\ref{ln:core:init_set_proposal}), and $\propi[\pid] = v$ by assumption, as $p_i$ entered view~$1$ upon starting \coda, when it input (line~\ref{ln:core:init_enter_view1}).
If $W > 1$, then $p_i$ enters view $W$ at line~\ref{ln:vc:upon_enter_enter}, with $W = V + 1$ and a \metablock{} $x$ for which a set of $2f+1$ messages $\langle \ECHO, V, \pid, x[\pid] \rangle$ from distinct senders exists (line~\ref{ln:vc:upon_enter_next_view}), or at line~\ref{ln:vc:upon_catch_up_enter}, with a \metablock{} $x$ for which a set of $f+1$ messages $\langle \ECHO, W, \pid, x[\pid] \rangle$ from distinct senders exists (line~\ref{ln:vc:upon_catch_up}).
In both cases, at least one of the certifying messages was broadcast by a correct validator, so $x[\pid] = v$ by the induction hypothesis (part~\ref{item:coda-unanimity-echo}), and the entry sets $\guardi[W] \gets x$ (line~\ref{ln:vc:upon_enter_set_guard} or~\ref{ln:vc:upon_catch_up_set_guard}).
In both cases, $\guardi[W][\pid] = v$. Moreover, guards are written upon entering a view alone, and view $W$ is entered once (\Cref{clm:coda:monotonic}), so the value persists.
\end{subproof}

\begin{subclaim} \label{clm:coda:unanimity_vote1}
If $m = \langle \VOTEONE, V, x \rangle$, then $x[\pid] = v$.
\end{subclaim}
\begin{subproof}[Proof of \Cref{clm:coda:unanimity_vote1}]
If $V = 1$, then $m$ is created at line~\ref{ln:core:on_enter_dummy_vote1} and carries $x = \guardi[1]$, which is $\propi$ (line~\ref{ln:core:init_set_proposal}), so $x[\pid] = v$ by assumption, as $p_i$ is in view~$1$ and has thus input.
If $V > 1$, then, by \Cref{clm:coda:vote1_guard}, validator $p_i$ has received $f+1$ messages $\langle \ECHO, V - 1, \pid, x[\pid] \rangle$ from distinct validators, at least one of which was broadcast by a correct validator. The induction hypothesis (part~\ref{item:coda-unanimity-echo}) gives $x[\pid] = v$.
\end{subproof}

\begin{subclaim} \label{clm:coda:unanimity_echo}
If $m = \langle \ECHO, V, \pid, x \rangle$, then $x = v$.
\end{subclaim}
\begin{subproof}[Proof of \Cref{clm:coda:unanimity_echo}]
Message $m$ is broadcast at line~\ref{ln:vc:fn_echo_broadcast} upon an invocation of $\mathsf{echo}(V, p_\pid, x)$, of which \coda has two sources.
\begin{compactitem}
\item The \emph{amplification rule} (line~\ref{ln:vc:upon_echo_amplification}) requires a set of $f+1$ messages $\langle \ECHO, V, \pid, x \rangle$ from distinct senders, at least one of which was broadcast by a correct validator, and thus earlier. The induction hypothesis (part~\ref{item:coda-unanimity-echo}) gives $x = v$.

\item The \emph{view-change rule} (line~\ref{ln:vc:on_vote2_quorum_compute_approved}) is triggered while $\viewi = V$ by a set $S$ of $2f+1$ accepted $\langle \VOTETWO, V, \cdot, \cdot \rangle$ messages from distinct senders.
At most $f$ of these senders are faulty, so at least $f+1$ messages of $S$ were broadcast by correct validators. Each was broadcast at a time at which its sender's current view was $V$ (\Cref{clm:coda:unique_votes}(\ref{item:coda-one-message})), and thus carries the guard of its sender at that time, whose entry for $p_\pid$ is $v$ (\Cref{clm:coda:unique_votes}(\ref{item:coda-vote2-carries-lock}) and the induction hypothesis, applied to that sender's entry into view $V$).
We distinguish the branch the rule takes.
\begin{compactitem}
    \item The first branch (line~\ref{ln:vc:on_vote2_quorum_set_approved_lock}) echoes $\ell[\pid]$ for a non-$\NONE$ lock $\ell$ of an accepted message of $S$.
    By \Cref{clm:coda:accepted_lock}, some correct validator broadcast $\langle \VOTEONE, V, \ell \rangle$, so $\ell[\pid] = v$ by the induction hypothesis (part~\ref{item:coda-unanimity-vote1}), and thus $x = v$.

    \item The second branch (line~\ref{ln:vc:on_vote2_quorum_guard_if}) is reached otherwise.
    Its condition asks for a value that is the entry for $p_\pid$ in the guards of $f+1$ messages of $S$, and $v$ is such a value: as established above, at least $f+1$ messages of $S$ carry a guard whose entry for $p_\pid$ is $v$.
    It is the only such value, as two values would require $2f+2 > |S|$ messages; hence $x = v$.
    The condition therefore holds, and $\ALLSAFE$ is not echoed for $p_\pid$ (line~\ref{ln:vc:on_vote2_quorum_allsafe_echo}).\subqedhere
\end{compactitem}
\end{compactitem}
\end{subproof}

The three subclaims cover every event of interest, which completes the induction.
\qed
\end{proof}

We are finally ready to prove strong unanimity.

\begin{lemma} [Strong unanimity] \label{lem:coda:strong_unanimity}
\coda{} (\Cref{alg:fallback}) satisfies strong unanimity.
\end{lemma}
\begin{proof}
Suppose that all correct validators that input input the same value $v \in \VALUE$ for a proposer $p_\pid$, and suppose, by contradiction, that a correct validator decides a \metablock{} $x$ with $x[\pid] = \ALLSAFE$ in some view $V$.
The decision requires the receipt of $2f+1$ messages $\langle \VOTETWO, V, x, \cdot \rangle$ from distinct senders (line~\ref{ln:core:on_vote2_quorum_header}), at least one of which was broadcast by a correct validator $p_w$.
By \Cref{clm:coda:unique_votes}(\ref{item:coda-vote2-carries-lock}), $\mathit{lock}_w[V] = x \neq \NONE$, which is set only upon the receipt of $2f+1$ messages $\langle \VOTEONE, V, x \rangle$ from distinct senders (lines~\ref{ln:core:on_vote1_quorum_header} and~\ref{ln:core:on_vote1_quorum_set_lock}), at least one of which was broadcast by a correct validator.
By \Cref{clm:coda:unanimity_invariant}(\ref{item:coda-unanimity-vote1}), $x[\pid] = v$. As $v \in \VALUE$ and $\ALLSAFE \notin \VALUE$ (line~\ref{ln:core:const_allsafe}), this contradicts $x[\pid] = \ALLSAFE$.
\qed
\end{proof}

\paragraph{Obligation.}
Obligation is, like strong unanimity, owed to each proposer separately, and it is proved by the same induction, with provenance in place of uniqueness.
We call a value $x \in \VALUE \cup \{ \ALLSAFE \}$ \emph{admissible for a proposer $p_\pid$} if $x = \ALLSAFE$ or some correct validator $p_w$ input it for $p_\pid$, that is, $\mathit{input}_w[\pid] = x$.
The invariant below is the counterpart of \Cref{clm:coda:unanimity_invariant}: instead of pinning a proposer's entry to one value, it confines that entry to the values a correct validator stands behind.

\begin{claim} \label{clm:coda:obligation_invariant}
Let $p_\pid$ be a proposer, and let $m$ be a message broadcast by a correct validator $p_i$.
Then:
\begin{compactenum}
    \item If $m = \langle \ECHO, V, \pid, x \rangle$, then $x$ is admissible for $p_\pid$. \label{item:coda-obligation-echo}
    \item If $m = \langle \VOTEONE, V, x \rangle$, then $x[\pid]$ is admissible for $p_\pid$. \label{item:coda-obligation-vote1}
\end{compactenum}
Moreover, $\guardi[V][\pid]$ is admissible for $p_\pid$ at every time at which $\viewi = V$.
\end{claim}
\begin{proof}
We proceed by induction on the order of the events of interest at correct validators, exactly as in the proof of \Cref{clm:coda:unanimity_invariant}, and we establish the claim through the same three subclaims.

\begin{subclaim} \label{clm:coda:obligation_guard}
If a correct validator $p_i$ enters a view $W$, then $\guardi[W][\pid]$ is admissible for $p_\pid$ from that moment onward.
\end{subclaim}
\begin{subproof}[Proof of \Cref{clm:coda:obligation_guard}]
If $W = 1$, the entry sets $\guardi[1]$ to $\propi$ (line~\ref{ln:core:init_set_proposal}), and $\propi[\pid]$ is admissible by definition, $p_i$ being correct.
If $W > 1$, then $p_i$ enters view $W$ at line~\ref{ln:vc:upon_enter_enter} or at line~\ref{ln:vc:upon_catch_up_enter}, and sets $\guardi[W] \gets x$ (line~\ref{ln:vc:upon_enter_set_guard} or~\ref{ln:vc:upon_catch_up_set_guard}) for a \metablock{} $x$ backed, for the proposer $p_\pid$, by a set of $2f+1$ messages $\langle \ECHO, W - 1, \pid, x[\pid] \rangle$ (line~\ref{ln:vc:upon_enter_next_view}) or of $f+1$ messages $\langle \ECHO, W, \pid, x[\pid] \rangle$ (line~\ref{ln:vc:upon_catch_up}) from distinct senders. Either way, at least one of them was broadcast by a correct validator, so $x[\pid]$ is admissible by the induction hypothesis (part~\ref{item:coda-obligation-echo}).
Guards are written upon entering a view alone, and view $W$ is entered once (\Cref{clm:coda:monotonic}), so the value persists.
\end{subproof}

\begin{subclaim} \label{clm:coda:obligation_vote1}
If $m = \langle \VOTEONE, V, x \rangle$, then $x[\pid]$ is admissible for $p_\pid$.
\end{subclaim}
\begin{subproof}[Proof of \Cref{clm:coda:obligation_vote1}]
If $V = 1$, then $m$ is created at line~\ref{ln:core:on_enter_dummy_vote1} and carries $x = \guardi[1]$, which is $\propi$ (line~\ref{ln:core:init_set_proposal}), whose entry for $p_\pid$ is admissible.
If $V > 1$, then, by \Cref{clm:coda:vote1_guard}, validator $p_i$ has received $f+1$ messages $\langle \ECHO, V - 1, \pid, x[\pid] \rangle$ from distinct validators, at least one of which was broadcast by a correct validator. The induction hypothesis (part~\ref{item:coda-obligation-echo}) makes $x[\pid]$ admissible.
\end{subproof}

\begin{subclaim} \label{clm:coda:obligation_echo}
If $m = \langle \ECHO, V, \pid, x \rangle$, then $x$ is admissible for $p_\pid$.
\end{subclaim}
\begin{subproof}[Proof of \Cref{clm:coda:obligation_echo}]
Message $m$ is broadcast at line~\ref{ln:vc:fn_echo_broadcast} upon an invocation of $\mathsf{echo}(V, p_\pid, x)$, of which \coda has two sources.
\begin{compactitem}
\item The \emph{amplification rule} (line~\ref{ln:vc:upon_echo_amplification}) requires a set of $f+1$ messages $\langle \ECHO, V, \pid, x \rangle$ from distinct senders, at least one of which was broadcast by a correct validator, and thus earlier. The induction hypothesis (part~\ref{item:coda-obligation-echo}) applies.

\item The \emph{view-change rule} (line~\ref{ln:vc:on_vote2_quorum_compute_approved}) is triggered by a set $S$ of $2f+1$ accepted $\langle \VOTETWO, V, \cdot, \cdot \rangle$ messages from distinct senders.
Unlike in the proof of \Cref{clm:coda:unanimity_invariant}, all three of its branches are reachable here, and we inspect each.
\begin{compactitem}
    \item The first branch (line~\ref{ln:vc:on_vote2_quorum_set_approved_lock}) echoes $\ell[\pid]$ for a non-$\NONE$ lock $\ell$ of an accepted message of $S$.
    By \Cref{clm:coda:accepted_lock}, some correct validator broadcast $\langle \VOTEONE, V, \ell \rangle$, so $\ell[\pid]$ is admissible by the induction hypothesis (part~\ref{item:coda-obligation-vote1}).

    \item The second branch (line~\ref{ln:vc:on_vote2_quorum_guard_echo}) echoes a value $G$ that is the entry for $p_\pid$ in the guards of $f+1$ messages of $S$.
    These come from distinct senders, at most $f$ of which are faulty, so at least one of them was broadcast by a correct validator, at a time at which that sender's current view was $V$ (\Cref{clm:coda:unique_votes}(\ref{item:coda-one-message})), and it carries the sender's guard at that time (\Cref{clm:coda:unique_votes}(\ref{item:coda-vote2-carries-lock})).
    That guard's entry for $p_\pid$ is admissible by the induction hypothesis, applied to that sender's entry into view $V$. Hence, so is $G$.

    \item The third branch (line~\ref{ln:vc:on_vote2_quorum_allsafe_echo}) echoes $\ALLSAFE$, which is admissible by definition.\subqedhere
\end{compactitem}
\end{compactitem}
\end{subproof}

The three subclaims cover every event of interest, which completes the induction.
\qed
\end{proof}

We are now ready to prove the obligation property.

\begin{lemma} [Obligation] \label{lem:coda:obligation}
\coda{} (\Cref{alg:fallback}) satisfies obligation.
\end{lemma}
\begin{proof}
Suppose that a correct validator decides a \metablock{} $x$ in some view $V$, and let $p_\pid$ be a proposer with $x[\pid] \neq \ALLSAFE$. We must exhibit a correct validator that input $x[\pid]$ for $p_\pid$.
The decision requires the receipt of $2f+1$ messages $\langle \VOTETWO, V, x, \cdot \rangle$ from distinct senders (line~\ref{ln:core:on_vote2_quorum_header}), at least one of which was broadcast by a correct validator $p_w$.
By \Cref{clm:coda:unique_votes}(\ref{item:coda-vote2-carries-lock}), $\mathit{lock}_w[V] = x \neq \NONE$, which is set only upon the receipt of $2f+1$ messages $\langle \VOTEONE, V, x \rangle$ from distinct senders (lines~\ref{ln:core:on_vote1_quorum_header} and~\ref{ln:core:on_vote1_quorum_set_lock}), at least one of which was broadcast by a correct validator.
By \Cref{clm:coda:obligation_invariant}(\ref{item:coda-obligation-vote1}), $x[\pid]$ is admissible. As $x[\pid] \neq \ALLSAFE$, some correct validator input $x[\pid]$ for $p_\pid$, which is precisely obligation.
\qed
\end{proof}

\subsection{Complexity}

We begin by bounding the echoes that a correct validator broadcasts for a single view.

\begin{claim} \label{clm:coda:echoes_per_view}
For every view $V \geq 1$ and every proposer $p_\pid$, a correct validator broadcasts at most $|P_\pid| + 1$ distinct $\langle \ECHO, V, \pid, \cdot \rangle$ messages, and it initiates the broadcast of each of them at most once.
\end{claim}
\begin{proof}
By \Cref{clm:coda:obligation_invariant}(\ref{item:coda-obligation-echo}), every $\langle \ECHO, V, \pid, v \rangle$ message broadcast by a correct validator carries a value $v$ that is admissible for $p_\pid$, that is, $v \in P_\pid \cup \{ \ALLSAFE \}$.
This leaves at most $|P_\pid| + 1$ distinct messages.

Each of them is broadcast at most once.
Echoes are broadcast only by the function $\mathsf{echo}$, which records every message it broadcasts in $\echosone[V]$ (line~\ref{ln:vc:fn_echo_broadcast}) and does not broadcast a recorded message again (line~\ref{ln:vc:fn_echo_if}).
It remains to check that the record survives as long as it is needed.
A correct validator invokes $\mathsf{echo}(V, \cdot, \cdot)$ only while $V \in \{ \viewi - 1, \viewi \}$: the view-change rule echoes for view $\viewi$ (line~\ref{ln:vc:on_vote2_quorum_compute_approved}), and the amplification rule echoes for views $\viewi - 1$ and $\viewi$ (line~\ref{ln:vc:upon_echo_amplification}).
Throughout that period, the memory rule retains $\echosone[V]$ (line~\ref{line:persistant_memory}), which starts empty, and the period never resumes once $\viewi$ exceeds $V + 1$, as views never decrease (\Cref{clm:coda:monotonic}).
\qed
\end{proof}

\codamessages*
\begin{proof}
Fix a correct validator $p_i$ and a view $V$.
Validator $p_i$ broadcasts at most one $\langle \PREPARE, V, \cdot \rangle$ message, upon entering view $V$ as its leader (line~\ref{ln:core:on_enter_if_leader}), which happens at most once (\Cref{clm:coda:monotonic}).
It broadcasts at most one $\VOTEONE$ message and at most one $\VOTETWO$ message for view $V$ (\Cref{clm:coda:unique_votes}(\ref{item:coda-one-message})).
Finally, for every proposer $p_\pid$, it broadcasts at most $|P_\pid| + 1 \leq c + 1$ $\ECHO$ messages for view $V$ (\Cref{clm:coda:echoes_per_view}), for at most $(c + 1) k$ echoes in total.
Every broadcast unicasts its message to each of the $n$ validators, and at most $n$ validators are correct, which yields the bound.
\qed
\end{proof}

\codamemory*
\begin{proof}
We count memory in words, where a word holds a value, a view number, or a validator identifier, so that a \metablock{} occupies $k$ words.
By the memory rule (line~\ref{line:persistant_memory}), a correct validator retains the view-specific variables of at most two views, alongside $\viewi$, the time at which it entered its current view, and its decision flag (one word each).
For each retained view $V$, the variables $\guardi[V]$ and $\locki[V]$ are \metablocks{}, $\voteone[V]$ is a message carrying one \metablock{}, and $\votetwo[V]$ is a message carrying two, for $O(k)$ words in total, while $\echosone[V]$ holds at most $(c + 1) k$ messages of $O(1)$ words each (\Cref{clm:coda:echoes_per_view}).
The \coda{} messages whose broadcasts are alive belong to views $\viewi - 1$ and $\viewi$ (line~\ref{ln:rule_abandon}), and each of them is either recorded in these variables or a $\PREPARE$ message carrying $\guardi[V]$, so keeping track of them costs $O(c k)$ words as well.
In total, $O(c k)$ words.
\qed
\end{proof}

\subsection{Evaluating the View-Change Rule}\label{app:coda:count}

The view-change rule (line~\ref{ln:vc:on_vote2_quorum_compute_approved}) fires once for every set $S$ of $2f+1$ accepted $\VOTETWO$ messages that satisfies its condition, but a validator can avoid enumerating these sets, since the values that the rule echoes are determined by counts.
Fix a correct validator $p_i$ in view $V$. Let $M$ be the $\langle \VOTETWO, V, \cdot, \cdot \rangle$ messages that $p_i$ has accepted, at most one per sender (line~\ref{ln:vc:rule_one_vote2_per_sender}), $U \subseteq M$ those without a lock, and $L_\ell \subseteq M$ those with lock $\ell$. For a proposer $p_\pid$ and an entry $v \in \VALUE \cup \{ \ALLSAFE \}$, let $c_\pid(v)$ be the number of messages in $U$ whose guard has entry $v$ for $p_\pid$.
Once $5\Delta$ have elapsed since $p_i$ entered view $V$, some set $S \subseteq M$ satisfying the condition of the rule causes it to invoke $\mathsf{echo}(V, p_\pid, v)$ if and only if one of the following holds.
\begin{compactenum}
\item $v = \ell[\pid]$ for a \metablock{} $\ell$ with $|L_\ell| \geq 1$ and $|L_\ell| + |U| \geq 2f+1$. We can choose $2f+1$ messages from $L_\ell \cup U$, including at least one from $L_\ell$, to form a set with lock $\ell$, on which the rule echoes the entry of $\ell$ for every proposer.
\item $|U| \geq 2f+1$ and $c_\pid(v) \geq f+1$. Then $f+1$ messages of $U$ with entry $v$ for $p_\pid$, together with $f$ further messages of $U$, form a set without locks in which $f+1$ messages carry $v$, so the rule echoes $v$ (no other entry can also reach $f+1$ in a set of $2f+1$).
\item $v = \ALLSAFE$ and $\sum_{w \in \VALUE \cup \{ \ALLSAFE \}} \min(c_\pid(w), f) \geq 2f+1$. The rule echoes $\ALLSAFE$ for $p_\pid$ on a set without locks in which no entry for $p_\pid$ is carried by $f+1$ messages. Such a set contains at most $f$ messages with each entry, so it exists exactly when at least $2f+1$ messages remain after keeping, for each entry, at most $f$ of the messages of $U$ that have this entry for $p_\pid$.
\end{compactenum}
Conversely, a set satisfying the condition either has locks that are all the same \metablock{} $\ell$, so it lies in $L_\ell \cup U$ and condition~1 holds, or has no locks, so it lies in $U$ and condition~2 holds if it echoes an entry carried by $f+1$ of its messages, and condition~3 holds otherwise.

\end{document}